\documentclass{article}

\usepackage{PRIMEarxiv}

\usepackage[utf8]{inputenc} 
\usepackage[T1]{fontenc}    
\usepackage{hyperref}       
\usepackage{url}            
\usepackage{booktabs}       
\usepackage{amsfonts}       
\usepackage{nicefrac}       
\usepackage{microtype}      
\usepackage{lipsum}
\usepackage{fancyhdr}       
\usepackage{graphicx}       
\graphicspath{{media/}}     

\usepackage{complexity-mod} 
\mathchardef\mhyphen="2D 
\usepackage{arydshln}
\usepackage{xcolor}
\usepackage{graphicx}
\usepackage{enumitem}
\usepackage{bigstrut}
\usepackage{comment}
\usepackage{tikz}
\usepackage{adjustbox}
\usetikzlibrary{positioning} 
\usepackage{amssymb,amsmath,amsthm}
\theoremstyle{theorem} 
\newtheorem{theorem}{Theorem}
\newtheorem{corollary}[theorem]{Corollary}
\newtheorem{lemma}[theorem]{Lemma}
\newtheorem{proposition}[theorem]{Proposition}
\newtheorem{claim}[theorem]{Claim}

\theoremstyle{definition} 
\newtheorem{definition}[theorem]{Definition}
\newtheorem{example}[theorem]{Example}

\newcommand{\defeq}{\mathrel{\mathop:}=}

\newcommand{\db}{\mathbf{db}}

\newcommand{\block}{\mathbf{blk}}
\newcommand{\rep}{\mathbf{r}}
\newcommand{\sep}{\mathbf{s}}
\newcommand{\tep}{\mathbf{t}}

\newcommand{\cchase}[3]{\mathsf{chase}({#1},{#2},{#3})}

\newcommand{\Cchase}[5]{\mathsf{chase}^{(\mathrm{#5})}({#1},{#2},{#4},{#3})}

\newcommand{\certainty}[2]{\mathsf{CERTAINTY}({#1},{#2})}
\newcommand{\cqa}[1]{\mathsf{CERTAINTY}({#1})}
\newcommand{\fk}{\mathcal{FK}}

\newcommand{\fkw}{\mathcal{FK}_{\mathsf{weak}}}
\newcommand{\pk}{\mathcal{PK}} 
\newcommand{\formula}[1]{\left({#1}\right)}
\newcommand{\foreignkey}[3]{{#1}[{#2}]\rightarrow{#3}}
\newcommand{\keyeq}{\sim}
\newcommand{\fkequiv}[1]{\stackrel{\mbox{}_{#1}}{\equiv}}

\newcommand{\fkmodels}[1]{\stackrel{\mbox{}_{#1}}{\models}}
\newcommand{\nfkmodels}[1]{\stackrel{\mbox{}_{#1}}{\not\models}}

\newcommand{\sd}{\oplus}
\newcommand{\sub}{\subseteq}
\newcommand{\problem}[1]{{\mathsf{#1}}}
\newcommand{\attacks}[1]{\stackrel{#1}{\rightsquigarrow}}
\newcommand{\nattacks}[1]{\stackrel{#1}{\not\rightsquigarrow}}
\newcommand{\sequencevars}[1]{{\mathbf{vars}}({#1})}
\newcommand{\keyvars}[1]{\mathsf{key}({#1})}
\newcommand{\keyconst}[1]{\mathsf{keyconst}({#1})}

\newcommand{\adom}[1]{\mathsf{adom}({#1})}

\newcommand{\atomvars}[1]{\mathsf{vars}({#1})}
\newcommand{\queryvars}[1]{\mathsf{vars}({#1})}
\newcommand{\queryconst}[1]{\mathsf{const}(#1)}
\newcommand{\tuple}[1]{\langle{#1}\rangle}
\newcommand{\FD}[1]{{\mathcal{K}}({#1})}
\newcommand{\fd}[2]{{#1}\rightarrow{#2}}
\newcommand{\signature}[2]{[{#1},{#2}]}
\newcommand{\reducesTo}[1]{\leq_{m}^{#1}}

\newcommand{\card}[1]{|{#1}|}
\newcommand{\substitute}[3]{{#1}_{[{#2}\rightarrow{#3}]}}

\newcommand{\closer}[1]{\preceq_{#1}}
\newcommand{\closerneq}[1]{\prec_{#1}}

\newcommand{\capcloserneq}[1]{\prec^\cap_{#1}}

\newcommand{\theblock}[2]{{\mathsf{block}}({#1},{#2})}
\newcommand{\filler}{\underline{\hspace{1ex}}}
\newcommand{\fkclosure}[3]{{#2}_{#1}^{#3}}
\newcommand{\lclosure}[1]{{#1}^{\ast}}

\newcommand{\calG}{{\mathcal{G}}}

\newcommand{\calP}{{\mathcal{P}}}

\newcommand{\constgraph}[2]{\mathcal{G}_{#1}({#2})}
\newcommand{\keyclosure}[2]{{#1}^{+,{#2}}}

\newcommand{\restrict}[2]{{#1}{\restriction}_{#2}}
\newcommand{\outgoing}[2]{{#2}[{#1}\rightarrow]}
\newcommand{\incoming}[2]{{#2}[\rightarrow{#1}]}
\newcommand{\att}[1]{\mathsf{#1}}
\newcommand{\sjfbcq}{\mathsf{sjfBCQ}}

\newcommand{\obed}{\mathsf{o}}
\newcommand{\nobed}{\mathsf{d}}
\newcommand{\prop}{\mathsf{str}}
\newcommand{\nprop}{\mathsf{weak}}
\newcommand{\weak}{\stackrel{\nprop}{\rightarrow}}
\newcommand{\strong}{\stackrel{\prop}{\rightarrow}}
\newcommand{\dpd}{\nobed\strong\nobed}
\newcommand{\opo}{\obed\strong\obed}
\newcommand{\dpo}{\nobed\strong\obed}

\newcommand{\opd}{\obed\stackrel{\prop}{\rightarrow}\nobed}
\newcommand{\posclosure}[2]{{#1}_{#2}}
\newcommand{\poscomplement}[2]{{#1}^{\mathsf{co}}_{#2}}
\newcommand{\bival}[2]{\Theta^{#1}_{#2}}
\newenvironment{redtext}{\color{red}}{\ignorespacesafterend}
\newenvironment{bluetext}{\color{blue}}{\ignorespacesafterend}
\newenvironment{olivetext}{\color{olive}}{\ignorespacesafterend}
\newcommand{\jef}[1]{\begin{redtext}Jef: {#1}\enspace\end{redtext}}
\newcommand{\miika}[1]{\begin{bluetext}Miika: {#1}\enspace\end{bluetext}}

\title{A Dichotomy in Consistent Query Answering for Primary Keys and Unary Foreign Keys
}

\author{
  Miika Hannula \\
  University of Helsinki \\
  Helsinki\\
  Finland\\
  \texttt{miika.hannula@helsinki.fi} \\
   \And
  Jef Wijsen \\
  University of Mons \\
  Mons \\
  Belgium\\
  \texttt{jef.wijsen@umons.ac.be} \\
}

\begin{document}
\maketitle

\begin{abstract}
Since 2005, significant progress has been made in the problem of Consistent Query Answering (CQA) with respect to primary keys.
In this problem, the input is a database instance that may violate one or more primary key constraints.
A repair is defined as a maximal subinstance that satisfies all primary keys.
Given a Boolean query~$q$, the question then is whether $q$ holds true in every repair.

So far, theoretical research in this field has not addressed the combination of primary key and foreign key constraints, despite the importance of referential integrity in database systems.
This paper addresses the problem of CQA with respect to both primary keys and foreign keys.
In this setting, it is natural to adopt the notion of symmetric-difference repairs, because foreign keys can be repaired by inserting new tuples.

We consider the case where foreign keys are unary, and queries are conjunctive queries without self-joins.
In this setting, we characterize the boundary between those CQA problems that admit a consistent first-order rewriting, and those that do not.
\end{abstract}

\keywords{consistent query answering \and primary key \and foreign key \and conjunctive query}

\maketitle

\section{Introduction}

Consistent query answering (CQA) was introduced in~\cite{DBLP:conf/pods/ArenasBC99} as a principled semantics for answering queries on inconsistent databases.
A \emph{symmetric-difference repair} (or $\sd$-repair) of a database $\db$ is defined as a consistent database $\rep$ that $\subseteq$-minimizes the symmetric difference with $\db$.
Informally, a $\sd$-repair $\rep$ becomes inconsistent as soon as we insert into it more tuples of $\db$, or delete from it tuples not in $\db$.
Then, given a query $q(\vec{x})$, an answer $\vec{a}$ is called \emph{consistent} if $q(\vec{a})$ holds true in every repair. 
The problem is often studied for Boolean queries $q$, where the question is to determine whether $q$ holds true on every repair of a given input database.

CQA has been studied in depth in case that the only constraints are primary keys, one per relation.
In~\cite{DBLP:conf/pods/Wijsen10}, this problem was coined as $\cqa{q}$,  in which notation it is understood that every relation name in $q$ has a predefined primary key.
More than a decade of research has eventually resulted in the following complexity classification~\cite{DBLP:journals/mst/KoutrisW21}: for every self-join-free Boolean conjunctive query~$q$, the problem $\cqa{q}$ is either in $\FO$, $\L$-complete, or $\coNP$-complete.

Now that this classification has been settled, it is natural to ask what happens if we add foreign key constraints. Indeed, every relational database textbook is likely to introduce very soon the notion of referential integrity, i.e., foreign keys referencing primary keys. 
In view thereof, one may even wonder why referential integrity in CQA has so far received little theoretical research attention.
One plausible explanation is that $\sd$-repairs with respect to primary keys are easy to characterize: every repair has to delete, in every \emph{block}, all tuples but one, where a block is a maximal set of tuples of the same relation that agree on their primary key.
In contrast, $\sd$-repairs with respect to foreign keys can introduce new tuples, as illustrated next.
It will become apparent in later sections that having, as repair primitives, both tuple insertions and tuple deletions considerably complicates the theoretical treatment of CQA.

Consider the database of Fig.~\ref{fig:exampledb}, in which primary keys are underlined.
A tuple $(d,o)$ in the relation $\att{R}$ means that the document with DOI~$d$ was written  by the author with ORCiD~$o$.
The set of foreign keys is $\fk_{0}\defeq\{\foreignkey{\att{R}}{1}{\att{DOCS}}$,
$\foreignkey{\att{R}}{2}{\att{AUTHORS}}\}$.
In this paper, we assume that every foreign key is unary (i.e., consists of a single attribute) and that the referenced primary key is the leftmost attribute in the referenced table.

\begin{figure}[h]\centering
\begin{tabular}{ccc}
\begin{tabular}[t]{c|cc}
$\att{R}$ &
$\underline{\att{doi}}$ & $\underline{\att{orcid}}$\\\cline{2-3}
& d1 & o1\\
& d1 & o2\\
& d1 & o3
\end{tabular}
&
\begin{tabular}[t]{c|cll}
$\att{AUTHORS}$ &
$\underline{\att{orcid}}$ & $\att{first}$ & $\att{last}$\\\cline{2-4}
& o1 & \emph{Jeff} & \emph{Ullman}\\
& o1 & \emph{Jeffrey} & \emph{Ullman}\\
& o2 & \emph{Jonathan} & \emph{Ullman}
\end{tabular}
&
\begin{tabular}[t]{c|clc}
$\att{DOCS}$ 
& $\underline{\att{doi}}$ & $\att{title}$ & $\att{year}$\\\cline{2-4}
& d1 & \emph{Some pairs problems} & 2016
\end{tabular}
\end{tabular}
\caption{Inconsistent database.}\label{fig:exampledb}
\end{figure}

There is one foreign-key violation: the fact $\att{R}$(d1,o3) is \emph{dangling}, because o3 is not an existing ORCiD in the table $\att{AUTHORS}$.
There is also one primary-key violation, because there are two distinct tuples with ORCiD o1 in the table $\att{AUTHORS}$.
This database has an infinite number of $\sd$-repairs.
To repair the primary-key violation, we must delete either tuple with ORCiD~o1 in the table $\att{AUTHORS}$.
To repair the foreign-key violation, we can either delete the fact $\att{R}$(d1,o3), or insert a new fact $\att{AUTHORS}$(o3, \emph{fi}, \emph{la}), where \emph{fi} and \emph{la} can be chosen arbitrarily. 
Consider now the Boolean query:

\begin{center}
Does some paper of \emph{2016} have an author with first name \emph{Jeff}?
\end{center}
There is a repair in which the answer to this query is ``no,'' in which case we also say that ``no'' is the consistent answer.
In our setting, this Boolean query will be denoted by the following set of atoms:
$$
q_{0}=\{
\att{DOCS}(\underline{x}, t, \mbox{\emph{`2016'}}),
\att{R}(\underline{x, y}),
\att{AUTHORS}(\underline{y}, \mbox{\emph{`Jeff'}}, z)
\}.
$$
We note that $q_{0}$ satisfies the foreign keys in $\fk_{0}$ (when distinct variables are treated as distinct constants) and every relation name that occurs in $\fk_{0}$ also occurs in $q_{0}$, in which case we say that \emph{$\fk_{0}$ is about $q_{0}$}.

\emph{Data cleaning}~\cite{DBLP:books/acm/IlyasC19,DBLP:journals/vldb/GeertsMPS20} differs from CQA in that it tries to single out the single best repair before asking any query. 
We view CQA as complementary to data cleaning. 
In the preceding example, it may take some time (and manual effort) to find out what is the correct first name of the author with ORCiD~o1, and how the dangling fact in $\att{R}$ has to be cleaned.
An advantage of CQA is that we can immediately obtain some consistent query answers, which will hold true no matter of which repair will be chosen during the data cleaning process.

For every self-join-free Boolean conjunctive query $q$,
for every set of foreign keys that is about $q$,
we define $\certainty{q}{\fk}$ as the following problem:

\begin{center}
\framebox{
\begin{minipage}{0.75\columnwidth}
\begin{description}
\item[{\sc Problem} $\certainty{q}{\fk}$.]
\item[Input:] A database instance $\db$.
\item[Question:] Is $q$ true in every $\sd$-repair w.r.t.\ foreign keys in $\fk$ and primary keys?
\end{description}
\end{minipage}}
\end{center}
Obviously, if $\fk=\emptyset$, then $\certainty{q}{\fk}$ becomes the well studied problem $\cqa{q}$.

Of special interest is the case where $\certainty{q}{\fk}$ is in the complexity class~$\FO$, which is the class of problems that take a relational database instance as input and can be solved by a relational calculus query (a.k.a.\ \emph{consistent first-order rewriting} in the context of CQA). 
A major contribution of this paper can now be stated.
\begin{theorem}\label{thm:mainfo}
For every self-join-free Boolean conjunctive query~$q$,
for every set of unary foreign keys $\fk$ that is about~$q$,
it can be decided whether or not $\certainty{q}{\fk}$ is in $\FO$.
Furthermore, if $\certainty{q}{\fk}$ is in $\FO$, its consistent first-order rewriting can be effectively constructed.
\end{theorem}

We briefly discuss the remaining restrictions, leaving a more detailed discussion to Section~\ref{sec:discussion}.
The requirement that foreign keys be unary (i.e., consist of a single attribute) is met in our example, and is likely to be met in many real life situations where entities are identified by unique identifiers (like DOI, ORCiD\dots). Note that we allow composite primary keys, as in the relation $\att{R}$ in our example, but such composite primary keys cannot be referenced by a foreign key.
Nevertheless, some results in this paper are already proved for foreign keys that need not be unary.

The restriction that the set of foreign keys must be \emph{about the query} needs some care during query writing. For example, the question whether the author with ORCiD~o1 has published some paper in 2016, should be formulated as follows:
$$
q_{1}=\{
\att{DOCS}(\underline{x}, t, \mbox{\emph{`2016'}}),
\att{R}(\underline{x, \mbox{`o1'}}),
\att{AUTHORS}(\underline{\mbox{`o1'}}, u, z)
\}.
$$
The third atom  may look redundant in the latter query. However, $\fk_{0}$ is not about the shorter query $\{
\att{DOCS}(\underline{x}, t, \mbox{\emph{`2016'}})$,
$\att{R}(\underline{x, \mbox{`o1'}})\}$, 
in which $\att{R}(\underline{x, \mbox{`o1'}})$ is dangling with respect to $\foreignkey{R}{2}{\att{AUTHORS}}$.

The remainder of this paper is organized as follows.
Section~\ref{sec:related-work} discusses related work.
Section~\ref{sec:preliminaries} introduces preliminary notions and results from the literature that are used in our work. 
In Section~\ref{sec:bi}, we define a novel notion, called \emph{block-interference}, which plays a central role in a main theorem, given in Section~\ref{sec:mainthm}, which implies Theorem~\ref{thm:mainfo}.
Sections~\ref{sec:lhard} and~\ref{sec:nl} show that $\certainty{q}{\fk}$ is $\L$-hard or $\NL$-hard (and thus not in $\FO$) under some conditions.
Section~\ref{sec:fo} shows that if these conditions are not satisfied,
then $\certainty{q}{\fk}$ is in $\FO$.
In this way, our main theorem will be proved.
A side result in Section~\ref{sec:nl} is the existence of $\NL$-complete and $\P$-complete cases of  $\certainty{q}{\fk}$, which complexity classes did not pop up in earlier studies that were restricted to primary keys. 
We conclude this paper with a discussion in Section~\ref{sec:discussion}. 
All proofs and several helping lemmas have been moved to the appendix.
\section{Related Work}
\label{sec:related-work}

Consistent query answering (CQA) was initiated in a seminal paper
by Arenas, Bertossi, and Chomicki~\cite{DBLP:conf/pods/ArenasBC99}, in which the notions of $\sd$-repairs and consistent query answers were introduced. 
Recent overviews of two decades of research in CQA are \cite{DBLP:conf/pods/Bertossi19,DBLP:journals/sigmod/Wijsen19}.
From the latter overview, it is clear that different classes of constraints have been studied independently in CQA.
The current study is different in that it combines constraints from two classes: primary keys belong to the larger class of equality-generating dependencies (egd), and foreign keys belong to the larger class of tuple-generating dependencies (tgd). 
CQA has also been studied in the context of ontologies formulated in description logics; see~\cite{DBLP:journals/ki/Bienvenu20} for a recent overview.

The term $\cqa{q}$ was coined in 2010~\cite{DBLP:conf/pods/Wijsen10} to refer to CQA for Boolean queries $q$ on databases that violate primary keys, one per relation, which are fixed by $q$'s schema. 
A systematic study of its complexity for self-join-free conjunctive queries had started already in~2005~\cite{DBLP:conf/icdt/FuxmanM05}, and was eventually solved in two journal articles by Koutris and Wijsen~\cite{DBLP:journals/tods/KoutrisW17,DBLP:journals/mst/KoutrisW21}, as follows: for every self-join-free Boolean conjunctive query, $\cqa{q}$ is either in $\FO$, $\L$-complete, or $\coNP$-complete, and it is decidable, given $q$, which case applies.

A few extensions beyond this trichotomy result are known. 
The complexity of $\cqa{q}$ for self-join-free conjunctive queries with negated atoms was studied in~\cite{DBLP:conf/pods/KoutrisW18}. 
For self-join-free conjunctive queries with respect to multiple keys, it remains decidable whether or not $\cqa{q}$ is in $\FO$~\cite{DBLP:conf/pods/KoutrisW20}.
The complexity landscape of $\cqa{q}$ for path queries, a subclass of (not necessarily self-join-free) conjunctive queries, was settled in~\cite{DBLP:conf/pods/KoutrisOW21}.
For unions of conjunctive queries~$q$, Fontaine~\cite{DBLP:journals/tocl/Fontaine15} established interesting relationships between $\cqa{q}$  and Bulatov's dichotomy theorem
for conservative CSP~\cite{DBLP:journals/tocl/Bulatov11}.

The counting variant of $\cqa{q}$, denoted
$\sharp\cqa{q}$, asks to count the number of repairs that satisfy some Boolean query~$q$.
For self-join-free conjunctive queries, $\sharp\cqa{q}$ exhibits a dichotomy between  $\FP$ and $\sharp\P$-complete under polynomial-time Turing reductions~\cite{DBLP:journals/jcss/MaslowskiW13}. This dichotomy also holds for queries with self-joins if primary keys are singletons~\cite{DBLP:conf/icdt/MaslowskiW14}.
Calautti, Console, and Pieris present in~\cite{DBLP:conf/pods/CalauttiCP19} a complexity analysis of these counting problems under weaker reductions, in particular, under many-one logspace reductions.
The same authors have conducted an experimental evaluation of randomized approximation schemes for approximating the percentage of repairs that satisfy a given query~\cite{DBLP:conf/pods/CalauttiCP21}.
Other approaches to making CQA more meaningful and/or tractable include operational repairs~\cite{DBLP:conf/pods/CalauttiLP18} and preferred repairs~\cite{DBLP:journals/amai/StaworkoCM12,DBLP:journals/tcs/KimelfeldLP20}.

It is worthwhile to note that theoretical research in $\cqa{q}$ has stimulated implementations and experiments in prototype systems~\cite{DBLP:conf/vldb/FuxmanFM05,DBLP:conf/sigmod/FuxmanFM05,DBLP:journals/pvldb/KolaitisPT13,DBLP:conf/sat/DixitK19,DBLP:conf/cikm/KhalfiouiJLSW20}.

\section{Preliminaries}\label{sec:preliminaries}
For a positive integer $n$, we write $[n]$ for the set $\{1, \ldots ,n\}$.
We assume denumerable sets of \emph{variables} and \emph{constants}. 
A \emph{term} is a variable or a constant.
Every \emph{relation name} is associated with a signature, which is a pair $\signature{n}{k}$ of positive integers, where $n$ is the \emph{arity} and $k\in[n]$; the set $[k]$ is the \emph{primary key} of~$R$, and each $i\in[k]$ is called a \emph{primary-key position}.

From here on, we assume a fixed \emph{database schema} (i.e., a finite set of relation names with their associated signatures).

\subsection{CQA for Primary Keys}

We summarize notations and results used in CQA for primary keys.
The following definitions are borrowed and adapted from~\cite{DBLP:journals/tods/KoutrisW17}.

If $R$ is a relation name with signature $\signature{n}{k}$,
and $t_{1},\dots,t_{n}$ are terms, then $R(\underline{t_{1},\dots,t_{k}},t_{k+1},\dots,t_{n})$ is an \emph{$R$-atom} (or simply \emph{atom}).
If $F$ is an atom, then $\atomvars{F}$ denotes the set of variables that occur in~$F$, and $\keyvars{F}$ denotes the set of variables that occur in $F$ at some primary-key position.
An atom without variables is called a \emph{fact}.
Two facts $A,B$ are said to be \emph{key-equal}, denoted $A\keyeq B$, if they use the same relation name and agree on all primary-key positions.

A \emph{database (instance)} is a finite set $\db$ of facts.
From here on, $\db$ stands for a database instance.
A \emph{Boolean conjunctive query} is a finite set~$q$ of atoms.
We write $\queryvars{q}$ for the set of variables that occur in~$q$,
and $\queryconst{q}$ for the set of constants that occur in $q$. 
If $x\in\queryvars{q}$ and $c$ is a constant,
then $\substitute{q}{x}{c}$ is the query obtained from $q$ by replacing each occurrence of~$x$ with~$c$; this notation naturally extends to sequences with more than one variable and constant. 
A Boolean conjunctive query is \emph{self-join-free} if it does not contain two atoms with the same relation name.
We write $\sjfbcq$ for the class of all self-join-free Boolean conjunctive queries.

In contexts where a query $q$ in $\sjfbcq$ is understood, whenever we use a relation name~$R$ where an atom is expected, we mean the (unique) $R$-atom of~$q$.

A \emph{valuation} over a set $V$ of variables is a total mapping $\theta$ from~$V$ to the set of constants. A valuation is extended to map every constant to itself.
A valuation naturally extends to atoms and sets of atoms.
A Boolean conjunctive query $q$ is satisfied by $\db$, denoted $\db\models q$, if there is a valuation over $\queryvars{q}$ such that $\theta(q)\subseteq\db$.

A \emph{block} of $\db$ is a maximal subset of key-equal facts.
If $A$ is a fact in $\db$, then $\theblock{A}{\db}$ denotes the block of~$\db$ that contains $A$.
If $A=R(\underline{\vec{a}},\vec{b})$, then $\theblock{A}{\db}$ is also denoted by $R(\underline{\vec{a}},*)$, and a fact in this block is said to be \emph{of the form $R(\underline{\vec{a}},\filler)$.}

A \emph{repair of $\db$ with respect to primary keys} is a maximal subset of $\db$ containing no two distinct key-equal facts.
If $q$ is a Boolean conjunctive query, then $\cqa{q}$ is the problem that, given an input database instance $\db$, asks whether $q$ is satisfied by every repair of $\db$ with respect to primary keys.


Instead of saying that a repair must not contain two distinct key-equal facts, we can say that, for every relation name~$R$ of signature $\signature{n}{k}$, a repair must satisfy the following \emph{primary-key constraint}:
\begin{equation}\label{eq:pkconstraint}
\begin{array}{l}
\forall x_{1}\dotsm\forall x_{n}\forall y_{k+1}\dotsm\forall y_{n}\\
\formula{
\left(
\begin{array}{l}
R(\underline{x_{1},\dots,x_{k}},x_{k+1},\ldots,x_{n})\land \\
R(\underline{x_{1},\dots,x_{k}},y_{k+1},\ldots,y_{n})
\end{array}
\right)
\rightarrow\formula{\bigwedge_{i=k+1}^{n} x_{i}=y_{i}}}.
\end{array}
\end{equation}
In the technical treatment, it will often be convenient to use $\pk$ for the set that contains such a formula for every relation name in the database schema under consideration.

The complexity classification of $\cqa{q}$ uses the notion of \emph{attack graph}~\cite{DBLP:journals/tods/KoutrisW17} recalled next.
For a query $q$ in $\sjfbcq$, we write $\FD{q}$ for the set $\{\fd{\keyvars{F}}{\atomvars{F}}\mid F\in q\}$, which is a set of functional dependencies over $\queryvars{q}$.
For an atom $F\in q$, we define $\keyclosure{F}{q}\defeq\{x\in\atomvars{q}\mid\FD{q\setminus\{F\}}\models\fd{\keyvars{F}}{x}\}$.
Informally, $\keyclosure{F}{q}$ is the set of variables that are functionally dependent on $\keyvars{F}$ via the functional dependencies in $\FD{q\setminus\{F\}}$.
The \emph{attack graph of~$q$} is a directed graph whose vertices are the atoms of $q$; there is a directed edge from $F$ to $G$, called an \emph{attack} and denoted $F\attacks{q}G$, if $F\neq G$ and there exists a sequence of variables $x_{0},x_{1},\ldots,x_{n}$, all belonging to~$\queryvars{q}\setminus\keyclosure{F}{q}$, such that $x_{0}\in\atomvars{F}$, $x_{n}\in\atomvars{G}$,  and every two adjacent variables occur together in some atom of~$q$.
Moreover, $F$ is said to attack every variable in such a sequence. 
The following result obtains.

\begin{theorem}[\cite{DBLP:journals/tods/KoutrisW17}]\label{thm:koutrisw}
Let $q$ be a query in $\sjfbcq$.
If the attack graph of $q$ is acyclic, then the problem $\cqa{q}$ is in $\FO$; 
otherwise $\cqa{q}$ is $\L$-hard.
\end{theorem}

$\FO$ is used for the class of decision problems that take a database instance as input, and that can be solved by a closed first-order formula.

\subsection{Foreign keys}

Let $R$ be a relation name with signature $\signature{n}{k}$,
and $S$ a relation name with signature $\signature{m}{1}$.
Possibly $R=S$. 
A \emph{foreign key} is an expression $\foreignkey{R}{i}{S}$ such that $1\leq i \leq n$.
It is called \emph{weak} if $i\leq k$, and \emph{strong} otherwise.
We say that this foreign key is \emph{outgoing from $R$} and \emph{referencing $S$}.
We say that a fact $R(a_{1},\dots,a_{n})$ of $\db$ is \emph{dangling (in $\db$)} with respect to this foreign key  if 
$\db$ contains no $S$-fact $S(\underline{b_{1}},b_{2},\dots,b_{n})$ such that $a_{i}=b_{1}$.
A fact of $\db$ is dangling with respect to a set of foreign keys if it is dangling with respect to some foreign key of the set.
A set of foreign keys is \emph{satisfied} by $\db$ if $\db$ contains no dangling facts. 
Remark that foreign keys are unary by definition.

We write $\lclosure{\fk}$ for the set that contains every foreign key that is logically implied by $\fk$ (and that only uses relation names of the database schema under consideration), where logical implication has its standard definition.

The following notion of \emph{dependency graph} is borrowed and adapted from~\cite[Def.~3.7]{DBLP:journals/tcs/FaginKMP05}, where it was defined for general tgds.
The \emph{dependency graph of a set $\fk$} of foreign keys is a directed graph.
There is a vertex $(R,i)$ whenever $R$ is a relation name that occurs in $\fk$, say with signature~$\signature{n}{k}$, and $i\in[n]$.
Such a pair $(R,i)$ will be called a \emph{position}. 
More specifically, we say that $(R,i)$ is a \emph{primary-key position} if $i\in [k]$, and otherwise a \emph{non-primary-key position}.
Each foreign key
$\foreignkey{R}{i}{S}$ in $\fk$, where $S$ has signature $[m,1]$, induces a directed edge from $(R,i)$ to $(S,j)$, for every $j\in [m]$.
An edge from $(R,i)$ to $(S,j)$ is called \emph{special} if $j\neq 1$.
For a set of positions $P$, we define the \emph{closure $\posclosure{P}{\fk}$ of $P$ under $\fk$} as the set of all positions $(R,i)$ such that there is a path (possibly of length $0$) from some position in $P$ to $(R,i)$ in the dependency graph of $\fk$.
The \emph{complement of $\posclosure{P}{\fk}$} (with respect to all positions of the database schema under consideration), denoted $\poscomplement{P}{\fk}$, is the set of positions $(R,i)\notin \posclosure{P}{\fk}$. Note that if a relation name $R$ of arity~$n$ occurs in a query but not in $\fk$,  then $\poscomplement{P}{\fk}$ includes $\{(R,i)|i\in[n]\}$, even though the positions in the latter set are not vertices of the dependency graph.

\begin{example}
Let $\fk=\{\foreignkey{R}{1}{S}$, $\foreignkey{R}{3}{T}\}$, where $R$ has signature~$\signature{3}{2}$, and $S$ and $T$ both have signature $\signature{2}{1}$.
The foreign key $\foreignkey{R}{1}{S}$ is weak, and $\foreignkey{R}{3}{T}$ is strong.
The dependency graph of $\fk$ contains directed edges from $(R,1)$ to every position in $\{(S,1), (S,2)\}$, and directed edges from $(R,3)$ to every position in $\{(T,1), (T,2)\}$.
The edges ending in $(S,2)$ or $(T,2)$ are special.
\qed
\end{example}

The following definition of query containment under foreign keys is borrowed and adapted from~\cite{DBLP:journals/jcss/JohnsonK84}, where it was studied for general inclusion dependencies.
For Boolean queries, containment boils down to logical entailment.
Let $\fk$ be a set of foreign keys, and let $q$ and $q'$ be two Boolean queries. We say that \emph{$q$ entails $q'$ under $\fk$}, written $q \fkmodels{\fk} q'$, if for every database instance $\db$ that satisfies $\fk$, if $\db\models q$, then $\db\models q'$. 
We say that $q$ and $q'$ are
 \emph{equivalent under $\fk$}, written $q\fkequiv{\fk} q'$, if $q \fkmodels{\fk} q'$ and $q' \fkmodels{\fk} q$.
For example, if $R$ and $S$ have arity~$1$ and $\fk=\{\foreignkey{R}{1}{S}\}$, then $\{R(\underline{x})\}\fkequiv{\fk}\{R(\underline{x}), S(\underline{x})\}$.

Finally, we will restrict the sets $\fk$ of foreign keys that will be allowed for a query~$q$ in $\sjfbcq$.
We say that $\fk$ is \emph{about $q$} if every foreign key in~$\fk$ is satisfied by~$q$ (when distinct variables are treated as distinct constants) and, moreover, every relation name that occurs in $\fk$ also occurs in $q$. 

\subsection{CQA for Primary and Foreign Keys}

Symmetric-difference repairs were defined in~\cite{DBLP:conf/pods/ArenasBC99} as follows, for any set of integrity constraints.

We write $\sd$ for symmetric set difference.
Let $\db$ be a database instance.
Whenever $\rep$, $\sep$ are database instances,
we write $\rep\closer{\db}\sep$ if $\db\sd\rep\subseteq\db\sd\sep$.
If $\rep\closer{\db}\sep$, we also say that $\rep$ is \emph{$\sd$-closer} to $\db$ than~$\sep$.
It can be easily verified that $\closer{\db}$ is a partial order on the set of all database instances.
We write $\rep\closerneq{\db}\sep$ if $\rep\closer{\db}\sep$ and $\rep\neq\sep$.

Let $\fk$ be a set of foreign keys.
A \emph{$\sd$-repair} of $\db$ with respect to $\fk\cup\pk$\footnote{Recall that $\pk$ is the set of primary-key constraints, of the form~\eqref{eq:pkconstraint}, that can be derived from the relation names in $\db$.}(or \emph{repair} for short)   is a database instance $\rep$ such that:
\emph{(i)}~$\rep$ satisfies $\fk\cup\pk$, and
\emph{(ii)}~there is no database instance $\sep$ such that $\sep\closerneq{\db}\rep$ and $\sep$ satisfies $\fk\cup\pk$.
A \emph{subset-repair} is a $\sd$-repair $\rep$ satisfying $\rep\subseteq\db$, and a \emph{superset-repair} is a $\sd$-repair $\rep$ satisfying $\db\subseteq\rep$.

The next example shows that $\oplus$-repairs can be less intuitive and more diverse than subset-repairs or superset-repairs alone.

\begin{example}
Let $q=\{R(\underline{x},y), S(\underline{y},z), T(\underline{z})\}$
and $\fk=\{\foreignkey{R}{2}{S}$, $\foreignkey{S}{2}{T}\}$.
Let $\db=\{R(\underline{a},b), S(\underline{b},c)\}$.
Then the following are three $\sd$-repairs:
\begin{align*}
\rep_{1} & = \{\},\\
\rep_{2} & = \{R(\underline{a},b), S(\underline{b},1), T(\underline{1})\},\\
\rep_{3} & = \{R(\underline{a},b), S(\underline{b},c), T(\underline{c})\}.
\end{align*}
$\rep_{1}$ is a subset-repair, and $\rep_{3}$ a superset-repair.
It may be counter-intuitive that $\rep_{3}$ is not strictly $\sd$-closer to $\db$ than $\rep_{2}$.
Note however:
\begin{align*}
\db\sd\rep_{2} & = \{S(\underline{b},c), S(\underline{b},1), T(\underline{1})\},\\
\db\sd\rep_{3} & = \{T(\underline{c})\}.
\end{align*}
Since the latter two sets are not comparable by~$\subseteq$,
we have that $\rep_{2}$ and $\rep_{3}$ are not comparable by $\closer{\db}$. 
\qed
\end{example}

\tikzstyle{dbfact}=[->,shorten >=1pt,>=stealth,thick]
\tikzstyle{inventedfact}=[->,shorten >=1pt,>=stealth,thick, dashed]
\tikzstyle{const}=[circle,draw=black!50,fill=blue!10,thick,inner sep=0pt,minimum size=6mm]

Let $q$ be a query in $\sjfbcq$, and $\fk$ a set of foreign keys about $q$.
We write $\certainty{q}{\fk}$ for the decision problem that takes as input a database instance and asks whether $q$ is true in every $\sd$-repair with respect to $\fk\cup\pk$.
 
The following is relative to a fixed problem $\certainty{q}{\fk}$.
A \emph{consistent first-order rewriting} is a closed first-order formula~$\varphi$ such that a database instance is a ``yes''-instance of the problem $\certainty{q}{\fk}$ if and only if it satisfies $\varphi$. 
Clearly, the existence of a consistent first-order rewriting coincides with the problem being in the complexity class~$\FO$.

\section{Block-Interference}\label{sec:bi}

\emph{Block-interference} is a novel notion that plays a significant role in the complexity classification of $\certainty{q}{\fk}$.
Its definition is technical, but the following example should be helpful to convey the intuition.  

Let $q=\{N(\underline{x},c,y), O(\underline{y})\}$ with $\fk=\{\foreignkey{N}{3}{O}\}$, where $c$ is a constant.
Consider the following database instance, where the value $\Box$ in the last $N$-fact is yet unspecified.
$$
\db=
\begin{array}{ll}
\begin{array}[t]{c|ccc}
N & \underline{x} & c & y\\\cline{2-4}
  & b_{1} & c & 1\\
  & b_{1} & d & 2\\\cdashline{2-4}
  & b_{2} & c & 2\\
  & b_{2} & d & 3\\\cdashline{2-4}
  & b_{3} & c & 3\\
  & b_{3} & d & 4\\\cdashline{2-4}
  & \vdots & \vdots & \vdots\\\cdashline{2-4}
  & b_{n} & c & n\\
  & b_{n} & d & n+1\\\cdashline{2-4}
  & b_{n+1} & \Box & n+1\\\cdashline{2-4}
\end{array}
&
\begin{array}[t]{c|c}
O & \underline{y}\bigstrut\\\cline{2-2}
  & 1\\\cdashline{2-2}
\end{array}
\end{array}
$$
Note that all $N$-facts, except the first one, are dangling.
Our goal is to construct a $\sd$-repair $\rep$ that falsifies $q$.
Such a $\sd$-repair must obviously choose $N(\underline{b_{1}},d,2)$ in the first $N$-block, which implies that $O(\underline{2})$ must be inserted.
But then $N(\underline{b_{2}},c,2)$ is no longer dangling, and, as a consequence, $\rep$ must contain an $N$-fact from the second $N$-block.
In order to falsify~$q$, $\rep$ must choose $N(\underline{b_{2}},d,3)$ in the second block, which implies that $O(\underline{3})$ must be inserted.
By repeating the same reasoning, $\rep$ must be as follows:
$$
\rep=
\begin{array}{ll}
\begin{array}[t]{c|ccc}
N & \underline{x} & c & y\\\cline{2-4}
  & b_{1} & d & 2\\
  & b_{2} & d & 3\\
  & \vdots & \vdots & \vdots\\
  & b_{n} & d & n+1\\
  & b_{n+1} & \Box & n+1
\end{array}
&
\begin{array}[t]{c|c}
O & \underline{y}\bigstrut\\\cline{2-2}
  & 1\\
  & 2\\
  & 3\\
  & \vdots\\
  & n+1
\end{array}
\end{array}
$$
This is a falsifying $\sd$-repair if (and only if) $\Box\neq c$.
It is now correct to conclude that $\db$ is a ``yes''-instance of $\certainty{q}{\fk}$ if and only if $\Box=c$.
Note that for $\db'\defeq\db\setminus\{O(\underline{1})\}$, we have that the empty database instance is a $\sd$-repair of $\db'$, and hence $\db'$  is a ``no''-instance of $\certainty{q}{\fk}$.

Informally, in deciding whether or not $\db$ is a ``yes''-instance of $\certainty{q}{\fk}$, we had to start from the first $N$-block, then repeatedly move to the next $N$-block, and finally inspect the value of~$\Box$ in the last $N$-block.
It is now unsurprising that $\certainty{q}{\fk}$ is not in $\FO$ (as formally proved in Section~\ref{sec:nl}), because our movement from block to block goes well beyond the locality of first-order logic~\cite[Chapter~4]{DBLP:books/sp/Libkin04}. 
The notion of \emph{block-interference} will capture what is going on in this example.
Two more things are to notice:
\begin{itemize}
\item
The occurrence of the constant $c$ in $N(\underline{x},c,y)$ is important in the above example, because it is used to distinguish, within each $N$-block, between satisfying and falsifying $N$-facts.
Instead of a constant, we could have used two occurrences of a same variable, for example, $N(\underline{x},y,y)$ (and adapt $\db$ accordingly).
On the other hand, block-interference disappears if we replace $N(\underline{x},c,y)$ with $N(\underline{x},z,y)$ in $q$, where $z$ is a fresh variable occurring only once.
\item
Block-interference will also disappear if we replace $O(\underline{y})$ with $O(\underline{y},c)$ or $O(\underline{y},y)$ in the above example, because then the $O$-facts in $\rep\setminus\db$ can take the form $O(\underline{i},\bot)$ for some fresh constant~$\bot$ which cannot be used for making the query true.
On the other hand, if we replace $O(\underline{y})$ with $O(\underline{y},w)$ in $q$, where $w$ is a fresh variable occurring only once, then block-interference will remain.
\end{itemize}

We now proceed with formalizing block-interference in a number of steps. First, we introduce a concept called \emph{obedience} which, as we will see, plays a central role in block-interference.

\begin{definition}[Obedience]\label{def:obedience}
Let $q$ be a query in $\sjfbcq$, and $\fk$ a set of foreign keys about $q$. Let $R$ be a relation name of signature $\signature{n}{k}$, and let $P\subseteq\{(R,i)\mid i\in \{k+1, \ldots,n\}\}$ be a set of positions. Define $\fkclosure{P}{q}{\fk}$ as the smallest subset of $q$ such that
if the closure $\posclosure{P}{\fk}$ contains a position 
 $(S,j)$, then $\fkclosure{P}{q}{\fk}$ contains the $S$-atom of $q$.
We also  write $\fkclosure{R}{q}{\fk}$ as a shorthand for  $\fkclosure{P_R}{q}{\fk}$, where $P_R \defeq \{(R,i) \mid i\in \{k+1, \ldots, n\}\}$.


Let the $R$-atom of~$q$ be $F=R(\underline{\vec{s}},t_{k+1},\dots,t_{n})$,
and define $F_{P}\defeq R(\underline{\vec{s}},u_{k+1},\dots,u_{n})$ where for every $i\in\{k+1,\dots,n\}$, $u_{i}=t_{i}$ if $(R,i)\not\in P$, and $u_{i}$ is a fresh variable otherwise.
We say that the set~$P$ of positions is \emph{obedient (over $\fk$ and $q$)} if 
\begin{equation}\label{eq:obedience}
\formula{q \setminus \fkclosure{P}{q}{\fk}} \cup \{F_{P}\} \fkmodels{\fk} q, 
\end{equation}
where it is to be noted that the logical entailment in the other direction holds vacuously true (and therefore we also get $\fkequiv{\fk}$-equivalence).
Furthermore, we say that atom $F$ is \emph{obedient (over $\fk$ and $q$)} if the set of positions $\{(R,i)\mid i \in \{k+1, \ldots ,n\}\}$ is obedient (over $\fk$ and $q$). If $\fk$ and $q$ are clear from the context, we may simply say that a set of positions or an atom is obedient. 
A set of positions (or an atom) is called \emph{disobedient} if it is not obedient.
\end{definition}
 
\begin{example}\label{ex:no}
Consider again $q=\{N(\underline{x},c,y), O(\underline{y})\}$ with $\fk=\{\foreignkey{N}{3}{O}\}$.
We first argue that $P_{0}\defeq\{(N,2)\}$ is not obedient.
We have $\fkclosure{P_{0}}{q}{\fk}=\{N(\underline{x},c,y)\}$, because the dependency graph has an empty path from $(N,2)$ to itself, and no path from $(N,2)$ to $(O,1)$.
The left-hand expression in~\eqref{eq:obedience} then becomes $\{N(\underline{x},u_{2},y), O(\underline{y})\}$, which is not $\fkequiv{\fk}$-equivalent to $q$.

We next argue that $P_{1}\defeq\{(N,3)\}$ is obedient.
We have $\fkclosure{P_{1}}{q}{\fk}=q$, because the dependency graph has an empty path from $(N,3)$ to itself, and an edge from $(N,3)$ to $(O,1)$.
The left-hand expression in~\eqref{eq:obedience} becomes $\{N(\underline{x},c,u_{3})\}$. We have $\{N(\underline{x},c,u_{3})\}\fkequiv{\fk}\{N(\underline{x},c,u_{3})$, $O(\underline{u_{3}})\}$, and the latter query is obviously $\fkequiv{\fk}$-equivalent to~$q$. 

Note finally that the atom $O(\underline{y})$ is obviously obedient, because it has no non-primary-key positions.  
\qed
\end{example}

The concept of obedience can also be given a purely syntactic description, which will be useful in the technical treatment. The proof of the following theorem is given in Appendix~\ref{sect:apA}.
 
\begin{theorem}[Syntactic obedience]\label{thm:syntactic-obedience}
Let $q$ be a query in $\sjfbcq$, and $\fk$ a set of unary foreign keys about $q$. Let $P\subseteq\{(R,i)\mid i\in \{k+1, \dots, n\}\}$ for some relation name $R$ of signature $\signature{n}{k}$. Then, $P$ is obedient if and only if  all the following conditions hold true on the dependency graph of $\fk$:
\begin{enumerate}[label=(\Roman*)]
\item\label{it:strongweak} no position of $P$ belongs to a cycle;
\item\label{it:no-const}  
no constant occurs in~$q$ at a position of $\posclosure{P}{\fk}$;
\item\label{it:disjoint} 
no variable occurs in $q$ both at a position of $\posclosure{P}{\fk}$ and a position of $\poscomplement{P}{\fk}$; and
\item\label{it:distinct} 
no variable occurs in~$q$ at two distinct non-primary-key positions of $\posclosure{P}{\fk}$.
\end{enumerate}
\end{theorem}

Theorem~\ref{thm:syntactic-obedience} has the following immediate corollary, which implies that obedience can  be treated as a property of single positions. 
\begin{corollary}\label{cor:obedient-position}
Let $q$, $\fk$, and $P$ be as in the statement of Theorem~\ref{thm:syntactic-obedience}.
Then, $P$ is obedient over $\fk$ and $q$ if and only if $\{(R,i)\}$ is obedient over $\fk$ and $q$ for all $(R,i)\in P$.
\end{corollary}

Informally, Theorem~\ref{thm:syntactic-obedience} implies that if a set $P$ of positions is obedient, then $\posclosure{P}{\fk}$ is of the form depicted in Fig.~\ref{fig:obgraph}, where arrows represent foreign keys and primary-key positions are boxed (relation names are omitted). 
In particular, the figure shows the absence of cycles, constants, and variables that are repeated within a same atom.

\tikzstyle{smallposition}=[regular polygon,regular polygon sides=4,dotted, draw=black!70,fill=gray!10,thick,inner sep=0.1pt, minimum size=7mm]
\tikzstyle{smallprimposition}=[regular polygon,regular polygon sides=4,draw=black!70,fill=gray!10,thick,inner sep=0.1pt, minimum size=7mm]
\tikzstyle{smallobedientposition}=[regular polygon,regular polygon sides=4,dotted, draw=black!70,fill=green!10,thick,inner sep=0.1pt, minimum size=7mm]
\tikzstyle{smallobedientprimposition}=[regular polygon,regular polygon sides=4,draw=black!70,fill=green!10,thick,inner sep=0.1pt, minimum size=7mm]
\tikzstyle{smallobedienthackposition}=[regular polygon,regular polygon sides=4,dotted,draw=black!70,fill=green!10,thick,inner sep=-1.2pt, minimum size=7mm]
\tikzstyle{smallobedientprimhackposition}=[regular polygon,regular polygon sides=4,,draw=black!70,fill=green!10,thick,inner sep=-1.2pt, minimum size=7mm]
\tikzstyle{smallprimhackposition}=[regular polygon,regular polygon sides=4,draw=black!70,fill=gray!10,thick,inner sep=-.7pt, minimum size=7mm]
\tikzstyle{smalledge}=[->,shorten <=0pt, shorten >=1pt,>=stealth,thin]

\usetikzlibrary{shapes.geometric}
\usetikzlibrary{positioning,decorations.pathreplacing,quotes}
\newcommand{\shiftpoints}{4pt}
\newcommand{\shiftmanypoints}{16pt}

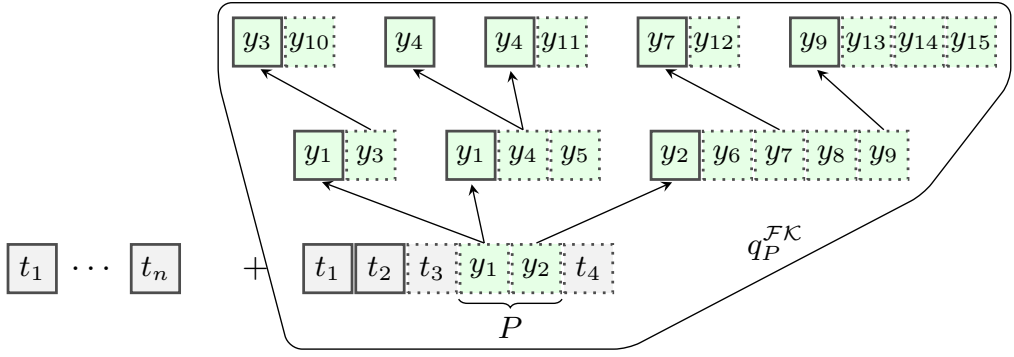
\begin{figure}
\begin{center}
\scalebox{1.3}{
\begin{tikzpicture}[    shifttl/.style={shift={(-\shiftpoints,\shiftpoints)}},
    shifttr/.style={shift={(\shiftpoints,\shiftpoints)}},
    shiftbl/.style={shift={(-\shiftpoints,-\shiftpoints)}},
    shiftbr/.style={shift={(\shiftpoints,-\shiftpoints)}},
    shiftBL/.style={shift={(-\shiftpoints,-\shiftmanypoints)}},  
    shiftBR/.style={shift={(\shiftpoints,-\shiftmanypoints)}},
]
\node (O1) [smallprimposition] {\footnotesize $t_1$};
\node (O2) [right =.1mm of O1, smallprimposition] {\footnotesize $t_2$};
\node (O3) [right =.1mm of O2, smallposition] {\footnotesize ${t_3}$};
\node (O4) [right =.1mm of O3, smallobedientposition] {\footnotesize ${y_1}$};
\node (O5) [right =.1mm of O4, smallobedientposition] {\footnotesize ${y_2}$};
\node (O6) [right =.1mm of O5, smallposition] {\footnotesize ${t_4}$};

\node (P1) [above left =9mm of O3, smallobedientprimposition] {\footnotesize ${y_1}$};
\node (P2) [right =.1mm of P1, smallobedientposition] {\footnotesize ${y_3}$};

\node (Q1) [right = 5mm of P2, smallobedientprimposition] {\footnotesize ${y_1}$};
\node (Q2) [right =.1mm of Q1, smallobedientposition] {\footnotesize ${y_4}$};
\node (Q3) [right =.1mm of Q2, smallobedientposition] {\footnotesize ${y_5}$};

\node (R1) [right = 5mm of Q3, smallobedientprimposition] {\footnotesize ${y_2}$};
\node (R2) [right =.1mm of R1, smallobedientposition] {\footnotesize ${y_6}$};
\node (R3) [right =.1mm of R2, smallobedientposition] {\footnotesize ${y_7}$};
\node (R4) [right =.1mm of R3, smallobedientposition] {\footnotesize ${y_8}$};
\node (R5) [right =.1mm of R4, smallobedientposition] {\footnotesize ${y_9}$};

\node (S1) [above left =9mm of P2, smallobedientprimposition] {\footnotesize ${y_3}$};
\node (S2) [right =.1mm of S1, smallobedienthackposition] {\footnotesize ${y_{10}}$};

\node (T1) [right = 5mm of S2, smallobedientprimposition] {\footnotesize ${y_4}$};

\node (U1) [right = 5mm of T1, smallobedientprimposition] {\footnotesize ${y_4}$};
\node (U2) [right =.1mm of U1, smallobedienthackposition] {\footnotesize ${y_{11}}$};

\node (V1) [right = 5mm of U2, smallobedientprimposition] {\footnotesize ${y_7}$};
\node (U2) [right =.1mm of V1, smallobedienthackposition] {\footnotesize ${y_{12}}$};

\node (W1) [right = 5mm of U2, smallobedientprimposition] {\footnotesize ${y_9}$};
\node (W2) [right =.1mm of W1, smallobedienthackposition] {\footnotesize ${y_{13}}$};
\node (W3) [right =.1mm of W2, smallobedienthackposition] {\footnotesize ${y_{14}}$};
\node (W4) [right =.1mm of W3, smallobedienthackposition] {\footnotesize ${y_{15}}$};

\node (N1) [left =25mm of O1, smallprimposition] {\footnotesize $t_1$};
\node (N2) [right =.1mm of N1] { $\dots$};
\node (N3) [right =.1mm of N2, smallprimhackposition] {\footnotesize $t_n$};

\node (plus) [right =5mm of N3] { $+$};

\path (O4.south west)
        edge[decorate,decoration={brace,mirror,raise=.5mm},"\footnotesize $P$"below=3pt]
        (O4.south west -| O5.south east);

\draw [smalledge] (O4.north) -- (P1.south);
\draw [smalledge] (O4.north) -- (Q1.south);
\draw [smalledge] (O5.north) -- (R1.south);
\draw [smalledge] (P2.north) -- (S1.south);
\draw [smalledge] (Q2.north) -- (T1.south);
\draw [smalledge] (Q2.north) -- (U1.south);
\draw [smalledge] (R3.north) -- (V1.south);
\draw [smalledge] (R5.north) -- (W1.south);

\begin{scope}[transform shape] 
\path [draw,rounded corners]  
               ([shiftBL] O1.south west) 
            -- ([shiftbl] S1.south west) 
            -- ([shifttl] S1.north west) 
            -- ([shifttr] W4.north east)
            -- ([shiftbr] W4.south east)  
            -- ([shiftbr] R5.south east)  
            -- ([shiftBR] O6.south east)  node [above,midway] {\footnotesize $\fkclosure{P}{q}{\fk}$}
            --  cycle; 
\end{scope} 

\end{tikzpicture}
}
\end{center}
\caption{Structure of $\fkclosure{P}{q}{\fk}$ over obedient $P$ (omitting weak foreign keys). Terms $t_1, \dots ,t_n$ occupying $\poscomplement{P}{\fk}$ do not occur among the (pairwise distinct) variables $y_1,\dots ,y_{15}$ occupying $\posclosure{P}{\fk}$.
The polygon encloses $\fkclosure{P}{q}{\fk}$; the green boxes mark  $\posclosure{P}{\fk}$.
\label{fig:obgraph}
}
\end{figure}


We now come to Definition~\ref{def:block-int} of block-interference, which uses the following adapted notion of Gaifman graph~\cite[Def.~4.1]{DBLP:books/sp/Libkin04}.
For a query in $\sjfbcq$ and $V\subseteq\queryvars{q}$, we define $\constgraph{V}{q}$ for the undirected graph whose vertex-set is $V$, and where $\{x,y\}$ is an undirected edge if $x=y$ or there is $F\in q$ such that $\{x,y\} \subseteq \atomvars{F}\cap V$. 


\begin{definition}[Block-interfering]\label{def:block-int}
Let $q$ be a query in $\sjfbcq$,
and $\fk$ a set of foreign keys about $q$. 
Let $\foreignkey{N}{j}{O}$ be a strong foreign key in $\lclosure{\fk}$. 
Let $N(\underline{t_{1}, \dots ,t_k},t_{k+1},\dots,t_{n})$ and  $O(\underline{t_j},\vec{y})$ be atoms in~$q$ (since the foreign key is strong, $j>k$).
Let $V=\{v\in \queryvars{q'}\mid \FD{q}\not\models\fd{\emptyset}{\{v\}}\}$, where $q'\defeq q\setminus\{N(\underline{t_1, \dots ,t_k},t_{k+1},\dots,t_{n})\}$.
We say that this foreign key is \emph{block-interfering (in $q$)} if the following hold: 
\begin{enumerate} 
\item\label{it:obedientO}
the atom $O(\underline{t_j},\vec{y})$ is obedient; 
\item\label{it:no-constant}
$t_j$ is a variable in $V$ (thus $\FD{q}\not\models\fd{\emptyset}{\{t_{j}\}}$); and
\item
at least one of the following holds true:
\begin{enumerate}
\item\label{it:marking}
 $\{(N,k+1), \ldots ,(N,n)\}\setminus\{(N,j)\}$ is not obedient; or 
\item\label{it:bifour}
for some $i\in\{1,\dots,k\}$, $t_i$ and $t_j$ are  (not necessarily distinct) variables that are connected in $\constgraph{V}{q'}$.
\end{enumerate}
\end{enumerate} 
We say that the pair $(q,\fk)$ has \emph{block-interference} if some foreign key of $\lclosure{\fk}$ is block-interfering in $q$.
\qed
\end{definition}

It can be seen that, due to properties~\eqref{it:marking} or \eqref{it:bifour} in Definition~\ref{def:block-int}, the $N$-atom in this definition will itself be disobedient.

\begin{example}\label{ex:four}
Continuing Example~\ref{ex:no}, consider again the query $q=\{N(\underline{x},c,y)$, $O(\underline{y})\}$ with $\fk=\{\foreignkey{N}{3}{O}\}$, where the atom $O(\underline{y})$ is obviously obedient.
Following the notations of Definition~\ref{def:block-int}, we obtain block-interference by letting $j=3$ and therefore $t_{j}=y$. 
The set difference in item~\eqref{it:marking} of Definition~\ref{def:block-int} becomes $\{(N,2)\}$, which is not obedient as shown in Example~\ref{ex:no}.
\qed
\end{example}

The following example shows the use of property~\eqref{it:bifour} in Definition~\ref{def:block-int}.

\begin{example}
Consider $q_{0}=\{N'(\underline{x},y)$, $O(\underline{y})$, $T(\underline{x,y})\}$ and $\fk=\{\foreignkey{N'}{2}{O}\}$.
In comparison with the previous Example~\ref{ex:four}, we removed the constant~$c$ that allowed us to distinguish, within an $N$-block, between satisfying and falsifying $N$-facts.
However, since $x$ and $y$ occur together in the $T$-atom of $q_{0}$, we can now use $T$-facts to make this distinction.
Indeed, in the database $\db$ at the beginning of this section, we can replace every ``satisfying'' fact  $N(\underline{b_{i}},c,i)$ with  two facts $N'(\underline{b_{i}},i)$ and $T(\underline{b_{i},i})$, while every ``falsifying'' fact  $N(\underline{b_{i}},d,i+1)$ is replaced with a single fact $N'(\underline{b_{i}},i+1)$ (for $1\leq i\leq n+1$).
Informally, the role of the constant $c$ is now played by~$T$.

To illustrate the role of the set $V$ in Definition~\ref{def:block-int},
we note that our construction with $T$-facts would fail if for some constant~$a$, the query $q_{0}$ also contained $R(\underline{a},x)$ (yielding a functional dependency $\fd{\emptyset}{\{x\}}$), because no $\sd$-repair can contain both $R(\underline{a},b_{i})$ and $R(\underline{a},b_{j})$ with $i\neq j$.
\qed
\end{example}

\section{Main Theorem}\label{sec:mainthm}
The following theorem refines Theorem~\ref{thm:mainfo} by adding the conditions to decide whether or not $\certainty{q}{\fk}$ is in $\FO$. 
To show that a problem $\certainty{q}{\fk}$ is not in $\FO$, we show that it is $\L$-hard or $\NL$-hard.

\begin{theorem}\label{thm:main}
Let $q$ be a query in $\sjfbcq$, and let $\fk$ be a set of unary foreign keys about $q$.
Then,
\begin{enumerate}
\item\label{it:forew}
if the attack graph of $q$ is acyclic and $(q,\fk)$ has no block-interference, then $\certainty{q}{\fk}$ is in $\FO$ (and its consistent first-order rewriting can be effectively constructed);
\item\label{it:lhard} if the attack graph of $q$ is cyclic, then $\certainty{q}{\fk}$ is $\L$-hard (and therefore not in $\FO$); and
\item\label{it:nlhard} if $(q,\fk)$ has block-interference, then $\certainty{q}{\fk}$ is $\NL$-hard (and therefore not in $\FO$).
\end{enumerate}
Moreover, it can be decided, given $q$ and $\fk$, which case applies.
\end{theorem}
There is an easy proof for the last line in the statement of the above theorem.
Indeed, it is known that, given $q$ in $\sjfbcq$, it can be decided in quadratic time whether or not $q$'s attack graph is acyclic~\cite[Theorem~3.2]{DBLP:journals/tods/KoutrisW17}.
Moreover, it is clear that the existence of block-interference is decidable in polynomial time by inspecting the conditions in Definition~\ref{def:block-int} and the syntactic characterization of obedience in Theorem~\ref{thm:syntactic-obedience}.

The following example illustrates Theorem~\ref{thm:main}, and shows that consistent query answering over foreign keys depends in a subtle way on the syntax of the query. 

\begin{example}\label{ex:block}
For variables $x,y,z,w$, and a constant $c$, let
\begin{align*}
\fk&=\{\foreignkey{N}{3}{O}\};\\
q_{1} & = \{N(\underline{x},u,y), O(\underline{y},w)\};\\
q_{2} & = \{N(\underline{x},c,y), O(\underline{y},w)\};\\
q_{3} & = \{N(\underline{x},c,y), O(\underline{y},c)\}.
\end{align*}
Note that $q_{2}$ and $q_{3}$ can be obtained from $q_{1}$ by replacing variables with constants: $q_{2}=\substitute{q_{1}}{u}{c}$ and $q_{3}=\substitute{q_{1}}{u,w}{c,c}$.
The attack graph of each query is acyclic, and hence $\cqa{q_{i}}$ is in $\FO$ for $i\in\{1,2,3\}$.
The complexity and consistent first-order rewritings change as follows in the presence of $\fk$.


\begin{itemize}
\item 
$\certainty{q_{1}}{\fk}$ is in $\FO$ because $\foreignkey{N}{3}{O}$ is not block-interfering in $q_1$, even though the atom $O(\underline{y},w)$ is obedient.
It can be formally verified that condition~\eqref{it:marking} in Definition~\ref{def:block-int} is not satisfied: the position $(N,2)$ in $q_{1}$ is obedient, because it is occupied by a variable that occurs only once in the query.
The consistent first-order rewriting for $\certainty{q_{1}}{\fk}$ is the query~$q_1$ itself.
Remarkably, this is different from the consistent first-order rewriting for $\cqa{q_{1}}$ (i.e., in the absence of foreign keys).
To see the difference, note that the following database instance is a ``yes''-instance of $\certainty{q_{1}}{\fk}$, but a ``no''-instance of $\cqa{q_{1}}$. 

$$
\begin{array}{cc}
\begin{array}[t]{c|ccc}
N & \underline{x} & u & y\\\cline{2-4}
  & c & 1 & a\\
  & c & 2 & b\\\cdashline{2-4}
\end{array}
&
\begin{array}[t]{c|cc}
O & \underline{y} & w\bigstrut\\\cline{2-3}
  & a & 3\\\cdashline{2-3}
\end{array}
\end{array}
$$ 
\mbox{}

\item
$\certainty{q_{2}}{\fk}$ is $\NL$-hard, because $\foreignkey{N}{3}{O}$ is block-interfering in $q_2$. 
Informally, this is because the position $(N,2)$ is now occupied by a constant~$c$ and therefore not obedient.

\item 
$\certainty{q_{3}}{\fk}$ is again in $\FO$, because $\foreignkey{N}{3}{O}$ is not block-interfering in $q_2$.  The reason is that the $O$-atom is no longer obedient because its non-primary-key position is now occupied by a constant.
With some effort, one can see that $\certainty{q_{3}}{\fk}$ and $\cqa{q_{3}}$ have the same consistent first-order rewriting.
\end{itemize}

To conclude, replacing a variable by a constant can increase or decrease the complexity, depending on where the variable occurs. This behavior is typical of foreign keys, and does not occur in the case of only primary keys. 
\qed
\end{example}

%
%

The following sections are devoted to the proof of Theorem~\ref{thm:main}. 
In Section \ref{sec:lhard}, we prove item~\eqref{it:lhard} of Theorem~\ref{thm:main},
and in Section \ref{sec:nl} we prove item~\eqref{it:nlhard}.
Finally, item~\eqref{it:forew} is shown in~Section~\ref{sec:fo}.

\section{$\L$-Hardness}\label{sec:lhard}

We know from Theorem~\ref{thm:koutrisw} that $\certainty{q}{\fk}$ is $\L$-hard if $\fk=\emptyset$ and the attack graph of $q$ is cyclic.
The following lemma tells us that this complexity lower bound remains valid if we add foreign keys to $\fk$.
It is worth mentioning that it can be proved for foreign keys that need not be unary (see Appendix \ref{sect:appC}). 

\begin{lemma}\label{lem:cutchase}
Let $q$ be a query in $\sjfbcq$,
and $\fk$ be a set of foreign keys about~$q$.
If $q$ has a cyclic attack graph, then $\certainty{q}{\fk}$ is $\L$-hard. 
\end{lemma}

For example, since the attack graph of $q=\{R(\underline{x},y), S(\underline{y},x)\}$ is cyclic, $\certainty{q}{\fk}$ is $\L$-hard, for every $\fk$ that is a (possibly empty) subset of $\{\foreignkey{R}{2}{S}, \foreignkey{S}{2}{R}\}$.

\section{\NL-hardness}\label{sec:nl}

The following lemma, proven in Appendix \ref{sect:apD}, restates item~\eqref{it:nlhard} of Theorem~\ref{thm:main}.

\begin{lemma}\label{lem:bi}
Let $q$ be a query in $\sjfbcq$,
and $\fk$ be a set of unary foreign keys about $q$.
If $(q,\fk)$ has block-interference, then the problem $\certainty{q}{\fk}$ is $\NL$-hard. 
\end{lemma}

For an intuition why block-interference leads to $\NL$-hardness, consider again the example with $q=\{N(\underline{x},c,y)$, $O(\underline{y})\}$ and $\fk=\{\foreignkey{N}{3}{O}\}$, elaborated in the beginning of Section~\ref{sec:bi}, where it was argued that $\certainty{q}{\fk}$ goes beyond locality of first-order logic.
With this preceding example in mind, it should not come as a surprise that directed graph reachability can be reduced to (the complement of) $\certainty{q}{\fk}$.
In graph reachability, the input consists of a directed graph and two vertices ($s$ and $t$), and the question is whether there is a directed path from~$s$ to~$t$.
The problem is $\NL$-hard, even if the graphs are acyclic.
Figure~\ref{fig:reach} illustrates a straightforward reduction: 
for every vertex $v$ such that $v\neq t$, add an $N$-fact $N(\underline{v},c,v)$;
for every directed edge $(u,w)$, add $N(\underline{u},d,w)$.
Finally, add $O(\underline{s})$.
The path from $s$ to $t$ (via vertex $2$) in the database instance of Fig.~\ref{fig:reach} can be cooked into the following $\sd$-repair that falsifies $q$:
$$
\begin{array}{cc}
\begin{array}[t]{c|ccc}
N & \underline{x} & c & y\\\cline{2-4}
  & s & d & 2\\
  & 2 & d & t
\end{array}
&

\begin{array}[t]{c|c}
O & \underline{y}\bigstrut\\\cline{2-2}
  & s\\
  & 2\\
  & t
\end{array}
\end{array}
$$
On the other hand, it can be easily verified that there would be no falsifying $\sd$-repair if every path starting from $s$ ended in a vertex other than $t$. 
The reasoning is analogous to the one used in the beginning of Section~\ref{sec:bi}.

\tikzstyle{edge}=[->,shorten >=.5pt,>=stealth,thick]
\tikzstyle{vertex}=[circle,draw=black!50,fill=blue!10,thick,inner sep=0pt,minimum size=4mm]

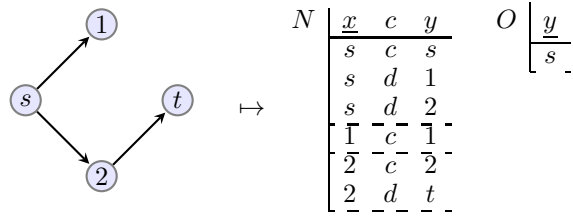
\begin{figure}\centering
\begin{tabular}{ccc}
\begin{adjustbox}{valign=t}
\begin{tikzpicture}
\node (1) [vertex] {$s$};
\node (2) [vertex, above right =10mm of 1] {$1$};
\node (3) [vertex, below right  =10mm of 1] {$2$};
\node (4) [vertex, below right  =10mm of 2] {$t$};

\draw [edge] (1) -- (2);
\draw [edge] (1) -- (3);
\draw [edge] (3) -- (4);
\end{tikzpicture}
\end{adjustbox}
&
$
\begin{array}[t]{cc|ccc}
& N & \underline{x} & c & y\\\cline{3-5}
&  & s & c & s\\
&  & s & d & 1\\
\mapsto &  & s & d & 2\\\cdashline{3-5}
&  & 1 & c & 1\\\cdashline{3-5}
&  & 2 & c & 2\\
&  & 2 & d & t\\\cdashline{3-5}
\end{array}
$
&
$
\begin{array}[t]{c|c}
O & \underline{y}\bigstrut\\\cline{2-2}
  & s \\\cdashline{2-2}
\end{array}
$
\end{tabular}
\caption{Reduction from graph reachability.}
\label{fig:reach}
\end{figure}

The previous example gives a correct intuition for the proof of Lemma~\ref{lem:bi}.
The reason why its proof is technically much more involved is that Definition~\ref{def:block-int} (and especially condition~\eqref{it:marking} in it) exhibits several ways in which block-interference can arise. In the previous example, we only looked at the very simple case where block-interference uses a constant. In more difficult situations, block-interference arises from cycles in the dependency graph or repetitions of variables. 


In the absence of foreign keys, for every $q$ in $\sjfbcq$, the problem $\cqa{q}$ is either in $\FO$, $\L$-complete, or $\coNP$-complete~\cite{DBLP:journals/mst/KoutrisW21}.
Interestingly, in the presence of foreign keys, $\NL$-completeness and $\P$-completeness also pop up, as shown next.

\begin{proposition}\label{pro:nlc}
$\certainty{q}{\fk}$ is $\NL$-complete for $q=\{N(\underline{x},x)$, $O(\underline{x})\}$ and $\fk=\{\foreignkey{N}{2}{O}\}$.
\end{proposition}

\begin{proposition}\label{pro:pc}
$\certainty{q}{\fk}$ is $\P$-complete for $q=\{N(\underline{x},c,y)$, $O(\underline{y})\}$ and $\fk=\{\foreignkey{N}{3}{O}\}$.
\end{proposition}

A fine-grained complexity classification for \emph{all} problems in the set $\{\certainty{q}{\fk}\mid\mbox{$q\in\sjfbcq$ and $\fk$ is about $q$}\}$ is open; in the current paper, we succeed in tracing the $\FO$-boundary in the above set.

\section{First-Order Rewritability}\label{sec:fo}

\begin{figure}\centering
\begin{tabular}{|c|c|c|c|}\hline
$R$-atom & $S$-atom & Type & \bigstrut\\\hline\hline
 &  & $\weak$ & Lemma~\ref{lem:nonProper}\\\hline
obedient & obedient & $\opo$ & Lemma~\ref{lem:opo}\\\hline
disobedient & disobedient & $\dpd$ & Lemma~\ref{lem:dpd}\\\hline
disobedient & obedient & $\dpo$ & Lemmas~\ref{lem:dpo} and~\ref{lem:constantkey}\\\hline
\end{tabular}
\caption{Reductions that remove foreign keys $\foreignkey{R}{i}{S}$.
}\label{fig:overview}
\end{figure}

The following lemma restates item~\eqref{it:forew} of Theorem~\ref{thm:main}.

\begin{lemma}\label{lem:inFO}
Let $q$ be a query in $\sjfbcq$, and $\fk$ a set of unary foreign keys about~$q$.
If the attack graph of $q$ is acyclic and $(q,\fk)$ has no block-interference,
then $\certainty{q}{\fk}$ is in $\FO$ (and its consistent first-order rewriting can be effectively constructed).
\end{lemma}

We sketch how the previous lemma is proved (see Appendix \ref{sect:apE} for full details).
For two decision problems $P_{1}$ and $P_{2}$, we write $P_{1}\reducesTo{\FO}P_{2}$ if there exists a first-order many-one reduction from $P_{1}$ to $P_{2}$.

Let $q$ and $\fk$ be as stated in Lemma~\ref{lem:inFO} such that the attack graph of $q$ is acyclic and $(q,\fk)$ has no block-interference.
The proof strategy is to show that one can construct a query $q'$ in $\sjfbcq$ such that $q'$ has an acyclic attack graph and 
\begin{equation}\label{eq:emptyfk}
\certainty{q}{\fk}\reducesTo{\FO}\certainty{q'}{\emptyset}.
\end{equation}
Since the latter problem has an empty set of foreign keys, it is in $\FO$ by Theorem~\ref{thm:koutrisw}. 

Equation~\eqref{eq:emptyfk} is shown by a composition of first-order reductions, each of which removes at least one foreign key, and some of which remove obedient atoms or replace variables with constants.
The helping lemmas that define these reductions are summarized in Fig.~\ref{fig:overview} and are given in Appendix~\ref{sect:apE}.
We distinguish between four \emph{types of foreign keys}.
A \emph{strong} foreign key $\foreignkey{R}{i}{S}$ is of a type in $\{\opo, \dpd, \dpo\}$, depending on whether the $R$-atom or $S$-atom are obedient (symbol $\obed$) or disobedient (symbol $\nobed$).
Note that there is no type $\opd$, because if the $R$-atom is obedient and the foreign key is strong, then the $S$-atom is necessarily obedient as well.
For weak foreign keys there is only one type, denoted $\weak$.

Note in Definition~\ref{def:block-int} that only foreign keys of type $\dpo$ can be block-interfering.
Unsurprisingly, the requirement, in Lemma~\ref{lem:inFO}, that $(q,\fk)$ has no block-interference is used in (and only in) the helping Lemma~\ref{lem:dpo} that shows the removal of foreign keys of type $\dpo$.

It becomes apparent from the proofs of the helping lemmas that whenever $\certainty{q}{\fk}$ is in $\FO$, its consistent first-order rewriting is very similar to that of $\cqa{q}$~\cite{DBLP:journals/tods/KoutrisW17}, except for obedient atoms referenced by strong foreign keys.
For example, consider  $q=\{N(\underline{c},y)$, $O(\underline{y})$,  $P(\underline{y})\}$ with $\fk=\{\foreignkey{N}{2}{O}\}$, where $O$ is referenced but $P$ is not. The following is a consistent first-order rewriting for $\certainty{q}{\fk}$:
$$
\exists y\formula{N(\underline{c},y)\land O(\underline{y})}
\land\forall y\formula{N(\underline{c},y)\rightarrow P(\underline{y})}.
$$
Note the asymmetric treatment of $O$ and $P$ in the above formula.
In this respect, it is instructive to note that the following database instance satisfies the previous formula and hence is a ``yes''-instance.
However, removing either $P(\underline{a})$ or $P(\underline{b})$ turns it into a ``no''-instance. 
$$
\begin{array}{ccc}
\begin{array}[t]{c|ccc}
N & \underline{c} & y\\\cline{2-3}
  & c & a\\
  & c & b\\\cdashline{2-3}
\end{array}
&
\begin{array}[t]{c|c}
O & \underline{y}\bigstrut\\\cline{2-2}
  & a\\\cdashline{2-2}
\end{array}
&
\begin{array}[t]{c|c}
P & \underline{y}\bigstrut\\\cline{2-2}
  & a\\\cdashline{2-2}
  & b\\\cdashline{2-2}
\end{array}
\end{array}
$$

%
%

\section{Discussion}\label{sec:discussion}

While CQA for primary keys was successfully studied in the past 15~years,
CQA with respect to both primary and foreign keys remained largely unexplored. 
We made a significant contribution by tracing the $\FO$-boundary in the set  
$\{\certainty{q}{\fk}\mid\mbox{$q\in\sjfbcq$ and $\fk$ is about $q$}\}$,
under the restriction that foreign keys are unary (but primary keys can be composite).
If $\fk=\emptyset$, then these problems only have primary-key constraints, in which case a complete complexity classification in $\FO$, $\L$-complete, and $\coNP$-complete is already known~\cite{DBLP:conf/pods/KoutrisW20}.
For non-empty sets $\fk$,  a complete complexity classification beyond $\FO$ is left open.
Our paper nevertheless shows that the complexity landscape is more diverse than for primary keys alone, as Propositions~\ref{pro:nlc} and~\ref{pro:pc} show that there are $\NL$-complete and $\P$-complete problems in the above set of problems.

It is an open research task to release our restrictions that foreign-keys are unary and are about the query, as discussed next.

\begin{itemize}
\item
Our assumption that all foreign keys are unary excludes, for example, a query with atoms
$R(\underline{x},y,z)$, $S(\underline{x,z},y)$ and foreign key $\foreignkey{R}{1,3}{S}$.
The difficulty here is that the foreign key covers both a primary-key and a non-primary-key position of~$R$. 
In future research, we will investigate how our constructs of obedience and block-interfering can be generalized to composite foreign keys.
\item
Our assumption that all foreign keys are about the query excludes, for example, the problem in the following Proposition~\ref{pro:notabout}, because $q=\{E(\underline{x},y)\}$ does not satisfy $\foreignkey{E}{2}{E}$ (when $x$ and $y$ are treated as distinct constants).

\begin{proposition}\label{pro:notabout}
Let $q=\{E(\underline{x},y)\}$ and $\fk=\{\foreignkey{E}{2}{E}\}$.
Then, $\certainty{q}{\fk}$ is $\NL$-hard.
\end{proposition}

Concerning the previous proposition, 
note that every conjunctive query $q'$ that includes~$q$ and satisfies $\fk$ contains a self-join.
The shortest such a query is $q'=\{E(\underline{x},y), E(\underline{y},x)\}$.
CQA for conjunctive queries with self-joins is a notorious open problem, even in the absence of foreign keys.
\end{itemize}


\section*{Acknowledgments}
Miika Hannula has been supported by Academy of Finland grants 308712 and 322795.

\bibliographystyle{unsrt}
\bibliography{dblp}

\appendix
\section{Helping Notions and Lemmas}\label{sec:apHelp}

In this section, we define more preliminary notions and helping lemmas. The following definitions are relative to a database instance~$\db$,  a query $q$ in $\sjfbcq$, and a set $\fk$ of foreign keys.

We write $\adom{\db}$ for the set of constants that occur in $\db$, also called its \emph{active domain}.

A variable~$x\in\queryvars{q}$ is called \emph{orphan (in $q$)} if $x$ occurs only once in~$q$, and this single occurrence is at a non-primary-key position.
Similarly, a constant~$c$ in~$\db$ is called \emph{orphan (in $\db$)} if $c$ occurs only once in~$\db$, and this single occurrence is at a non-primary-key position.

Two variables $x,y\in\queryvars{q}$ are said to be \emph{connected in $q$} if $x=y$ or there exists a sequence $x_{0},x_{1},\ldots,x_{\ell}$ of variables in $\queryvars{q}$ such that $x_{0}=x$, $x_{\ell}=y$, and every two adjacent variables occur together in some atom of~$q$.

We say that a fact $A$ of $\db$ is \emph{relevant for $q$ in $\db$} if there exists a valuation $\theta$ over $\queryvars{q}$ such that $A\in\theta(q)\subseteq\db$ (and therefore $A\in\db$); otherwise $A$ is \emph{irrelevant}. 
A block of $\db$ is \emph{relevant} if it contains at least one relevant fact. 

We write $\restrict{\db}{q}$ for the restriction of $\db$ to those facts whose relation name occurs in~$q$.
We write $\restrict{\fk}{q}$ for the set of those foreign keys in $\fk$ that only use relation names in $q$. Clearly, if $\fk$ is about $q$, then $\restrict{\fk}{q}=q$. 

If $R$ is a relation name with signature~$\signature{n}{1}$, then the (weak) foreign key $\foreignkey{R}{1}{R}$ is called \emph{trivial}, because it cannot be falsified. 
If $\fk$ is a set of foreign keys and $R$ a relation name, then 
$\outgoing{R}{\fk}$ is the set of foreign keys in $\fk$ that are outgoing from $R$, and $\incoming{R}{\fk}$ is the set of foreign keys in $\fk$ that are referencing~$R$.

\begin{lemma}\label{lem:relevantblockter}
Let $\pk\cup\fk$ be a set of primary keys and foreign keys.
Let $\db$ be a (possibly inconsistent) database instance, and let $\rep$ be a $\sd$-repair of $\db$.
Let $\sep$ be a database instance such that $\sep\subseteq\rep\cup\db$ and $\sep\models\pk\cup\fk$.
For every fact $A\in\sep\cap\db$,
there is a fact $A'\in\rep\cap\db$ such that $A'\keyeq A$.
\end{lemma}
\begin{proof}
Let $\sep\setminus\rep=\{A_{1}, A_{2},\dots, A_{n}\}$.
Since $\sep\subseteq\rep\cup\db$, each $A_{i}$ belongs to $\db$.
Let $\tep_{0}\defeq\rep$.
For $i=1,2,\ldots,n$,
\begin{enumerate}
\item\label{it:relevantnop}
if $\rep\cap\db$ contains a fact that is key-equal to $A_{i}$, 
let $\tep_{i}\defeq\tep_{i-1}$;
\item\label{it:relevantreplace}
if $\rep\setminus\db$ contains a fact $A_{i}'$ that is key-equal to $A_{i}$,
let $\tep_{i}\defeq\formula{\tep_{i-1}\setminus\{A_{i}'\}}\cup\{A_{i}\}$; and
\item\label{it:relevantinsert}
if $\rep$ contains no fact that is key-equal to $A_{i}$,
let $\tep_{i}\defeq\tep_{i-1}\cup\{A_{i}\}$.
\end{enumerate}
Let $\tep\defeq\tep_{n}$.
From $\rep\models\pk$ and $\sep\models\pk$, it follows $\tep\models\pk$ by construction.

We show that $\tep\models\fk$.
To this end, let $\foreignkey{R}{i}{S}$ be a foreign key in $\fk$, and let $R(a_{1},\dots,a_{n})$ be a fact in $\tep$.
If $R(a_{1},\dots,a_{n})\in\rep$, then this foreign key is satisfied by $\tep$ because $\rep\models\fk$ and, by construction, every fact in $\rep$ is key-equal to a fact in~$\tep$. 
Assume next that $R(a_{1},\dots,a_{n})\in\sep\setminus\rep$. 
Since $R(a_{1},\dots,a_{n})\in\sep$ and $\sep\models\fk$,
$\sep$~contains a fact $S(\underline{a_{i}},\filler)$.
By construction, $\tep$ will contain a fact that is key-equal to $S(\underline{a_{i}},\filler)$.

By construction,
$\rep\cap\db\subseteq\tep$ and  $\tep\subseteq\rep\cup\db$.
It follows $\tep\closer{\db}\rep$.
If~\eqref{it:relevantreplace} or~\eqref{it:relevantinsert} are applied once or more, then $\tep\closerneq{\db}\rep$, contradicting that $\rep$ is a $\sd$-repair.
It follows that only~\eqref{it:relevantnop} applies,
which means that for every $A\in\sep\cap\db$, $\rep$ contains a fact of $\theblock{A}{\db}$.
\end{proof}

\begin{lemma}\label{lem:norepair}
Let $q$ be a query in $\sjfbcq$, and $\fk$ a set of foreign keys that is satisfied by $q$ (when distinct variables are treated as distinct constants).
Let $\db$ be a (possibly inconsistent) database instance.
Let $\rep$ be a database instance that satisfies $\fk\cup\pk$.
Let $\theta$ be a valuation over $\queryvars{q}$ satisfying the following conditions:
\begin{enumerate}
\item\label{it:allindb}
$\theta(q)\subseteq\db\cup\rep$; and
\item\label{it:outside}
there is a fact $A\in\theta(q)\setminus\rep$ such that $\rep\cap\db$ contains no fact that is key-equal to $A$.
\end{enumerate}
Then $\rep$ is not a $\sd$-repair.
\end{lemma}
\begin{proof}
Since $q\models\fk$ and since $q\in\sjfbcq$,
we have $\theta(q)\models\pk\cup\fk$.
Assume towards a contradiction that $\rep$ is a $\sd$-repair.
By Lemma~\ref{lem:relevantblockter}, 
for every fact $A\in\theta(q)\cap\db$, there is a fact $A'\in\rep\cap\db$ such that $A'\keyeq A$, contradicting~\eqref{it:outside}.
\end{proof}

\begin{corollary}\label{cor:relevantblockbis}
Let $\fk$ be a set of foreign keys.
Let $q$ be a query in $\sjfbcq$ that satisfies $\fk$ (when distinct variables are treated as distinct constants).
Let $N$ be a relation name that occurs in $q$.
For every database instance $\db$,
if $\block$ is an $N$-block of $\db$ that is relevant for $\fkclosure{N}{q}{\fk}$ in $\db$,
then every $\sd$-repair of $\db$ contains a fact from $\block$.
\end{corollary}
\begin{proof}
Let $\db$ be a database instance.
Let $\theta$ be a valuation over $\queryvars{\fkclosure{N}{q}{\fk}}$ such that $\theta(\fkclosure{N}{q}{\fk})\subseteq\db$.
Let $A$ be the $N$-fact in $\theta(q)$.
Since it is easily verified that $\theta(\fkclosure{N}{q}{\fk})$ is consistent with respect to $\fk$ and primary keys, it follows by Lemma~\ref{lem:relevantblockter} that every $\sd$-repair contains a fact of $\theblock{A}{\db}$.
\end{proof}
\section{Proofs for Section~\ref{sec:bi}}\label{sect:apA}

In this section we show that the concept of obedience can be characterized in syntactic terms (Theorem~\ref{thm:syntactic-obedience}). The proof relies on Lemma~\ref{lem:gadget} which is proven using the below chase rule.

 Let $\fk$ be a set of unary foreign keys that is about some query $q$ in $\sjfbcq$. Let $\db$ be a database, and let $C\subseteq \adom{\db}$. Let $\bot,\top$ be two fresh constant (i.e., $\bot,\top\notin \queryconst{q}\cup\adom{\db}$).
 Consider the following non-deterministic chase rule.
 
 \noindent
 \textbf{Chase rule.} If $T(\underline{a_1,\dots ,a_k},a_{k+1},\dots ,a_m)\in \db$ is dangling with respect to $\foreignkey{T}{i}{U}\in \fk$, extend $\db$ with $U(\underline{a_i},b_2, \dots ,b_{m})$, where  $\{b_2, \dots ,b_{m}\}\subseteq \hat C$. 

Above, we say that $U(\underline{a_i},b_2, \dots ,b_{m})$ is  \emph{generated} by $\foreignkey{T}{i}{U}$ and $T(\underline{a_1,\dots ,a_k},a_{k+1},\dots ,a_m)$.  
 Denote by $\cchase{\db}{\fk}{C}$ the set of all database instances obtained from $\db$ by applying the chase rule as many times as possible with respect to foreign key set $\fk$ and constant set $C$. Note that each database in $\cchase{\db}{\fk}{C}$ is finite.

Given positions $(R,i)$ and $(S,j)$, consider also the following additional restrictions for the chase rule:
\begin{enumerate}
\item\label{it:restone} 
 $b_2= \dots =b_{m}=a_i$; except that $ b_i\neq a_i$ if $U=R$.
\item\label{it:resttwo}  $b_2, \dots ,b_{m}\in C$; except that $b_i=\bot$ if $U=R$, and $b_j=\bot$ if $U=S$.
\end{enumerate}
We denote by $\Cchase{\db}{\fk}{(R,i)}{C}{1}$ (resp. \\$\Cchase{\db}{\fk}{(R,i),(S,j)}{C}{2}$) the set of all $\db\in \cchase{\db}{\fk}{C}$ whose construction obeys restriction \eqref{it:restone} (resp. restriction \eqref{it:resttwo}) over position $(R,j)$ (resp. positions $(R,i)$ and $(S,j)$) of the chase rule.

\begin{definition}
\label{def:gadget}
Let $q$ be a query in $\sjfbcq$, and let $\fk$ be a set of unary foreign keys that is about $q$. Let $\db$ be a database instance, and let $A=R(\underline{\vec{a}},b_{k+1},\dots,b_n)\in \db$. Let $F=R(\underline{\vec{s}}, t_{k+1}, \dots ,t_{n})\in q$ be an atom over signature $[n,k]$, and let $P\subseteq \{(R,i)\mid i\in \{k+1,\dots ,n\}\}$ be a set of positions that does not satisfy some of the conditions listed in Theorem~\ref{thm:syntactic-obedience}.
Define $F^*\defeq R(\underline{\vec{s}},u_{k+1},\dots,u_{n})$, where for every $i\in\{k+1,\dots,n\}$, $u_{i}=t_{i}$ if $(R,i)\not\in P$, and $u_{i}=b_i$ otherwise.
 Define $q^*\defeq (q\setminus \{F\})\cup\{F^*\}$,
and 
\[
\db_{A,P}\defeq \db\setminus q^*,
\]
for a database instance $\db$ constructed as follows (depending on which condition of Theorem~\ref{thm:syntactic-obedience} is violated, and viewing $q^*$ as a database by interpreting its variables as constants):
\begin{enumerate}[label=(\alph*)]
\item If $P$ does not satisfy condition~\ref{it:strongweak},
\[
\db\in \Cchase{q^*}{{\fk}}{(R,i)}{C}{1},
\] 
for any $(R,i)\in P$ that belongs to a cycle in the dependency graph of $\fk$.
\item Otherwise, if $P$ does not satisfy condition~\ref{it:no-const} or~\ref{it:disjoint}, 
\[\db\in \cchase{q^*}{{\fk}}{C}.
\]
\item Otherwise, if $P$ does not satisfy condition~\ref{it:distinct}, 
\[
\db\in \Cchase{q^*}{{\fk}}{(R,i),(S,j)}{C}{2},
\]
for any two non-primary-key positions $(R,i),(S,j)\in \posclosure{P}{\fk}$ that are occupied in $q$ by the same variable.
\end{enumerate}
\end{definition}

For a database instance $\db$, define 
\[
\keyconst{\db}\defeq \adom{\{R'(\underline{\vec{a}})\mid \exists \vec{b}: R(\underline{\vec{a}},\vec{b})\in \db\}}.
\] In words, $\keyconst{\db}$ is the set of constants that appear at a primary-key position in $\db$. 
\begin{lemma}\label{lem:gadget}
Let $q$ be a query in $\sjfbcq$, and let $\fk$ be a set of unary foreign keys that is about $q$. Let $\db$ be a database instance, and let $A=R(\underline{\vec{a}},b_{k+1},\dots,b_n)\in \db$. Let $P\subseteq \{(R,i)\mid i\in \{k+1,\dots ,n\}\}$ be a set of positions that does not satisfy some of the conditions listed in Theorem~\ref{thm:syntactic-obedience}. Assume that $C\defeq \{b_i\mid (R,i)\in P\}$ consists of orphan constants of $\db$ that do not belong to $\queryconst{q}$. 
Then, the following holds:
\begin{enumerate}
\item\label{it:gadget1}
$\keyconst{\db}\cap \adom{\db_{A,P}}=\emptyset$; 
\item\label{it:gadget1.5}
$\adom{\db}\cap \adom{\db_{A,P}}\subseteq  C$; 
\item\label{it:gadget2}
$\db_{A,P}\models \pk\cup {\fk}$;
\item\label{it:gadget4}
$A$ is not dangling in $\{A\}\cup \db_{A,P}$ with respect to any $\foreignkey{R}{i}{S}\in{\fk}$ such that $(R,i)\in P$;
 and
\item\label{it:gadget3}
every fact of $\{A\}\cup \db_{A,P}$ is irrelelevant for $q$ in $ \db\cup \db_{A,P}$.
\end{enumerate}
\end{lemma}
\begin{proof}
Let $q^*$  and $F^*$ be as in Definition \ref{def:gadget}. Define $\hat C \defeq C\cup \{\bot,\top\}$.

\noindent
\textbf{Item \eqref{it:gadget1}.}
Since $q\models {\fk}$, we note that no atom in $q^*$ is dangling with respect to ${\fk}$, with the exception that $F^*$ is dangling with respect to foreign keys of the form $\foreignkey{R}{i}{S}\in  {\fk}$, where $(R,i)\in P$. Hence we observe that $\adom{\db_{A,P}}\subseteq \hat C$. Moreover, no constant of $\hat C$ appears at a primary-key position of any fact in $\db$ (due to $C$ consisting of orphan constants of $\db$, each of which appears at a non-primary-key position of some fact in $\db$). Hence we obtain that
$\keyconst{\db}\cap \adom{\db_{A,P}}=\emptyset$.

\noindent
\textbf{Item \eqref{it:gadget1.5}.}
Trivial by $\adom{\db_{A,P}}\subseteq \hat C$ and $\bot,\top\notin \adom{\db}$.

\noindent
\textbf{Item \eqref{it:gadget2}.} 
We have $\adom{\db_{A,P}}\subseteq \hat C$ and $\hat C\cap \queryconst{q}=\emptyset$. It follows that no constant of $\adom{\db_{A,P}}$ appears at a primary-key position of any atom in $q^*$. It follows by the chase construction that 
$\db_{A,P}\models {\fk}$. That $\db_{A,P}\models {\pk}$ is likewise a consequence of the chase construction.

\noindent
\textbf{Item \eqref{it:gadget4}.} Let $\foreignkey{R}{i}{S}\in {\fk}$, where $(R,i)\in P$. Since $q^*\not\models \foreignkey{R}{i}{S}$, it must be the case that $\db_{A,P}$ contains an $S$-fact whose unique primary-key constant is $b_i$. The statement of item \eqref{it:gadget4} follows from this.

\noindent
\textbf{Item \eqref{it:gadget3}.}
Assume toward contradiction that there exists a valuation $\mu$ such that $\mu(q)\subseteq \db\cup  \db_{A,P}$ and $\mu(q)\cap (\{A\}\cup \db_{A,P})\neq \emptyset$. 
Consider first the following claim. The proof does not depend on the version of the chase rule being used.

\begin{claim}\label{claim:calF0}
$\mu(F)\in \{A\}\cup \db_{A,P}$.
\end{claim}
\begin{proof}
  As observed previously, since $q\models {\fk}$, no atom in $q^*$ is dangling with respect to ${\fk}$, except that $F^*$ is be dangling with respect to foreign keys of the form $\foreignkey{R}{i}{S}\in {\fk}$, $(R,i)\in P$. In particular, any sequence of applications of the chase rule to $q^*$ is initialized by a foreign key of this form.  
Hence, using the assumption that $\mu(q)\cap \db_{A,P}\neq \emptyset$,
 we find a path $((T_1,i_1), \dots ,(T_p,i_p))$ in the dependency graph of $\fk$ such that $(T_1,i_1)\in P$ and $ \mu(G_p)\in \db_{A,P}$, where by $G_i$, $i\in [p]$, we denote the unique $T_i$-atom of $q$.  W.l.o.g. we may assume that $\{(T_2,i_2), \dots ,(T_p,i_p)\}\cap P = \emptyset$. Let $s_1, \dots ,s_p$ be the terms occupying positions $(T_1,i_1), \dots ,(T_p,i_p)$ in $q$. 
 We first show by backward induction that 
 $\mu(s_k)\in \hat C$ 
 for $k\in [p]$. 
 
 \noindent
 \textbf{Base step $k=p$.}
  Immediate, for we have $ \mu(G_p)\in \db_{A,P}$ and $\adom{\db_{A,P}} \subseteq \hat C$. 
 
 \noindent
 \textbf{Induction step $k=h-1$.}
  Since $(T_{h},i_{h})\notin P$, and since $\mu(s_{h})\in \hat C$ by the induction hypothesis, we note  that $\mu(G_{h})\in \db_{A,P}$. In particular, we observe that $\mu(G_{h})$ contains only constants from $\hat C$.
 Since $\foreignkey{T_k}{i_k}{T_{h}}$ must belong to ${\fk}$, and since $q\models {\fk} $, we observe that the term $s_{k}$ occupies the unique primary-key position of $G_{h}$.   
 We conclude from these observations that $\mu(s_{k})$ must be from $\hat C$. This concludes the induction step $k=h-1$.

We have showed that $\mu(s)\in \hat C$ for a term $s$ that occurs at a position of $P$ in $q$. Since $F$ is the unique $R$-atom of $q$,  we obtain that $\mu(F)\in \{A\}\cup \db_{A,P}$.
\end{proof}

 \begin{claim}\label{claim:calF}
 Let $G=S(\underline{u_1, \dots ,u_{k'}}, u_{k'+1}, \dots ,u_{n'})\in q$. Then,
 \[
 (S,l)\in \posclosure{P}{\fk}\iff \mu(u_l)\in \hat C.
 \] 
 \end{claim}
 \begin{proof}[Proof of Claim \ref{claim:calF}]
 \framebox{$\implies$}
 Let $((R_1,i_1), \dots ,(R_m,i_m))$ be a path in the dependency graph of $\fk$, where $(R_1,i_1)\in P$ and $(R_m,i_m)=(S,l)$. For $j\in [m]$, denote by $H_j$ the unique $R_j$-atom of $q$, and let
  $w_j$ be the term that occurs at the $i_j$th position in $H_j$.
 It suffices to show that $\mu(w_j)\in \hat C$ for $j\in [m]$. The proof is by induction on $j$. 
  
  \noindent
 \textbf{Base step $j=1$.} Follows immediately by Claim \ref{claim:calF0}, as $\adom{\db_{A,P}}\subseteq \hat C$.
 
 \noindent
 \textbf{Induction step $j=h+1$.} 
    Since $q\models \fk$, and since there must be a non-special edge from $(R_h,i_h)$ to $(R_j,1)$, we observe that $w_h$ occupies the unique primary-key position of $H_j$. Applying the induction hypothesis that $\mu(w_h)\in \hat C$, and the fact that the constants of $\hat C$ do not appear at a primary-key position of any fact in $\db$, we obtain that $\mu(H_j)\in \db_{A,P}$, whence $\mu(w_j)\in \hat C$. This concludes the induction step $j=h+1$.

  \framebox{$\impliedby$} There are two possibilities: either $\mu(G)\in \db_{A,P}$ or  $\mu(G)=A$. In the first case, we can show that $(S,l)\in \posclosure{P}{\fk}$ using arguments of the kind used in the beginning of the proof of Claim \ref{claim:calF0}. In the second case we obtain $(S,l)\in {P}$, whence $(S,l)\in \posclosure{P}{\fk}$ vacuously.
 This concludes the proof of the claim.
 \end{proof}
Using Claim \ref{claim:calF} we obtain a contradiction if any of the conditions of Theorem \ref{thm:syntactic-obedience} is not true. Let us consider each case separately. 

\noindent
\textbf{Condition~\ref{it:strongweak}.} Assume that condition~\ref{it:strongweak} is not true. 
Then, the dependency graph of $\fk$ contains a cycle
\[((R_1,i_1), \dots ,(R_m,i_m), (R_1,i_1)),\]
 where $(R_1,i_1)\in P$. W.l.o.g. we may assume that  $R_1\notin\{R_2, \dots ,R_m\}$. 
Denote by $u_j$, $j\in [m]$, the unique term that occupies position $(R_j,i_j)$ in $q$. The unique $R_1$-atom of $q$ is then of the form $F=R_1(\underline{u_m},t_2, \dots ,t_{i_1-1},u_1,t_{i_1+1},\dots ,t_n)$.

Since vacuously $(R_1,i_1)\in \posclosure{P}{\fk}$, by Claim \ref{claim:calF} it holds that $\mu(u_1)\in \hat C$. Now, applying the chase restriction \eqref{it:restone}, and the assumption that  $q\models \fk$, a straightforward induction shows that $\mu(u_1)=\mu(u_2)=\dots =\mu(u_m)$. In particular, $\mu(F)$ must belong to $\db_{A,P}$, in which case the chase restriction \eqref{it:restone} entails that 
 $\mu(u_1) \neq \mu(u_m)$.  

\noindent
\textbf{Condition~\ref{it:no-const}.} Otherwise, assume that condition~\ref{it:no-const} is not true; i.e., some position $(S,j)\in \posclosure{P}{\fk}$ is occupied by a constant $d$ in $q$. Since $\mu$ fixes constants, we obtain by Claim \ref{claim:calF} that $d=\mu(d)\in \hat C$. This contradicts the fact that $\hat C$ does not intersect  $\queryconst{q}$. 

\noindent
\textbf{Condition~\ref{it:disjoint}.} Otherwise, assume that condition~\ref{it:disjoint} is not true; i.e, some positions $(S,j)\in \posclosure{P}{\fk}$ and $(T,k)\in \poscomplement{P}{\fk}$ are occupied by the same variable $x$ in $q$. In this case, Claim \ref{claim:calF} leads to an immediate contradiction.

\noindent
\textbf{Condition~\ref{it:distinct}.} Otherwise, assume that condition~\ref{it:distinct} is not true; i.e., there are two distinct non-primary-key positions $(R,i),(S,j)\in \posclosure{P}{\fk}$ that are occupied in $q$ by the same variable $x$.
 Since $\mu(x)\in \hat C$ by Claim \ref{claim:calF}, we find a constant from $\hat C$ that appears at both positions $(R,i)$ and $(S,j)$ in $\{A\} \cup \db_{A,P}$. 
 It is easy to see that this contradicts the chase restriction \eqref{it:resttwo} that applies in this case.

We observed that each case leads to a contradiction. 
 We conclude by contradiction that there cannot be a valuation $\mu$ such that $\mu(q)\subseteq \db\cup  \db_{A,P}$ and $\mu(q)\cap (\{A\}\cup \db_{A,P})\neq \emptyset$. In particular, no fact of $\{A\}\cup \db_{A,P}$ is relevant for $q$ in $\db \cup \db_{A,P}$.
\end{proof}

 \begin{example}
 Let $q=\{N(\underline{x},x), O(\underline{x},y)\}$ and $\fk=\{\foreignkey{N}{2}{N},\foreignkey{N}{2}{O}\}$. Consider a database $\db$ of the form
 $$
\db
=
\begin{array}{cc}
\begin{array}[t]{c|ll}
N & \underline{x}& x\bigstrut\\\cline{2-3}
  & a & a \\\cdashline{2-3}
  & b & c\\\cdashline{2-3}
\end{array}
&
\begin{array}[t]{c|ll}
O & \underline{x}& y\bigstrut\\\cline{2-3}
  & a & b \\\cdashline{2-3}
\end{array}
\end{array}
$$
Select $A=N(\underline{b},c)$, $P=\{(N,2)\}$, and observe that $\{(N,2)\}$ belongs to a cycle in the dependency graph, thus violating Theorem \ref{thm:syntactic-obedience}\ref{it:strongweak}.
As Lemma \ref{lem:gadget} predicts, we find  a database instance $\db_{A,P}$ satisfying all the items of the lemma statement:
 $$
\db_{A,P}
=
\begin{array}{cc}
\begin{array}[t]{c|ccc}
N & \underline{x}& x\bigstrut\\\cline{2-3}
  & c & \bot \\\cdashline{2-3}
  & \bot & c\\\cdashline{2-3}
\end{array}
&
\begin{array}[t]{c|ccc}
O & \underline{x}& y\bigstrut\\\cline{2-3}
  & c & \bot \\\cdashline{2-3}
    & \bot & c \\\cdashline{2-3}
\end{array}
\end{array}
$$
 \end{example}


We next turn to the proof of Theorem \ref{thm:syntactic-obedience}.  

\begin{proof}[Proof of Theorem \ref{thm:syntactic-obedience}]
It is straightforward to verify that the empty set of positions is obedient and satisfies all the items listed in Theorem \ref{thm:syntactic-obedience}. From here on, we assume that $P$ is non-empty.
We also assume that the unique $R$-atom of $q$ is of the form $F=R(\underline{\vec{s}},t_{k+1},\dots ,t_n)$, and define 
\[
q'\defeq \formula{q \setminus \fkclosure{P}{q}{\fk}} \cup \{F_{P}\},
\]
where $F_P$ is obtained from $F$ by substituting fresh variables for the terms occurring at positions of $P$ (see Definition \ref{def:obedience}).

\noindent
\framebox{$\implies$}  We show the contraposition. Assume that some of the conditions listed in Theorem \ref{thm:syntactic-obedience} is violated. 
Let $\theta$ be a one-to-one valuation mapping variables $x$ to constants $c_x$ (that are not from $\queryconst{q}$). 
We can then apply Lemma \ref{lem:gadget} to obtain a database instance $\db_{A,P}$ given $A\defeq \theta(F_P)$ and $\db\defeq\theta(q')$. 
The lemma states that every fact of $\{A\}\cup \db_{A,P}$ is irrelevant for $q$ in $\db\cup \db_{A,P}$.
This entails that no $R$-fact is relevant for $q$ in $\db\cup \db_{A,P}$, whence  $\db\cup \db_{A,P}\not\models q$. On the other hand, it is obvious that $\db\cup \db_{A,P}\models q'$. Finally, $\db\cup \db_{A,P}\models \fk$ follows by Lemma \ref{lem:gadget} and the fact that $\fk$ is about $q$. We thus conclude that $q'\nfkmodels{\fk} q$, i.e., $P$ is disobedient.

\noindent
\framebox{$\impliedby$} Let $F$ be the unique $R$-atom of $q$. Assuming conditions \ref{it:strongweak}--\ref{it:distinct} in Theorem~\ref{thm:syntactic-obedience} hold true, we show that $P$ is obedient, i.e.,  $q'\fkmodels{\fk}q$. 

Suppose $\db$ is a database that satisfies both $q'$ and $\fk$. We need to show that $\db$ satisfies also $q$. Let $\theta_0$ be a valuation such that $\theta_0(q')\subseteq \db$. In what follows, we will extend $\theta_0$ to a valuation $\theta$ such that $\theta(q)\subseteq \db$.

 Let $(G_1,\ldots ,G_m)$ list the atoms of $\fkclosure{P}{q}{\fk}$ in such an order that
 \begin{itemize}
 \item
  $G_1=F$, and
  \item for all $j\in [m-1]$ there is some $k\in [j]$ such that $\foreignkey{S_k}{l}{S_{j+1}}\in \fk$ for some integer $l$,
  \end{itemize}   
  where it is to be assumed that for each $h\in[m]$, $S_h$ is the relation name of $G_h$.

We show by induction that, for all $j \in [m]$, there exists a valuation $\theta_j$ over $\queryvars{q_j}$ such that $\theta_j(q_j)\subseteq \db$, where
   \[
   q_j \defeq \formula{q \setminus \fkclosure{P}{q}{\fk}}\cup \{G_1,\ldots ,G_j\}.
   \]

    For the base step suppose $j=1$. Denote by $P^c$ the set of positions $\{(R,i)\mid (R,i)\notin P, i\in [n]\}$. Concerning the positions over relation names appearing in $q_1=\formula{q \setminus \fkclosure{P}{q}{\fk}} \cup \{F\}$, let us make a few observations.
   First, we note that $P^c\subseteq  \poscomplement{P}{\fk}$, because otherwise some position of $P$ would belong to a cycle, 
   contradicting condition~\ref{it:strongweak}. 
   Second, every position of a relation name appearing in $q \setminus \fkclosure{P}{q}{\fk}$ belongs  to $ \poscomplement{P}{\fk}$ by definition.
    Third, it readily holds that ${P}\subseteq \posclosure{P}{\fk}$. We conclude that a position $(T,k)$ of a relation name $T$ that appears in $q_1$ belongs to $\posclosure{P}{\fk}$ if and only if it belongs to $P$. It follows by conditions~\ref{it:no-const}--\ref{it:distinct} that the positions of $P$ are occupied in $F$ by variables that are orphan in $q_1$. Clearly, we can extend $\theta_0$ to these orphan variables to obtain $\theta_1$ such that $\theta_0(F_P)=\theta_1(F)$. In particular, we obtain that $\theta_1(q_1)\subseteq \db$, where $q_1 = \formula{q \setminus \fkclosure{P}{q}{\fk}} \cup \{F\}$.
    
    For the induction step suppose $j\in [m-1]$. The induction claim is that  $\theta_{j+1}(q_{j+1})\subseteq \db$ for some valuation $\theta_{j+1}$ over $\queryvars{q_{j+1}}$, given the induction hypothesis that there is a valuation $\theta_j$ over $\queryvars{q_j}$ such that $\theta_j(q_j)\subseteq \db$.
 Let $ k\in [j]$ be such that $\foreignkey{S_k}{l}{S_{j+1}}\in \fk$ for some integer $l$. 
    Assuming $G_k=S_k(\underline{s_1, \dots ,s_a},s_{a+1} \dots ,s_{b})$, we can write $G_{j+1}=S_{j+1}(\underline{s_l},u_2, \dots ,u_c)$ since $\fk$ is about $q$.
    Since  $\theta_j(G_j)=S_k(\underline{\theta_j(s_1), \dots ,\theta_j(s_a)},\theta_j(s_{a+1})\dots ,\theta_j(s_b))\in \db$ and $\db\models \foreignkey{S_k}{l}{S_{j+1}}$, we find a fact $
    S_{j+1}(\underline{\theta_j(s_l)},b_2, \dots ,b_{c})\in \db.$
    Observe by condition~\ref{it:distinct} that $u_2, \dots, u_{c}$ are pairwise distinct variables. 
    Hence $\theta_{j+1} \defeq \theta_j \cup \{(u_i,b_i)\}_{i=2}^c$ is a well-defined valuation over $\queryvars{q_{j+1}}$ such that $\theta_{j+1}(q_{j+1})\subseteq \db$, if we can establish the following claim.
    \begin{claim}\label{claim:nointersection}
    $U\cap \queryvars{q_{j}}=\emptyset$, for $U\defeq\{u_2, \dots, u_{c}\}$.
    \end{claim}
    \begin{proof}[Proof of Claim \ref{claim:nointersection}]
    Let $P'\defeq \{(S_{j+1},2), \dots ,(S_{j+1},c)\}$.
    Let us first turn attention to $q_1=\formula{q \setminus \fkclosure{P}{q}{\fk}} \cup \{F\}$.
    Recall that $ \poscomplement{P}{\fk}$ contains  ${P^c} $ as well as the positions of relation names appearing in $q \setminus \fkclosure{P}{q}{\fk}$. Since $P' \subseteq { \posclosure{P}{\fk}}$, it follows by condition~\ref{it:disjoint} that $U$ does not contain any variable that appears in $F$ at a position of ${P^c}$, nor does it contain any variable from $\queryvars{q \setminus \fkclosure{P}{q}{\fk}}$. Furthermore, it follows by condition~\ref{it:distinct} that $U$ does not contain any variable that appears in $F$ at a position of ${P}$.
     We thus obtain  that $U\cap \queryvars{q_1}=\emptyset $.
   
 For the sake of contradiction, suppose now the claim is false, i.e., $u_p\in \queryvars{q_j}$ for some $p\in \{2, \dots ,c\}$.
 Let $h\leq j$ be the smallest integer such that $u_p\in \queryvars{q_{h}}$.
 We may assume, by the previous paragraph, that $h>1$. Suppose $G_h=S_h(\underline{v_1},v_{2} \dots ,v_d)$. By construction of the sequence $(G_1, \dots ,G_m)$, and since $q$ is self-join free and $\fk$ is about $q$, the primary-key term $v_1$ of $G_h$ must appear in $G_{h'}$ 
 for some $h'< h$.
By minimality of $h$, it must be that $u_{p}\neq v_{1}$ and, consequently, $u_p$ occurs at a non-primary-key position in $S_h(\underline{v_1},v_{2} \dots ,v_d)$; i.e.,
 $u_p=v_{p'}$ for some  $p\in \{2, \dots ,d\}$. 
 But then $(S_h,p')$ and $(S_{j+1},p)$ are two distinct non-primary-key positions of $\posclosure{P}{\fk}$ that are occupied in $q$ by the same variable, contradicting condition~\ref{it:distinct}. We conclude by contradiction that the claim holds. 
 \end{proof}
 
Having concluded the induction proof, we
note that $\theta(q)\subseteq \db$ for $\theta\defeq \theta_n$. This concludes the proof of Theorem \ref{thm:syntactic-obedience}.
\end{proof}

A particular consequence of the previous proof is that $(q \setminus \{F\}) \cup \{F_{P}\} \nfkmodels{\fk} q $ if $P$ is disobedient over $\fk$ and $q$.

\section{Proofs for Section~\ref{sec:lhard}}\label{sect:appC}


The following proof of Lemma~\ref{lem:cutchase} goes through for foreign keys that need not be unary.
The following definition of (not necessarily unary)  foreign keys is standard.
Let $R$ be a relation name with arity~$n$, and $S$ an atom with signature $\signature{m}{k}$.
An \emph{(unrestricted) foreign key} is an expression $\foreignkey{R}{j_{1},j_{2},\ldots,j_{k}}{S}$ with $j_{1},j_{2},\dots,j_{k}$ distinct integers in $[n]$. 
Given a database instance $\db$,
an $R$-fact $R(a_{1},\dots,a_{n})$ in $\db$ is \emph{dangling} with respect to this foreign key if $\db$ contains no $S$-fact $S(\underline{b_{1},\dots,b_{k}},b_{k+1},\dots,b_{n})$ such that $a_{j_{1}}=b_{1}$,  $a_{j_{2}}=b_{2}$, \dots, $a_{j_{k-1}}=b_{k-1}$, and $a_{j_{k}}=b_{k}$.

\begin{proof}[Proof of Lemma~\ref{lem:cutchase}]
Suppose $q$ has a cyclic attack graph. Then, by~\cite[Lemma~3.6]{DBLP:journals/tods/KoutrisW17}, there are atoms $F$ and $G$ such that $F\attacks{q}G\attacks{q}F$. 
For two constants $a$ and $b$, define the following valuation $\bival{a}{b}$ over $\queryvars{q}$:
\[
\bival{a}{b}(x) = 
\begin{cases}
a & \textnormal{if $x\in \keyclosure{F}{q}\setminus\keyclosure{G}{q}$},\\
b &\textnormal{if $x\in\keyclosure{G}{q}\setminus\keyclosure{F}{q}$},\\
\bot & \textnormal{if $x\in \keyclosure{F}{q}\cap\keyclosure{G}{q}$},\\
(a,b) & \textnormal{if $x\in \queryvars{q}\setminus \formula{\keyclosure{F}{q}\cup \keyclosure{G}{q}}$}.
\end{cases}
\]
Let $R$, $S$ be two sets of ordered pairs of constants. Define 
\begin{align*}
\db_{R,S} :=  
& \{\bival{a}{b}(H)\mid H\in q\setminus \{F,G\}, (a,b) \in R\cup S\}\\
& \cup \{\bival{a}{b}(F)\mid  (a,b) \in R\}\\ 
& \cup \{\bival{a}{b}(G)\mid (a,b) \in  S\}.
\end{align*}
 
The following follows from the proof of~\cite[Lemma~4.3]{DBLP:journals/tods/KoutrisW17}:
\begin{itemize}
\item
$\db_{R,S}$ is consistent with respect to primary keys in $q\setminus \{F,G\}$; and
\item
$\certainty{q}{\pk}$ is $\L$-hard, and remains $\L$-hard when inputs are restricted to database instances that are equal to $\db_{R,S}$ for binary relations $R$ and $S$.
\end{itemize}
We claim that the following are equivalent for all binary relations $R$ and $S$:
\begin{enumerate}
\item\label{it:rspk}
$\db_{R,S}$ is a ``no''-instance of $\certainty{q}{\pk}$; and
\item\label{it:rspkfk}
$\db_{R,S}$ is a ``no''-instance of $\certainty{q}{\pk\cup \fk}$.
\end{enumerate}

\framebox{\ref{it:rspk}$\implies$\ref{it:rspkfk}} 
Let $\rep$ be a repair of $\db_{R,S}$ with respect to $\pk$ such that $\rep\not\models q$. Informally, we construct a repair $\rep'$ of $\db_{R,S}$ with respect to $\pk\cup\fk$ by closing each dangling fact of $\rep$ by a cycle that is long enough. Initialize $\rep'$ as $\rep$, and chase $\rep'$ by the following rule:
Whenever there is some fact $A\in\rep'$ that is dangling with respect to some foreign key $\foreignkey{H}{\vec{\jmath}}{H'}$ in $\fk$,
pick constants $a,b$ such that $A=\bival{a}{b}(H)$,
\begin{enumerate}
\item 
if $H'\in q\setminus\{F,G\}$, then add $\bival{a}{b}(H')$ to $\rep'$;
\item 
if $H'=F$, then add $\bival{a}{c}(F)$ to $\rep'$, where $c$ is a fresh constant; and
\item\label{it:addG}
if $H'=G$, then add $\bival{c}{b}(G)$ to $\rep'$, where $c$ is a fresh constant.
\end{enumerate}
We only make one exception to this rule. 
Suppose that, according to~\eqref{it:addG}, we should add to $\rep'$ a $G$-fact, say $\bival{e}{d}(G)$ with $e$ a fresh constant, while having already added $\bival{a}{b}(F)$, $\bival{c}{b}(G)$, and $\bival{c}{d}(F)$. Then,  instead of introducing a fresh value, we add $\bival{a}{d}(G)$. We deal symmetrically with additions of $F$-facts. It is now easy to see that the chase terminates, and that $\rep'$ is a repair with respect to $\pk\cup\fk$. 

Assume for the sake of contradiction that $\mu(q)\sub \rep'$ for some valuation $\mu$.  The attacks between~$F$ and $G$ imply that $\{\bival{a}{b}(F), \bival{a'}{b'}(G)\}\sub \mu(q)$ if and only if $a=a'$ and $b=b'$.
Thus no added $F$-fact or $G$-fact is in $\mu(q)$, and hence we find constants $a,b$ such that $\{\bival{a}{b}(F), \bival{a}{b}(G)\}\sub \mu(q) \cap \rep$. 
Moreover, $\db_{R,S}$ is consistent with respect to primary keys in $q\setminus \{F,G\}$, and thus by construction, $\bival{a}{b}(q\setminus \{F,G\}) \sub \rep $. We obtain $\bival{a}{b}(q)\sub\rep $, hence $\rep\models q$, a contradiction. 
We conclude by contradiction that $\rep'$ does not satisfy~$q$.

\framebox{\ref{it:rspkfk}$\implies$\ref{it:rspk}} 
Let $\rep$ be a repair of $\db_{R,S}$ with respect to $\pk\cup\fk$ such that $\rep\not\models q$. 
Note that $\rep$ need not be a repair of $\db_{R,S}$ with respect to $\pk$, because
\begin{itemize}
\item some facts of $\rep$ may not belong to $\db_{R,S}$; or
\item some blocks of $\db_{R,S}$ may be disjoint with $\rep$.
\end{itemize}
Let $\sep$ be a $\subseteq$-minimal database instance such that
\begin{itemize}
\item 
$\rep\cap\db_{R,S}\subseteq\sep$; and
\item 
for every block $\block$ of $\db_{R,S}$ such that $\rep\cap\db_{R,S}=\emptyset$, $\sep$ contains a fact arbitrarily picked from $\block$.
\end{itemize}
By construction, $\sep\subseteq\db$.
It is easily verified that $\sep$ is a repair of $\db_{R,S}$ with respect to $\pk$.
Note incidentally that since $\sep\closer{\db_{R,S}}\rep$ is easily verified, it must hold that either $\sep=\rep$ or $\sep\not\models\fk$.

It suffices to show that $\sep$ falsifies $q$.
Suppose for the sake of contradiction that $\sep\models q$.
Then, we can assume a valuation $\theta$ such that $\theta(q)\subseteq\sep$, and therefore $\theta(q)\subseteq\db$. 
Since $\theta(q)\nsubseteq\rep$, there is a fact $A\in\theta(q)\setminus\rep$.
Moreover, from the construction of $\sep$, it follows that $\rep\cap\db$ contains no fact that is key-equal to $A$.
Then, by Lemma~\ref{lem:norepair}, $\rep$ is not a repair, a contradiction.
\end{proof}
\section{Proofs for Section~\ref{sec:nl}}\label{sect:apD}

\subsection{Preliminaries}

Before proceeding with the proof of Lemma~\ref{lem:bi}, we consider some useful auxiliary concepts. 
A database instance $\rep$ is \emph{irrelevantly dangling} if, using Lemma \ref{lem:gadget}, it can be extended to a consistent database instance $\rep'$ in such a way that every fact that is dangling in $\rep$ is irrelevant in $\rep'$.




\begin{definition}[Irrelevantly dangling instance]\label{def:irr}
Let $q$ be in $\sjfbcq$. Let   $\fk$ be a set of foreign keys about $q$.
 Let  $\db$ be a database instance. A database instance $\rep$ is \emph{irrelevantly dangling (with respect to $(\db,\fk,q)$)} if for all $R(\underline{\vec{a}},b_{k+1}, \dots ,b_n)$ that are dangling in $\rep$ with respect to a foreign key $\foreignkey{R}{j}{S}\in {\fk}$, it holds that:
 \begin{enumerate}[label=(\arabic*)]
\item\label{it:irr1}  $P$ is not obedient over ${\fk}$ and~$q$; and
 \item\label{it:irr2} $(R,j)\in P$;
 \end{enumerate}
 where $P$ is the set of all non-primary-key positions $(R,i)$ such that $b_i$ is orphan in $\rep\cup\db$ and does not belong to $\queryconst{q}$
 \end{definition}
 Note that \ref{it:irr2} entails that $\foreignkey{R}{j}{S}$ is strong.

Let $\db$ be a database instance.
Whenever $\rep$, $\sep$ are database instances,
we write $\rep\capcloserneq{\db}\sep$ if $\rep\closer{\db}\sep$ and $\sep\cap\db\subsetneq\rep\cap\db$.
Note that $\db\setminus \rep\subsetneq\db\setminus \sep$ if and only if $\sep\cap \db\subsetneq\rep\cap \db$. 
It is straightforward to verify that 
$\capcloserneq{\db}$ is a strict partial order.

\begin{definition}[Pre-repair]\label{def:pre-repair}
Let $q$ be in $\sjfbcq$. Let   $\fk$ be a set of foreign keys about $q$.
We say that $\rep$ is a \emph{pre-repair} of a database $\db$ (over $\fk$ and $q$) if $\rep$ is a $\capcloserneq{\db}$-minimal database satisfying the following conditions: 
\begin{enumerate}[label=(\arabic*)]
\item\label{item:pre-one} $\rep\models\pk$; and
\item\label{item:pre-two} $\rep$ is irrelevantly dangling with respect to $(\db,{\fk},q)$.  
\end{enumerate}
\end{definition}

Recall that $\pk$ denotes the set of primary keys underlying $q$.
By $\capcloserneq{\db}$-minimality of $\rep$ we mean that
there is no database instance $\sep$ satisfying conditions~\ref{item:pre-one} and~\ref{item:pre-two} such that $\sep\capcloserneq{\db}\rep$.

The following simple lemma states that any \emph{consistent} pre-repair subsumes a repair. Thus, to provide a repair that falsifies a query $q$, we only need to look for consistent pre-repairs that do the same. 
\begin{lemma}\label{lem:rep-pre}
Let $q$ be a query in $\sjfbcq$. 
Let $\fk$ be a set of unary foreign keys about $q$. Let $\db$ be a database.
If $\rep$ is a pre-repair of $\db$ over $\fk$ and $q$ such that $\rep\models \fk$, then there exists a repair $\sep$ of $\db$ over $\fk$ such that $\sep\subseteq \rep$.
\end{lemma}
\begin{proof}

Let $\rep$ be a pre-repair of $\db$ over $\fk$ and $q$ such that  $\rep \models  \fk$. Then, there exists a repair
 $\sep$ of $\db$ such that $\sep\closer{\db}\rep$. We claim that $\sep\subseteq \rep$. 
Assume toward contradiction that $\sep\subsetneq \rep$. 
Then, it must be the case that $\rep\cap \db\subsetneq\sep\cap \db$, whence $\sep \capcloserneq{\db} \rep$. Since $\sep$ is a repair, it in particular satisfies items~\ref{item:pre-one} and~\ref{item:pre-two} of Definition \ref{def:pre-repair}. Consequently, $\rep$ cannot be a pre-repair, which contradicts the assumption. We conclude by contradiction that the claim holds.
\end{proof}

\begin{theorem}\label{thm:pre-repair}
Let $q$ be in $\sjfbcq$. Let   $\fk$ be a set of foreign keys about $q$.
 Then, every repair of $\db$ over $\fk$ satisfies $q$ if and only if every pre-repair of $\db$ over $\fk$ and $q$ satisfies $q$. 
\end{theorem}

\begin{proof}
\framebox{$\impliedby$} We show the contraposition. Let $\rep$ be a repair of $\db$ over $\fk$ that does not satisfy $q$. In particular, $\rep$ satisfies items~\ref{item:pre-one} and~\ref{item:pre-two} of Definition \ref{def:pre-repair}. Then, either $\rep$ is is a pre-repair, or there exists a pre-repair $\sep$ such that $\sep \capcloserneq{\db} \rep$. We only need to consider the latter option, with respect to which we claim that $\sep\not\models q$. Assume toward contradiction that this is not true. Let $\theta$ be a valuation such that $\theta(q)\subseteq \sep$.
Since  $\rep\not\models q$, we find a fact $A\in \theta(q)\setminus \rep$. By $\sep\models \pk$ and $\rep\cap\db\subsetneq \sep\cap\db$ we observe that $\rep\cap\db$ does not contain any fact that is key-equal with $A$.
Furthermore, we have $\sep  \subseteq\db\cup\rep$, whence $\theta(q)\subseteq \db\cup\rep$.
Thus, by Lemma \ref{lem:norepair}, $\rep$ is not a repair, which contradicts the assumption. We conclude by contradiction that $\sep\not\models q$.

\framebox{$\implies$}
We show the contraposition. Let $\rep_0$ be a pre-repair of $\db$ over $\fk$ and $q$ such that $\rep_0\not\models q$. We need to construct
 a repair $\sep$ of $\db$ over $\fk$ such that $\sep\not\models q$. 
For this, let us first show how each dangling fact can be made non-dangling using Lemma \ref{lem:gadget}

Suppose $A=R(\underline{\vec{a}},b_{k+1}, \dots ,b_{n})\in \rep_0$ is dangling with respect to $\foreignkey{R}{j}{S} \in {\fk}$ in $\rep_0$. 
Let $P$ be the set of all non-primary-key positions $(R,i)$ such that $b_i$ is orphan in $\rep_0\cup\db$ and does not belong to $\queryconst{q}$. 
Since $\rep_0$ is a pre-repair, we observe that $(R,j)$ belongs to  $P$, which in turn is not obedient over $\fk$ and $q$.
Hence, there exists a database instance $\db_{A,P}$ that satisfies the statement of Lemma \ref{lem:gadget}.
 We now show the following claim.
\begin{claim}\label{claim:nodangling}
  $\rep_1\defeq\rep_0\cup \db_{A,P}$ is a pre-repair of $\db$ over $\fk$ and $q$  such that $\rep_1\not\models q$. In particular, $A$ is not dangling in $\rep_1$ with respect to ${\fk}$.
\end{claim}
 \begin{proof}[Proof of Claim \ref{claim:nodangling}]
 Lemma \ref{lem:gadget}\eqref{it:gadget3}  entails that every fact of $\db_{A,P}$ is irrelelevant for $q$ in $\rep_1$ (and thus in $\rep_0$). Since $\rep_0 \not\models q$, we  obtain that $\rep_1 \not\models q$. Using Lemma \ref{lem:gadget}\eqref{it:gadget4} and the assumption that $\rep_0$ is a pre-repair it is also easy to see that $A$ is not dangling in $\rep_1$ with respect to ${\fk}$.
  
 It remains to show that $\rep_1$ is a pre-repair. For this, we observe first 
 by Lemma \ref{lem:gadget}\eqref{it:gadget1} that
$\keyconst{\db}\cap \adom{\db_{A,P}}=\emptyset$, whence, in particular, $\keyconst{\rep_0}\cap \keyconst{\db_{A,P}}=\emptyset$.
Moreover, we have  $\db_{A,P}\models\pk$ by Lemma \ref{lem:gadget}\eqref{it:gadget1}, and $\rep_0\models \pk$ by virtue of $\rep_0$ being a pre-repair.  We may thus conclude that $\rep_1\models \pk$.

Lastly, we need to show that $\rep_1$ is irrelevantly dangling. Suppose some fact $A'=R'(\vec{a}',b'_{k'+1}, \dots ,b'_{n'})\in \rep_1$ is dangling in $\rep_1$ with respect to $\foreignkey{R'}{i'}{S'} \in {\fk}$. 
Note that $A'$ must then belong to $\rep_0$ and  be dangling (in $\rep_0$) with respect to $\foreignkey{R'}{i'}{S'}$. This follows by Lemma \ref{lem:gadget}\eqref{it:gadget2}, which states that $\db_{A,P}\models{\fk}$.
For $i\in \{0,1\}$, define $\db_i\defeq \rep_i\cup\db$, and 
define $P_i$ as the set of non-primary-key positions of $(R',j)$ such that $b'_j$ is orphan in $\db_i$ and does not belong to $\queryconst{q}$.  
Since $\rep_0$ is a pre-repair of $\db$, we know that $(R',i')\in P_0$, and that $P_0$ is not obedient over $\fk$ and $q$.
 It thus suffices to show that $P_0=P_1$. This boils down to showing that $P_0\subseteq P_1$, as $P_1\subseteq P_0$ by definition.
 
At this point, we observe that $A\neq A'$. For, assuming this is not the case, we obtain that $A$ is dangling in $\{A\}\cup \db_{A,P}$ for a foreign key of ${\fk}$ outgoing position $(R',i')\in P_0$. On the other hand, by Lemma \ref{lem:gadget}\eqref{it:gadget4}, $A$ cannot be dangling in $\{A\}\cup \db_{A,P}$ for any foreign key of ${\fk}$ outgoing a position  of $P$. Hence, we obtain a contradiction by $P=P_0$, which holds in this case.

Suppose $(R',j)\in P_0$. Then $b'_j$ is orphan in $\db_0$ and does not belong to $\queryconst{q}$. Define $C\defeq \{b_i\mid (R,i)\in P\}$.
Since $A\neq A'$, and since $C$ consists of orphan constants of $\rep_0$,  we have $b'_j\notin C$. Since $\adom{\rep_0}\cap\adom{\db_{A,P}}\subseteq C$ by Lemma \ref{lem:gadget}\eqref{it:gadget1.5}, it follows that $b'_j$ is also orphan in $\db_1$.
Hence $P_0\subseteq P_1$, and therefore we have shown 
 that $\rep_1$ is irrelevantly dangling and thus a pre-repair. This concludes the proof of Claim \ref{claim:nodangling}.
\end{proof}
 It is now easy to see that repeated application of Claim \ref{claim:nodangling} yields a pre-repair $\sep'$ of $\db_0$ such that $\sep'\models \fk$ and $\sep\not\models q$. By Lemma \ref{lem:rep-pre} we thus find a repair $\sep$ of $\db_0$ such that $\sep\not\models q$. This concludes the proof of Theorem \ref{thm:pre-repair}.
\end{proof}

\subsection{Proof of Lemma~\ref{lem:bi}}\label{sect:proofnl}

\begin{proof}[Proof of Lemma~\ref{lem:bi}]
$\problem{REACHABILITY}$ is the following problem: Given a directed graph $\calG_0$ and its two vertices $s$ and $t$, is there a directed path from~$s$ to~$t$?
$\problem{REACHABILITY}$ is $\NL$-complete, and remains $\NL$-complete
when the inputs are acyclic graphs.
We give a first-order reduction from $\problem{REACHABILITY}$ to the complement of $\certainty{q}{\fk}$.
This proves the lemma, because $\NL$ is closed under complement.

Let $\calG_0$ be a directed acyclic graph, and let $s$ and $t$ be its two vertices. Let $\calG=(V,E)$ be obtained from $\calG_0$ by adding an edge from $t$ to $s$.
Obviously, there is a first-order reduction from $\calG_0$ to $\calG$, and $s$ is connected to $t$ in $\calG_0$ if and only if the same holds in $\calG$. Since first-order reductions are closed under composition, it suffices to construct a  first-order reduction from $\calG$ to a database $\db$ such that $s$ is connected to $t$ in $\calG$ if and only if $\db$ is a ``no"-instance of $\certainty{q}{\fk}$. 

Assume that some foreign key $\foreignkey{N}{j}{O}\in \lclosure{\fk}$ is block-interfering.
Let the unique $N$-atom of $q$ be 
\[
F=N(\underline{t_1, \dots ,t_k},t_{k+1},\dots,t_{n}).
\] 

Let $y=t_j$. The unique $O$-atom of $q$ is of the form $G\defeq O(\underline{y},\vec{z})$ for some sequence of distinct variables $\vec{z}$. Moreover, $y$ is a variable such that $\FD{q}\not\models\fd{\emptyset}{\{y\}}$.

Let
\begin{align}
C &=\{z\in\queryvars{q}\mid\FD{q}\models\fd{\emptyset}{\{z\}}\}; \mbox{and}\label{eq:C}\\
q_{0} &=q\setminus\fkclosure{O}{q}{\fk}.\nonumber
\end{align}

We consider cases \eqref{it:marking} and \eqref{it:bifour} of Definition \ref{def:block-int} simultaneously. Define $P_a\defeq \{(N,k+1), \ldots ,(N,n)\}\setminus\{(N,j)\}$ and $P_b\defeq \emptyset$ for cases \eqref{it:marking} and \eqref{it:bifour}, respectively. Let $e\in \{a,b\}$. 

Let $c$ be a fresh constant.
For every vertex $u\in V$, let $\theta_{u}$ be a valuation over $\queryvars{q\cup\{F_{P_e}\}}$ such that 
$$
\theta_{u}(z)=
\begin{cases}
c & \mbox{if $z\in C$};\\
c_{z,u}  & \mbox{otherwise}.
\end{cases}
$$

for every variable $z$ to $c_{z,u}$
where $c_{z,u}$ denotes a fresh constant which depends on (and only on) $z$ and $u$.
That is, $c_{z,u}=c_{z',u'}$ if and only if $z=z'$ and $u=u'$.

If $J$ is an $n$-ary atom, $u$ a term, and $i\leq n$, we write $\substitute{J}{i}{u}$ for the atom obtained from $J$ by replacing its $i$th term with $u$.

Construct a database instance $\db$ as follows:
\begin{itemize}
\item
$\db$ includes $\theta_{s}(q)$;
\item
for every vertex $u\in V\setminus\{s\}$, $\db$ includes $\theta_{u}(q)\setminus  \{\theta_u(G)\}$; and
\item
for every edge $(u,v)$ in $E$, $\db$ contains a fact 
\begin{equation}\label{eq:auv}
A_{u,v}=N(\underline{a_1, \dots ,a_k},a_{k+1},\dots,a_{n}),
\end{equation}
where 
\[
a_i = 
\begin{cases}
c_{u,v,i}&\text{if }(R,i)\in P_e,\\
\theta_v(t_i)&\text{if } i = j,\\
\theta_{u}(t_i) &\text{otherwise},
\end{cases}
\]
where $c_{u,v,i}$ is a fresh constant. 
\end{itemize}
In particular, $c_{u,v,i}$ is a constant that is orphan in $\db$ and does not belong to $\queryconst{q}$.
%

\tikzstyle{edge}=[->,shorten >=.5pt,>=stealth,thick]
\tikzstyle{vertex}=[circle,draw=black!50,fill=blue!10,thick,inner sep=0pt,minimum size=4mm]

\begin{figure}
\begin{center}
\begin{tabular}{cc}
\begin{adjustbox}{valign=t}
\begin{tikzpicture}
\node (1) [vertex] {$s$};
\node (2) [vertex, above right =10mm of 1] {$1$};
\node (3) [vertex, below right  =10mm of 1] {$2$};
\node (4) [vertex, below right  =10mm of 2] {$t$};

\draw [edge] (1) -- (2);
\draw [edge] (1) -- (3);
\draw [edge] (3) -- (4);
\draw [edge] (4) -- (1);
\end{tikzpicture}
\end{adjustbox}

&
$
\begin{array}[t]{cc|c}
&F & \\\cline{3-3}
& & \theta_s(F)\\
& & A_{s,1} \\
& & A_{s,2} \\
\cdashline{3-3}
\mapsto & & \theta_1(F)\\
\cdashline{3-3}
& & \theta_2(F)\\
& & A_{2,t}\\
\cdashline{3-3}
& & \theta_t(F)\\
& & A_{t,s} \\\cdashline{3-3}
\end{array}
$
\,
$
\begin{array}[t]{c|c}
G & \\
\cline{2-2}
  & \theta_s(G)\\\cdashline{2-2}
\end{array}
$
\,
$
\begin{array}[t]{c|c}
H_1 & \\
\cline{2-2}
  & \theta_s(H_1)\\
  \cdashline{2-2}
    & \theta_1(H_1)\\
  \cdashline{2-2}
  & \theta_2(H_1)\\
  \cdashline{2-2}
  & \theta_t(H_1)\\\cdashline{2-2}
\end{array}
$
\,
$\dots$
$
\begin{array}[t]{c|c}
H_m & \\
\cline{2-2}
  & \theta_s(H_m)\\
  \cdashline{2-2}
    & \theta_1(H_m)\\
  \cdashline{2-2}
  & \theta_2(H_m)\\
  \cdashline{2-2}
  & \theta_t(H_m)\\\cdashline{2-2}
\end{array}
$
\end{tabular}
\end{center}
\caption{Reduction from graph reachability over $(\calG,s,t)$ to database instance $\db$ with respect to $q=\{F,G,H_1, \dots ,H_m\}$}
\label{fig:reach2}
\end{figure}
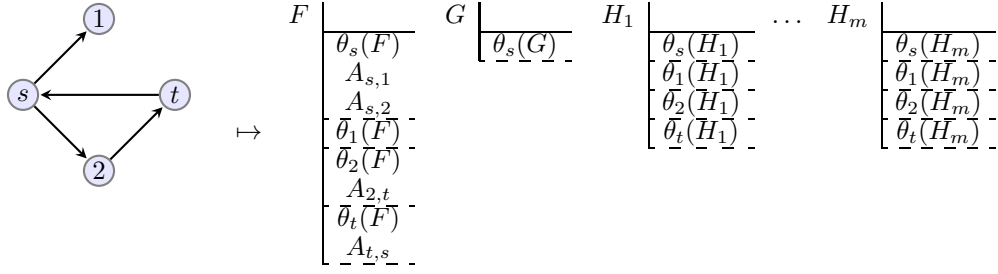

Note that a fact of the form $A_{u,v}$ is key-equal with others of the form $A_{u,w}$ and $\theta_u(F)$.
It is easy to see that primary key violations arise only because of facts of the form $A_{u,v}$.
For this, suppose $\rep$ contains key-equal facts $\theta_u(H)$ and $\theta_v(H)$. Then, $\keyvars{H}$ must be a subset of $C$, that is, $\FD{q}\models \keyvars{H}$. Since $\FD{q}$ contains $\fd{\keyvars{H}}{\atomvars{H}}$, it follows that $\atomvars{H}$ is likewise a subset of $C$, whence $\theta_u(H)$ and $\theta_v(H)$ must be identical and there is no primary key violation.

For instance, suppose $F$ is of the form $N(\underline{x},y,z)$, and let $j=2$. Then, $\theta_u(F) = N(\underline{c_{x,u}},c_{y,u},c_{z,u})$,
and $A_{u,v}$ is either the form
 $N(\underline{c_{x,u}},c_{y,v},c_{z,u})$ or $N(\underline{c_{x,u}},c_{y,v},c_{u,v,2})$, for $e=a$ and $e=b$ respectively. 
 
The construction of the database instance $\db$ is illustrated in Fig.~\ref{fig:reach2} for the graph that appends the acyclic example graph of Fig.~\ref{fig:reach} with an extra edge from $t$ to $s$.

We claim that $\db$ is a ``no"-instance of $\certainty{q}{\fk}$ if and only if there is a path from $s$ to $t$ in $\calG$. 

\noindent
\framebox{$\impliedby$}
Suppose there is a path from $s$ to $t$. 
 Let $U$ consist of all vertices $u\in V$ such that there is a path from $u$ to $t$. We construct a pre-repair $\rep$ of $\db$ over $q$ that does not satisfy $q$. First, let $\sep$ be a database instance such that
 for each vertex $u\in U$,
\begin{itemize}
\item $\sep$ includes $\theta_{u}(q)\setminus  \{\theta_u(F)\}$; and
\item $\sep$ takes from $\db$ exactly one fact of the form $A_{u,v}$, where it is required that $v\in U$.
\end{itemize}
Observe that $\sep \setminus \db$ consists of facts of the form $\theta_u(G)$ where $u\in U\setminus \{s\}$.
Moreover, by the assumption that there is an edge from $t$ to $s$, we have included $A_{t,s}$ in $\sep$. See Fig. \ref{fig:reach3} for an illustration of $\sep$ with regards to the  example case in Fig. \ref{fig:reach2}.

\begin{figure}
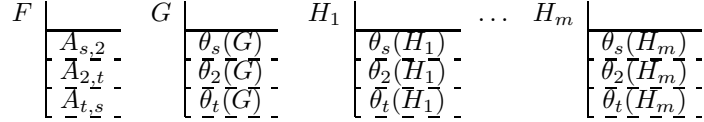

\begin{center}
$
\begin{array}[t]{c|c}
F & \\\cline{2-2}
 & A_{s,2} \\
\cdashline{2-2}
 & A_{2,t}\\
\cdashline{2-2}
 & A_{t,s}\\\cdashline{2-2}
\end{array}
$
\,
$
\begin{array}[t]{c|c}
G & \\
\cline{2-2}
  & \theta_s(G)\\
  \cdashline{2-2}
    & {\theta_2(G)}\\
    \cdashline{2-2}
  & {\theta_t(G)}\\\cdashline{2-2}
\end{array}
$
\,
$
\begin{array}[t]{c|c}
H_1 & \\
\cline{2-2}
  & \theta_s(H_1)\\
  \cdashline{2-2}
  & \theta_2(H_1)\\
  \cdashline{2-2}
  & \theta_t(H_1)\\\cdashline{2-2}
\end{array}
$
\,
$\dots$
$
\begin{array}[t]{c|c}
H_m & \\
\cline{2-2}
  & \theta_s(H_m)\\
  \cdashline{2-2}
  & \theta_2(H_m)\\
  \cdashline{2-2}
  & \theta_t(H_m)\\\cdashline{2-2}
\end{array}
$
\end{center}
\caption{Database instance $\sep$ with respect to Fig.~\ref{fig:reach2}.}
\label{fig:reach3}
\end{figure}

Obviously $\sep$ is consistent with $\pk$. We claim that is it also
  irrelevantly dangling for $(\db,\fk,q)$. For this, using the fact that  $q\models \fk$, in Case \eqref{it:bifour} there are no facts in $ \sep$ that are dangling with respect to ${\fk}$. In Case \eqref{it:marking}, suppose $B\in \sep$ is dangling for $\sigma\in {\fk}$. Then, using the fact that  $q\models \fk$, we observe that $B$ is of the form $A_{u,v}=N(\underline{\vec{a}},b_{k+1}, \dots ,b_n)$ and $\sigma$ of the form $\foreignkey{N}{i}{S}$ where $(N,i)\in P_a$. Note that $P_a$ is the set of all non-primary-key positions $(N,i)$ such that $b_i$ is orphan in $\sep\cup \db$ and does not belong to $\queryconst{q}$. Moreover, $P_a$ is not obedient. The claim thus follows.

We may now take a pre-repair $\rep$ such that $\rep \capcloserneq{\db} \sep$. By Theorem \ref{thm:pre-repair} it suffices to show that $\rep$ does not satisfy $q$. For this, note that any $N$-fact in $ \rep$ is either of the form $\theta_u(F)$ for $u\in V\setminus U$, or of the form $A_{u,v}$. 
 Assume toward contradiction that $\gamma$ is a valuation such that $\gamma(q)\subseteq \rep$.

Assume first that $\gamma(F)$ is of the form $\theta_u(F)$ for  $u\in V\setminus U$.  Then, $\theta_u(F)$ is dangling in $\rep$ for $\foreignkey{N}{j}{O}$, for neither $\sep$ nor $\db$ contains an $O$-fact with a matching primary-key value. Since $\rep \subseteq\sep \cup  \db$, it follows that $\gamma(F)$ is dangling in $\rep$. However, by consistency of $q$ this is not possible, a contradiction.

Suppose then $\gamma(F)$ is of the form $A_{u,v}$. 
 We consider cases \eqref{it:marking} and \eqref{it:bifour} separately:
 
 \noindent
\textbf{Case \eqref{it:marking}.} Recall that $P_a$ is not obedient. It follows by Theorem \ref{thm:syntactic-obedience} that for some $(N,i)\in P_a$, $t_i$ is not an orphan variable of $q$, that is, $t_i$ is a constant or a variable that occurs also elsewhere in $q$.
However, the $i$th constant in $A_{u,v}$ is fresh; in particular, this constant does not appear in ${q}$ or elsewhere in ${\rep}$. 
This leads to a contradiction with the assumptions that $\gamma(F)=A_{u,v}$ and $\gamma(q)\subseteq \rep$.

\noindent
\textbf{Case \eqref{it:bifour}.} For some $i\in [l]$, we have that $t_i$ and $t_j$ are variables that are connected in $\constgraph{V'}{q'}$, where $q'\defeq q\setminus\{F\}$ and $V'\defeq \queryvars{q}\setminus C$, for the set $C$ defined in \eqref{eq:C}. Denote $t_i$ by $x$, and recall that we have earlier denoted $t_j$ by $y$.
Since $x$ is not from $C$, 
we observe that 
\begin{equation}\label{eq:diff}
\gamma(x)=a_i=\theta_u(x)=c_{x,u},
\end{equation}
for $A_{u,v}$ of the form \eqref{eq:auv}.

Let $(z_1, \ldots ,z_n)$ be a path in $\constgraph{V'}{q'}$, where $z_1=x$ and $z_n=y$. Then, we find $H_1, \ldots ,H_{n-1}\in q'$, over relation names $R_1, \ldots ,R_{n-1}$, such that $\{z_i,z_{i+1}\}\subseteq \atomvars{H_i}\cap V'$, for $i\in [n-1]$.

Let $i\in [n-1]$. Note that each $R_i$-fact in $\sep$ is of the form $\theta_w(H_i)$ for a vertex $w$. Since all $R_i$-facts in $\rep\setminus \sep$ must be from $\db$, the same holds for $\rep$. We now show by induction that $\gamma(H_i)=\theta_u(H_i)$, where $u$ is the vertex that appears in the subscript of  $A_{u,v}$.
The base step follows by \eqref{eq:diff}, since $x\in \atomvars{H_1}$. For the induction step, suppose $\gamma(H_i)=\theta_u(H_i)$. Using this and the fact that $z_{i+1}\in \atomvars{H_i}\cap V'$, we first obtain that $\gamma(z_{i+1})=\theta_u(z_{i+1})$. Since $z_{i+1}\in \atomvars{H_{i+1}}\cap V'$, it thus must hold that $\theta_u(z_{i+1})=c_{z_{i+1},u}$, whence $\gamma(H_{i+1})=\theta_u(H_{i+1})$. This concludes the induction step and the induction proof.

We have now established that $\gamma(H_{n-1})=\theta_u(H_{n-1})$. Since $y\in \atomvars{H_{n-1}}\cap V'$,  we obtain that $\gamma(y)=\theta_u(y)=c_{y,u}$. On the other hand, since $y$ is not from $C$, we also have by \eqref{eq:auv} that $\gamma(y)= a_j= \theta_{v}(y)=c_{y,v}$. The vertices $u$ and $v$ cannot be identical since there are no self-loops in $G$ (unless $s=t$, which can be ruled out w.l.o.g.). Hence we obtain a contradiction by $c_{y,u}\neq c_{y,v}$. 

We conclude by contradiction that $\gamma(F)$ cannot be of the form $A_{u,v}$ either. In particular, we have shown that $\gamma(F)\notin \rep$, which contradicts the assumption that $\gamma(q)\subseteq \rep$. Hence $\rep$ does not satisfy  $q$. This concludes the direction from right to left.

\noindent
\framebox{$\implies$}
Suppose there is no path from $s$ to $t$ in $\calG$. We show that every repair of $\db$ satisfies $q$. Assume toward contradiction that $\rep$ is a repair of $q$ that does not satisfy $q$. We claim that $s$ is then the starting point of an infinite path in $\calG$. Recall that $\calG$ extends a directed acyclic graph with an edge from $t$ to $s$. Hence, it suffices to prove the claim, as it entails that there is a path from $s$ to $t$, thus contradicting the assumption.

First, we observe by Corollary \ref{cor:relevantblockbis} 
  that $\theta_s(G)\in \rep$, since $\{\theta_s(G)\}\in \theblock{\db}{S}$ is a  block that is relevant for $q$ in $\db$. Assume we have constructed a path $(u_0, \ldots ,u_i)$ in $\calG$, where $u_0=s$ and $A\in \rep$ for some $A$ such that $A \keyeq \theta_{u_i}(G)$. We show how to extend this path by $u_{i+1}$ such that $B\in \rep$ for some $B$ such that $B \keyeq \theta_{u_{i+1}}(G)$. The claim that there is an infinite path starting from $s$ follows from this.

We first claim that $\rep$ must subsume $\theta_{u_i}(q\setminus \{F,G\})$ and contain some fact from $\db$ that is key-equal with $\theta_{u_i}(F)$. Suppose this were not the case. Then, build a consistent database instance as follows. Extend first $\rep$ with $\theta_{u_i}(q\setminus \{F,G\})$, and then remove any fact of $\rep$ that is key-equal with some distinct fact in $\theta_{u_i}(q\setminus \{F,G\})$; in particular, the removed facts cannot be from $\db$. Furthermore, if the obtained database instance does not contain any fact from $\db$ that is key-equal with $\theta_{u_i}(F)$, then add $\theta_{u_i}(F)$ and remove its possible key-equal fact that is not in $\db$. Denote the obtained database instance by $\sep$. Note that no primary-key value has been removed, and no primary-key violation has been introduced. Recall also that $q\models \fk$, and that
 $\sep$ contains a fact $A$ such that $A \keyeq \theta_{u_i}(G)$. Hence we obtain that $\sep\models \pk\cup\fk$. Furthermore, since we removed only facts that were not in $\db$, and added a non-empty set of facts from $\db$, it follows that $\sep \closerneq{\db} \rep$, a contradiction with the assumption that $\rep$ is a repair. Hence the claim. 

Suppose now $\theta_{u_i}(F)\in \rep$. Then, $\rep$ subsumes $\theta_{u_i}(q\setminus \{G\})$ and contains some $A$ such that $A \keyeq \theta_{u_i}(G)$.
Recall that $G$ is an obedient atom of the form $O(\underline{y},\vec{z})$, meaning that 
$
q' \fkmodels{\fk} q
$
 for  $q'= (q \setminus \fkclosure{O}{q}{\fk}) \cup \{O(\underline{y},\vec{v})\}$,
where $\vec{v}$ is a sequence of distinct fresh variables. We can now easily construct  from $\theta_{u_i}$ and $A$ a valuation $\gamma$ such that  $\gamma(q')\subseteq \rep$. Having $\rep\models \fk$ and $\rep \models q'$, we thus obtain  $\rep \models q$. This however contradicts the assumption made in the beginning. We conclude that $\theta_{u_i}(F)\notin \rep$.  As shown in the previous paragraph,  $\rep$ must contain some fact from $\db$ that is key-equal with $\theta_{u_i}(F)$.
Hence $\calG$ must contain an edge $(u_i,v)$ such that some fact of the form $A_{u_i,v}\in\db$ belongs to $\rep$.

Now, $\rep$ contains an $N$-fact whose $j$th position is occupied by $c_{y,v}$. This is only possible if $\rep$ contains an $O$-fact with primary-key value $c_{y,v}$. In other words, $\rep$ must contain a fact that is key-equal with $\theta_v(G)$. Hence, setting $u_{i+1}$ as $v$, we have extended the path $(u_0, \ldots ,u_i)$ with a vertex $u_{i+1}$ such that $B\in \rep$ for some $B$ such that $B \keyeq \theta_{u_{i+1}}(G)$. We conclude that there is an infinite path starting from $s$, which leads to a contradiction with our assumption.
This concludes the direction from left to right. 

We conclude the proof by noting that the the reduction from $\calG$ to $\db$ is clearly in $\FO$ (the composed reduction from $\calG_0$ to $\db$ belongs in fact to quantifier-free $\FO$).
\end{proof}

\subsection{Proofs of Propositions~\ref{pro:nlc} and~\ref{pro:pc}}

\begin{proof}[Proof sketch of Proposition~\ref{pro:nlc}]
$\NL$-hardness is a consequence of Lemma~\ref{lem:bi}.
To show membership in $\NL$, we reduce the complement of $\certainty{q}{\fk}$ to the problem $\problem{REACHABILITY}$.
Let $\db$ be a database instance that is input to $\certainty{q}{\fk}$.
Construct a directed graph as follows.
The vertex-set is $V\defeq\{c\mid N(\underline{c},c)\in\db\}\cup\{\bot\}$, where $\bot$ is a fresh value.
Edges are introduced as follows:
for every vertex $c$, 
if the block $N(\underline{c},*)$ of $\db$ is $\{N(\underline{c},c)$, $N(\underline{c},d_{1})$, \dots, $N(\underline{c},d_{n})\}$, then 
\begin{itemize}
\item
if $\{d_{1},\ldots,d_{n}\}\subseteq V$, then add edges $\{(c,d_{i})\mid 1\leq i\leq n\}$;
\item
otherwise add an edge $(c,\bot)$.
\end{itemize}
Finally, for every fact $O(\underline{c})$ in $\db$, if $c\in V$, then mark the vertex~$c$.
It can now be verified that $\db$ is a ``no''-instance of $\certainty{q}{\fk}$ if and only if $\bot$ can be reached from every marked vertex.
\end{proof}

\begin{proof}[Proof sketch of Proposition~\ref{pro:pc}]
Reduction from and to $\problem{DUAL}$ $\problem{HORN}$ $\problem{SAT}$, which is $\P$-complete~\cite{DBLP:conf/stoc/Schaefer78}.
Given an instance $\varphi$ of $\problem{DUAL}$ $\problem{HORN}$ $\problem{SAT}$,
construct a database instance $\db_{\varphi}$ as follows:
\begin{itemize}
\item
add $O(\underline{1})$;
\item
for every clause $C_{i}=p_{1}\lor\dotsm\lor p_{n}$, add $N(\underline{i},c,1)$ together with $\{N(\underline{i},d,p_{i})\mid i\in[n]\}$; and
\item
for every clause $C_{i}=\neg q\lor p_{1}\lor\dotsm\lor p_{n}$, add $N(\underline{i},c,q)$ together with $\{N(\underline{i},d,p_{i})\mid i\in[n]\}$.
\end{itemize}
We claim that $\varphi$ is satisfiable if and only if $\db_{\varphi}$ is a ``no''-instance of $\certainty{q}{\fk}$.

To show membership in $\P$, we reduce the complement of the problem $\certainty{q}{\fk}$ to $\problem{DUAL}$ $\problem{HORN}$ $\problem{SAT}$.
Given a database instance $\db$ that is input to the problem $\certainty{q}{\fk}$, construct $\varphi_{\db}$ as follows:
\begin{itemize}
\item
for every fact $O(\underline{p})$, add a clause $p$;
\item
consider every $N$-block of the following form, with $n\geq 1$, $m\geq 0$, and $b_{i}\neq c$ for $i\in[m]$:
$$
\begin{array}{c|ccc}
N & \underline{x} & c & y\\\cline{2-4}
  & i & c & p_{1}\\
  & i & c & p_{2}\\
  & \vdots & \vdots & \vdots\\
  & i & c & p_{n}\\
  & i & b_{1} & q_{1}\\
  & i & b_{2} & q_{2}\\
  & \vdots & \vdots & \vdots\\
  & i & b_{m} & q_{m}\\\cdashline{2-4}
\end{array}
$$
For such an $N$-block, add, for every $j\in[n]$, a clause $\neg p_{j}\lor q_{1}\lor q_{2}\lor\dotsm\lor q_{m}$. 
\end{itemize}
We claim that $\db$ is a ``no''-instance of $\certainty{q}{\fk}$ if and only if $\varphi_{\db}$ is satisfiable.
\end{proof}

\section{Proofs for Section~\ref{sec:fo}}\label{sect:apE}

\subsection{Helping Lemma}\label{sec:helpingfo}


We will use the following helping lemma.

\begin{lemma}\label{lem:obedience-opo}
Let $q$ be query in $\sjfbcq$, and $\fk$ a set of foreign keys about~$q$.
Assume that every foreign key in $\fk$ is strong.
Let $\foreignkey{R}{i}{S}$ be a foreign key in $\fk$, where $S$ is obedient over $\fk$ and~$q$.
Assume that $\fkclosure{S}{q}{\fk}=\{S\}$.
Assume that at least one of the following properties holds:
\begin{enumerate}[label=(\roman*)]
\item\label{it:everyF} The attack graph of $q$ is acyclic, and $\keyvars{F}\neq\emptyset$ for every $F\in q$.
\item\label{it:obedientR} $R$ is obedient over $\fk$ and $q$.
\end{enumerate}
Let $q_{0}=q\setminus\{S\}$ and $\fk_{0}=\fk\setminus\{\foreignkey{R}{i}{S}\}$.
Suppose $(q,\fk)$  has no  block-interference.
Then $\fk_0$ is about $q_0$, and  $(q_0,\fk_0)$ has  no block-interference.
\end{lemma}
\begin{proof}
First we make some observations about the unique $S$-atom of $q$. 
Suppose this atom is of the form $S(\underline{y}, v_2, \dots ,v_m)$.  We may assume that $S$ has signature $[m,1]$, for $S$ has an  incoming (unary) foreign key. 

Since $S$ is obedient, $\fkclosure{S}{q}{\fk}=\{S\}$, and $\fk$ contains only strong foreign keys, it follows that $\fk$ does not contain any foreign key of the form $\foreignkey{S}{i}{T}$ for $T\neq S$. Moreover, it does not contain any foreign key of the form $\foreignkey{S}{i}{S}$ by Theorem~\ref{thm:syntactic-obedience}~\ref{it:strongweak}. 
This implies that $R\neq S$, and moreover that $\posclosure{(P_S)}{\fk}=P_S$ for $P_S\defeq \{(S,2), \dots, (S,m)\}$.
It follows by Theorem \ref{thm:syntactic-obedience} that $v_1, \dots ,v_m$ are orphan variables of $q$. 

Also, since $\fk$ is about $q$, position $(R,i)$ is occupied in $q$ by the primary key term $y$ of $S(\underline{y}, v_2, \dots ,v_m)$. Since $R\neq S$, the unique $R$-atom of $q$ must belong to $q_0$. Hence, position $(R,i)$ is occupied also in $q_0$ by $y$. Note that $y$ must be a variable.
In case~\ref{it:everyF} this is immediate. In case~\ref{it:obedientR}, we note that that the non-primary-key position $(R,i)$ of an obedient atom $R$ cannot be occupied by a constant.

Let us then turn to the claim that $\fk_0$ is about $q_0$. For this, we show that $\fk_0$ contains no foreign key in which $S$ appears. Since we have already shown that  $\outgoing{S}{\fk_0}=\emptyset$, it remains to be shown that $\incoming{S}{\fk_0}=\emptyset$.

Assume for the sake of contradiction that $\incoming{S}{\fk_0}$ contains a foreign key. This foreign key must be of the form $\foreignkey{T}{j}{S}$ with either $i\neq j$ or $T\neq R$ (or both), and it is strong by the hypothesis of the lemma.
Two cases can occur:
\begin{description}
\item[Case that $T\neq R$.] 
Since $\fk$ is about $q$, the variable $y$ occurs in $q$ at position $(T,j)$. Assuming~\ref{it:obedientR}, this is not possible since the non-primary-key position $(R,i)$ of an obedient atom cannot be occupied by the same variable as the non-primary-key position $(T,j)$ of another atom. Hence we assume~\ref{it:everyF}.
By the assumption that $q$ has an acyclic attack graph, either $R\nattacks{q}T$ or  $T\nattacks{q}R$ (or both).
Assume $R\nattacks{q}T$ (the other possibility is symmetrical).
Then, it must be that $y\in\keyclosure{R}{q}$. 
From the hypothesis that $\keyvars{F}\neq\emptyset$ for every $F\in q$, it follows that $y$ is connected to some variable of $\keyvars{R}$ in $q\setminus\{R\}$.
Then $(q,\fk)$ has block-interference, a contradiction. 
\item[Case that $T=R$ and $i\neq j$.] 
In this case, $\fk$ contains strong foreign keys $\foreignkey{R}{i}{S}$ and $\foreignkey{R}{j}{S}$.  
Both non-primary-key positions $(R,i)$ and $(R,j)$ are non-obedient since they share the variable $y$ by the assumption that $\fk$ is about $q$. 
Hence condition~\eqref{it:marking} of block-interference holds for both foreign keys. Moreover, $R$ is disobedient, and thus case~\ref{it:everyF} of the current lemma statement must hold. 
Since $\keyvars{F}\neq\emptyset$ for every $F\in q$, it follows that $\FD{q}\not\models\fd{\emptyset}{\{y\}}$. 
Consequently, condition~\eqref{it:no-constant} of block-interference holds for the variable $y$. Lastly, since $O$ is obedient, we conclude that both foreign keys are block-interfering in $q$. This contradicts the assumption that $(q,\fk)$ has block-interference. 
\end{description}
We conclude by contradiction that $\incoming{S}{\fk_0}=\emptyset$. This establishes that $\fk_{0}$ is about $q_{0}$.

Next, we turn to the claim that $(q_0,\fk_0)$ has  no block-interference.
Consider first the following claim.
\begin{claim}\label{claim:P}
Let $P$ be a set of non-primary-key positions concerning some relation name appearing in $q_0$. Then, $P$ is obedient over $\fk$ and $q$ if and only if it is obedient over $\fk_0$ and $q_0$.
\end{claim}
\begin{proof}[Proof of Claim \ref{claim:P}]
We prove this claim using Theorem \ref{thm:syntactic-obedience}. 
Consider two different cases:

\paragraph{Case that $(S,1)\in \poscomplement{P}{\fk}$.} It is clear that in this case $\posclosure{P}{\fk_0}=\posclosure{P}{\fk}$. 
Moreover, since $\fk_0$ is about $q_0$, all the $S$-positions must belong to $\poscomplement{P}{\fk_0}$. It also follows that $\poscomplement{P}{\fk_0}=\poscomplement{P}{\fk}$.
Hence,
 the conditions listed in Theorem \ref{thm:syntactic-obedience} that concern only the closure (and not the complement) of $P$ are equivalent under $(\fk,q)$ and $(\fk_0,q_0)$. We thus only need to consider
the condition Theorem~\ref{thm:syntactic-obedience}~\ref{it:disjoint}. 

For one direction, it is clear that if no variable occurs in $q$ at positions of $\posclosure{P}{\fk}$ and $\poscomplement{P}{\fk}$, then likewise  no variable occurs in $q_0$ at positions of $\posclosure{P}{\fk_0}$ and $ \poscomplement{P}{\fk_0}$. 

For the other direction, assume that no variable occurs in $q_0$ at positions of $\posclosure{P}{\fk_0}$ and $\poscomplement{P}{\fk_0}$. We need to show that no variable occurs in $q$ at positions of $\posclosure{P}{\fk}$ and $\poscomplement{P}{\fk}$. 
Note that every position of $\posclosure{P}{\fk}$ that occurs in $q$ also occurs in $q_0$, and the $S$-positions are the only positions 
 of   $\poscomplement{P}{\fk}$ that occur in $q$ but not in $q_0$. 
 It now remains to be shown that  in $q$ no $S$-position shares a variable with some position in $\posclosure{P}{\fk}$. 
 Since $v_2, \dots ,v_m$ are orphan, it suffices to show that $y$  does not occur at a position of  $\posclosure{P}{\fk}$ in $q$.
 Recall that $y$ occurs in $q_0$ at position $(R,i)$. 
 Since $(S,1)\in \poscomplement{P}{\fk}$, we obtain that $(R,i)\in \poscomplement{P}{\fk}$. Since $\poscomplement{P}{\fk}=\poscomplement{P}{\fk_0}$, 
and since we assumed that no variable occurs in $q_0$ at positions of $\posclosure{P}{\fk_0}$ and $\poscomplement{P}{\fk_0}$, the variable $y$ cannot occur in any position of $\posclosure{P}{\fk_0}$. Thus $y$ cannot occur in $\posclosure{P}{\fk}$, for $\posclosure{P}{\fk_0}=\posclosure{P}{\fk}$.  This  concludes the proof in the case that $(S,1)\in \poscomplement{P}{\fk}$.

\paragraph{Case that $(S,1)\in \posclosure{P}{\fk}$.} Recall that $\foreignkey{R}{i}{S}\in \fk$ is the only foreign key of $\fk$ in which the relation name $S$ appears, since we showed that $q_0$ is about $\fk_0$. Hence $(R,i)\in \posclosure{P}{\fk_0}$, and $\posclosure{P}{\fk}=\posclosure{P}{\fk_0}\cup \{(S,1), \dots ,(S,m)\}$. It is thus easy to see using Theorem \ref{thm:syntactic-obedience} that if $P$ is obedient over $(\fk,q)$, then the same holds over $(\fk_0,q_0)$.

For the other direction, suppose $P$ is obedient over $(\fk_0,q_0)$. 
By Theorem~\ref{thm:syntactic-obedience}~\ref{it:strongweak}, no position of $P$ belongs to a cycle in the dependency graph of $\fk_0$. Since $\fk$ extends $\fk_0$ with a foreign key $\foreignkey{R}{i}{S}$, and since $\fk$ contains no foreign keys outgoing from $S$, we observe that $P$ cannot belong to any cycle in the dependency graph of $\fk$; that is, condition~\ref{it:strongweak} of Theorem~\ref{thm:syntactic-obedience} holds with respect to $P$ and $\fk$.

Recall that the primary-key variable $y$ of $S(\underline{y}, v_2, \dots ,v_m)$ occurs at position $(R,i)$ in $q_0$.
Since $v_2, \dots ,v_m$ are orphan variables of~$q$, we observe that $\posclosure{P}{\fk}$ and $q$ inherit conditions~\ref{it:no-const} and~\ref{it:distinct} of Theorem~\ref{thm:syntactic-obedience} from $\posclosure{P}{\fk_0}$ and $q_0$. 

Let us then turn to the remaining condition~\ref{it:disjoint} of Theorem~\ref{thm:syntactic-obedience}. Note that every position of $\poscomplement{P}{\fk}$ that occurs in $q$ also occurs in~$q_0$, and the $S$-positions are the only positions of $\posclosure{P}{\fk}$ that occur in $q$ but not in $q_0$. Hence it suffices to show that no $S$-position shares a variable in $q$ with some position of 
$\poscomplement{P}{\fk}$. This is clear for positions $(S,2), \dots ,(S,m)$, for we know that $v_2, \dots ,v_m$ are orphan. The variable $y$ that occurs at $(S,1)$ is also known to occur at $(R,i)$ in~$q$.
Since $(R,i)\in \posclosure{P}{\fk_0}$, we obtain by Theorem~\ref{thm:syntactic-obedience}~\ref{it:disjoint} that  $y$ does not occur in $q_0$ at any position of $\poscomplement{P}{\fk_0}$. Since $q\setminus q_0=\{S\}$ and all the $S$-positions belong to $\posclosure{P}{\fk}$, this entails that
  $y$ does not occur in $q$ at any position of $\poscomplement{P}{\fk}$.
  This shows that condition~\ref{it:disjoint} of Theorem~\ref{thm:syntactic-obedience} extends from $\posclosure{P}{\fk_0}$ and $q_0$  to $\posclosure{P}{\fk}$ and $q$. 
  
  Hence  $P$ is obedient over $(\fk,q)$. This concludes the proof of Claim \ref{claim:P}.
\end{proof}

We can now prove that $(q_0,\fk_0)$ has no block-interference.
Toward contradiction, assume that this is not the case. 
Then, $\lclosure{\fk_0}$ contains a block-interfering foreign key $\sigma\defeq \foreignkey{N}{j}{O}$ for some atoms 
 $N(\underline{t_{1}, \dots ,t_k},t_{k+1},\dots,t_{n})$ and  $O(\underline{t_j},\vec{y})$ of $q_0$.
Note that $\lclosure{\fk_0}=\fk_0$ since $\fk_0$ contains only strong foreign keys. In what follows, we establish a contradiction by showing that $\sigma$ is block-interfering also with respect to $(q,\fk)$.

Claim \ref{claim:P} now entails that, with respect to $(q,\fk)$, $\sigma$ satisfies 
Definition~\ref{def:block-int}~\eqref{it:obedientO}, and also Definition~\ref{def:block-int}~\eqref{it:marking} if this holds with respect to $(q_0,\fk_0)$. It remains to show that $\sigma$ satisfies  Definition~\ref{def:block-int}~\eqref{it:no-constant}, 
and also Definition~\ref{def:block-int}~\eqref{it:bifour} if this holds with respect to $(q_0,\fk_0)$. 
Let us denote the terms $t_i$ and $t_j$ of the remaining conditions by $u$ and $z$, respectively.
Next we turn to these conditions.

Concerning Definition~\ref{def:block-int}~\eqref{it:no-constant}, we claim that $z$ belongs to
\begin{equation}\label{eq:V}
V=\{v\in \queryvars{q'}\mid \FD{q}\not\models\fd{\emptyset}{\{v\}}\},
\end{equation}
where $q'\defeq q\setminus \{N\}$.
 By the assumption that $(q,\fk)$ has no block-interference, we obtain that $z$ belongs to 
 \begin{equation}\label{eq:V0}
 V_0=\{v\in \queryvars{q'_0}\mid \FD{q_0}\not\models\fd{\emptyset}{\{v\}}\},
\end{equation}
where $q_0' \defeq q_0\setminus \{N\}$.
Toward contradiction, assume that $z\notin V$.  Let $(F_1, \dots, F_n)$ be the shortest proof of $ \FD{q}\models\fd{\emptyset}{\{z\}}$, i.e., the shortest sequence of atoms from $q$ such that $z\in \atomvars{F_n}$, $\keyvars{F_1}=\emptyset$, and $\keyvars{F_{i+1}}\subseteq \bigcup_{j=1}^i \atomvars{F_j}$ for $i\in [n-1]$. By our assumption, we find $i\in [n]$ such that $F_i$ does not belong to $q_0$, whence $F_i=S(\underline{t}, v_2, \dots ,v_m)$. Since $v_2, \dots ,v_m$ are  orphan variables of $q$, they cannot occur at a primary-key position of any atom in $q$. Since we selected the shortest sequence, we obtain that $i=n$. For the same reason, $z\in \{v_2, \dots ,v_m\}$. But then, $z$ is an orphan variable of $q$ that appears in an atom of $q$ that is not in $q_0$. It follows that $z\notin \queryvars{q_0}$. This contradicts our assumption that $z$ belongs to $ V_0$. We conclude that $z$ belongs to $ V$. Hence $\sigma$ satisfies Definition~\ref{def:block-int}~\eqref{it:no-constant} with respect to $(q,\fk)$.

Concerning Definition~\ref{def:block-int}~\eqref{it:bifour}, suppose $(v_1, \dots ,v_n)$ is a path in $\constgraph{V_0}{q'_0}$ for $q'_0\defeq q_0\setminus\{N(\underline{t_1, \dots ,t_k},t_{k+1},\dots,t_{n})\}$ such that $v_1=u$ and $v_n=z$. We claim that that the same path exists in $\constgraph{V}{q'}$ for $q'\defeq q\setminus\{N(\underline{t_1, \dots ,t_k},t_{k+1},\dots,t_{n})\}$. Note that $V$ and $V_0$ are given by Eq.~\eqref{eq:V} and Eq.~\eqref{eq:V0}. Thus, we can argue exactly as in the previous case that $v_i\in V$ for $i\in [n]$. The claim follows from this. We conclude that $\sigma$ satisfies Definition~\ref{def:block-int}~\eqref{it:bifour} with respect to $(q,\fk)$ if the same happens with respect to $(q_0,\fk_0)$.

We have shown that $\sigma$ is block-interfering with respect to $(q,\fk)$, contradicting our assumption. We conclude by contradiction that $(q_0,\fk_0)$ has no block-interference.
This concludes the proof of Lemma~\ref{lem:obedience-opo}.
\end{proof}

\subsection{Removal of Weak Foreign Keys}

In the following lemma, we assume that $\fk$ is closed under logical implication.
Under this assumption, it is not sufficient to remove one weak foreign key $\sigma$ at a time, because it may be that $\lclosure{\fk}\setminus\{\sigma\}\equiv\lclosure{\fk}$.
Instead, we remove all weak foreign keys referencing a same relation name.

\begin{lemma}[Weak foreign keys]\label{lem:nonProper}
Let $\fk$ be a set of foreign keys such that $\lclosure{\fk}=\fk$.
Let $\fkw$ be the set of weak foreign keys in $\fk$.
Assume that some non-trivial foreign key in $\fkw$ references $S$, and
let $\fk_{0}=\fk\setminus\incoming{S}{\fkw}$.
Let $q$ be a query in $\sjfbcq$ such that $\fk$ is about $q$.
Then, $\fk_{0}$ is about $q$, and the following hold:
\begin{itemize}
\item
$\certainty{q}{\fk}\reducesTo{\FO}\certainty{q}{\fk_{0}}$; and
\item
if $(q,\fk)$ has no block-interference,
then $(q,\fk_{0})$ has no block-interference. 
\end{itemize}
\end{lemma} 
\begin{proof}[Proof of Lemma~\ref{lem:nonProper}]
It is obvious that $\fk_{0}$ is about $q$.

\textbf{Proof of the second item.}
Let  $P$ be a set of non-primary-key positions concerning some relation name appearing in $q$. We claim that $P$ is obedient over $(\fk,q)$ if and only if it is obedient over $(\fk_0,q)$.  It is easy to see that the second item  follows readily from this claim. Moreover, by Theorem \ref{thm:syntactic-obedience} the claim holds if $\posclosure{P}{\fk}= \posclosure{P}{\fk_0}$. Since $\posclosure{P}{\fk_0}\subseteq \posclosure{P}{\fk}$ is clear, we only need to show that
 $\posclosure{P}{\fk}\subseteq \posclosure{P}{\fk_0}$.
 
 Assume toward contradiction that $\posclosure{P}{\fk}\not\subseteq \posclosure{P}{\fk_0}$. It follows that $(S,1)\in\posclosure{P}{\fk}\setminus \posclosure{P}{\fk_0}$. That is, the dependency graph of $\fk$ contains a path 
 $
 \calP=((R_1,i_1),\dots ,(R_n,i_n),(S,1))
 $
 where $(R_1,i_1)\in P$, while no such path from $P$ to $(S,1)$ exists in the dependency graph of $\fk_0$.
 We may assume w.l.o.g. that the subpath $\calP' =((R_1,i_1),\dots ,(R_n,i_n))$ is in the dependency graph of $\fk_0$. 
 
 Let $k\in [n]$ be the greatest integer such that $(R_k,i_k)$ is a non-primary-key position; note that $k$ is well defined since $(R_1,i_1)$ is a non-primary-key position. Observe that the subpath of $\calP$ from $(R_k,i_k)$ to $(S,1)$ is weak, whence $\foreignkey{R_k}{i_k}{S}$ must be a strong foreign key in $ \lclosure{\fk}$. But now, since $\lclosure{\fk}=\fk$ and since $\fk\setminus\fk_0$ contains only weak foreign keys, we observe that $\foreignkey{R_k}{i_k}{S}\in \fk_0$. In particular, the dependency graph of $\fk_0$ contains a path $\calP''=((R_1,i_1),\dots ,(R_k,i_k),(S,1))$, whence $(S,1) \in \posclosure{P}{\fk_0}$. This leads to a contradiction, by which we conclude that $\posclosure{P}{\fk}\subseteq \posclosure{P}{\fk_0}$. This concludes the proof of the second item.
  
\textbf{Proof of the first item.}
Let $\db$ be a database instance that is input to $\certainty{q}{\fk}$.
We show the following:
\begin{enumerate}[label=(\Alph*)]
\item\label{it:nno}
if $\db$ is a ``no''-instance of $\certainty{q}{\fk}$,
then it is a ``no''-instance of $\certainty{q}{\fk_{0}}$; and
\item\label{it:yyo}
if $\db$ is a ``no''-instance of $\certainty{q}{\fk_{0}}$,
then it is a ``no''-instance of $\certainty{q}{\fk}$.
\end{enumerate}
Note that the reduction is the identity.

\framebox{Proof of \ref{it:nno}}
Assume that $\db$ is a ``no''-instance of the problem $\certainty{q}{\fk}$.
There is a repair $\rep$ of $\db$ with respect to $\fk\cup\pk$ such that $\rep\not\models q$.
Since $\fk_{0}\subseteq\fk$, we have $\rep\models\fk_{0}$.
Let $\rep^{*}$ be a database instance that satisfies $\fk_{0}\cup\pk$ such that $\rep^{*}\closer{\db}\rep$.
It suffices to show $\rep^{*}\not\models q$.
We have
\begin{align}
\rep\cap\db & \subseteq\rep^{*}\label{eq:rightstar}\\
\rep^{*} & \subseteq\rep\cup\db\label{eq:leftstar}
\end{align}
Assume for the sake of contradiction that there is a valuation $\theta$ over $\queryvars{q}$ such that $\theta(q)\subseteq\rep^{*}$.
Since $\theta(q)\nsubseteq\rep$ and by~\eqref{eq:leftstar}, there must be a fact $A\in\theta(q)$ such that $A\in\db\setminus\rep$.
By Lemma~\ref{lem:norepair}, it is correct to conclude that $\rep$ is not a repair of $\db$ with respect to $\fk$, a contradiction.

\framebox{Proof of \ref{it:yyo} }
Assume that $\db$ is a ``no''-instance of the problem $\certainty{q}{\fk_{0}}$.
There is a repair $\rep$ of $\db$ with respect to $\fk_{0}\cup\pk$ that falsifies~$q$.

If $\rep$ satisfies $\incoming{S}{\fkw}$, then $\rep$ is also a repair of $\db$ with respect to $\fk\cup\pk$, and the desired result holds vacuously.
Assume from here on that $\rep\not\models\incoming{S}{\fkw}$.
Let $\Delta^{-}$ be the set of all facts of $\rep$ that are dangling with respect to $\incoming{S}{\fkw}$.

We show that $\rep\setminus\Delta^{-}\models\fk$.
Assume for the sake of contradiction that $R(\underline{a_{1},\dots,a_{k}},\vec{b})\in\Delta^{-}$ is dangling with respect to some foreign key $\foreignkey{R}{i}{S}$ in $\incoming{S}{\fkw}$, and its removal entails a dangling $T$-fact $T(b_{1},\dots,b_{m})$ with respect to a foreign key $\foreignkey{T}{j}{R}$ in $\fk$. 
Clearly, $T(b_{1},\dots,b_{m})\notin\Delta^{-}$.
Then it must be that $R$ has signature $\signature{n}{1}$ (that is, $k=1$) and $b_{j}=a_{1}$.
Then $\fk\models\foreignkey{T}{j}{S}$, and therefore, by the hypothesis of the lemma that $\fk=\lclosure{\fk}$, it follows that $\foreignkey{T}{j}{S}$ belongs to $\fk$.

We distinguish two cases. 
\begin{itemize}
\item
\emph{Case that $\foreignkey{T}{j}{S}$ is a strong foreign key.}
Then $\foreignkey{T}{j}{S}$ belongs to $\fk_{0}$ and would be falsified by $\rep$, contradicting that $\rep$ is a repair of $\db$ with respect to $\fk_{0}$.
\item
\emph{Case that $\foreignkey{T}{j}{S}$ is a weak foreign key.}
Then $\foreignkey{T}{j}{S}$ belongs to $\incoming{S}{\fkw}$, and therefore $T(b_{1},\dots,b_{m})$ is in $\Delta^{-}$, a contradiction.
\end{itemize}
We conclude by contradiction that $\rep\setminus\Delta^{-}\models\fk$.

We show $\Delta^{-}\subseteq\db$. 
Assume for the sake of contradiction that $\Delta^{-}\setminus\db$ contains $R(\underline{\vec{a}},\vec{b})$.
Since $\rep\setminus\{R(\underline{\vec{a}},\vec{b})\}\closerneq{\db}\rep$ and, by the reasoning in the previous paragraph, 
$\rep\setminus\{R(\underline{\vec{a}},\vec{b})\}\models\fk_{0}$, it follows that $\rep$ is not a repair of $\db$ with respect to $\fk_{0}$, a contradiction.


Clearly, $\rep\models\pk$.
Let $\rep^{*}$ be a database instance such that $\rep^{*}\models\fk\cup\pk$ and $\rep^{*}\closer{\db}\rep\setminus\Delta^{-}$. That is,
\begin{align}
\formula{\rep\setminus\Delta^{-}}\cap\db & \subseteq\rep^{*}\label{eq:seprightstar}\\
\rep^{*} & \subseteq\formula{\rep\setminus\Delta^{-}}\cup\db\subseteq\rep\cup\db\label{eq:sepleftstar}
\end{align}

It suffices to show $\rep^{*}\not\models q$.
Assume for the sake of contradiction that there is a valuation $\theta$ over $\queryvars{q}$ such that $\theta(q)\subseteq\rep^{*}$.
Since $\theta(q)\nsubseteq\rep$ and by~\eqref{eq:sepleftstar}, there must be a fact $A\in\theta(q)$ such that $A\in\db\setminus\rep$.

By Lemma~\ref{lem:relevantblockter}, $\rep\cap\db$ contains a fact $A'$ such that $A'\keyeq A$, where $A'\neq A$.
By~\eqref{eq:seprightstar} and since~$\rep^{*}$ satisfies primary keys, $A'\in\Delta^{-}$.
So for some $\vec{a}\defeq\tuple{a_{1},a_{2},\ldots,a_{k}}$, we can assume $A=R(\underline{\vec{a}},\vec{b})$, $A'=R(\underline{\vec{a}},\vec{b}')$, and the latter fact belongs to $\db\cap\rep$ and is dangling in $\rep$ with respect to some foreign key $\foreignkey{R}{i}{S}$ in $\incoming{S}{\fkw}$. 
Since $\rep^{*}$ satisfies $\foreignkey{R}{i}{S}$ and by~\eqref{eq:sepleftstar}, there is a fact of the form $S(\underline{a_{i}},\filler)$ in $\db\cap\theta(q)$.
On the other hand, $\rep$ contains no fact of the non-empty block $S(\underline{a_{i}},*)$ of $\db$.
By Lemma~\ref{lem:norepair}, it is now correct to conclude that $\rep$ is not a repair of $\db$ with respect to $\fk_{0}$, a contradiction.
This concludes the proof of~\ref{it:yyo}.
\end{proof}

\subsection{Removal of Strong Foreign Keys}



\subsubsection{Removal of $\opo$ Foreign Keys}

\begin{lemma}[Type $\opo$]\label{lem:opo}
Let $q$ be query in $\sjfbcq$, and $\fk$ a set of foreign keys about~$q$.
Let $\foreignkey{R}{i}{S}$ be a strong foreign key of type $\opo$ in $\fk$.
Assume that $\fkclosure{S}{q}{\fk}=\{S\}$.\footnote{Recall from Section~\ref{sec:preliminaries} that if we use a relation name~$S$ wherever an atom is expected, we mean the unique $S$-atom of the self-join-free query that can be understood from the context.}
Following Lemma~\ref{lem:nonProper}, assume that every foreign key in $\fk$ is strong.
Let $q_{0}=q\setminus\{S\}$ and $\fk_{0}=\fk\setminus\{\foreignkey{R}{i}{S}\}$.
Suppose $(q,\fk)$ has no block-interference.
Then, $\fk_{0}$ is about $q_{0}$, and the following hold:
\begin{itemize}
\item
$\certainty{q}{\fk}\reducesTo{\FO}\certainty{q_{0}}{\fk_{0}}$; 
\item $(q_{0},\fk_{0})$  has no block-interference; and
\item
if the attack graph of $q$ is acyclic, then the attack graph of $q_{0}$ is acyclic.
\end{itemize}
\end{lemma} 
\begin{proof}[Proof of Lemma~\ref{lem:opo}]
That $\fk_0$ is about $q_0$ follows by Lemma~\ref{lem:obedience-opo}.

\textbf{Proof of the third item.}
Easy.

\textbf{Proof of the second item.}
Follows by Lemma \ref{lem:obedience-opo}.



\textbf{Proof of the first item.}
Let $\db$ be a database instance that is input to $\certainty{q}{\fk}$.
Let $\Delta$ be the union of all $R$-blocks of $\db$ that are not relevant for $\fkclosure{R}{q}{\fk}$ in $\db$.
Let $\db_{0}=\restrict{\formula{\db\setminus\Delta}}{q\setminus\{S\}}$.
Clearly, $\db_{0}$ can be obtained from $\db$ by a first-order query.
Note that if the $R$-atom is, for example, $R(\underline{c},y)$,
then $\Delta$ contains all $R$-facts $\db$ of the form $R(\underline{a},\filler)$ where $a\neq c$. 

We show the following:
\begin{enumerate}[label=(\Alph*)]
\item\label{it:opoA}
if $\db$ is a ``no''-instance of $\certainty{q}{\fk}$,
then $\db_{0}$ is a ``no''-instance of $\certainty{q\setminus\{S\}}{\fk\setminus\{\foreignkey{R}{i}{S}\}}$; and
\item\label{it:opoB}
if $\db_{0}$ is a ``no''-instance of $$\certainty{q\setminus\{S\}}{\fk\setminus\{\foreignkey{R}{i}{S}\}},$$
then $\db$ is a ``no''-instance of $\certainty{q}{\fk}$.
\end{enumerate}

\framebox{Proof of \ref{it:opoA}}
Assume $\db$ is a ``no''-instance of $\certainty{q}{\fk}$.
There is a repair $\rep$ of $\db$ (with respect to foreign keys in $\fk$ and primary keys) that falsifies $q$. 

Let $\rep_{0}=\restrict{\rep}{q\setminus\{S\}}$.
Assume for the sake of contradiction that $\rep_{0}\models q\setminus\{S\}$.
Then there is a valuation $\mu$ over $\queryvars{q\setminus\{S\}}$ such that $\mu(q\setminus\{S\})\subseteq\rep_{0}$.
Since $\rep\models\foreignkey{R}{i}{S}$ and the $S$-atom is obedient, it follows that $\mu$ can be extended into a valuation $\mu^{+}$ over $\queryvars{q}$ such that $\mu^{+}(q)\subseteq\rep$, a contradiction.
We conclude by contradiction that $\rep_{0}\not\models q\setminus\{S\}$.

Clearly, $\rep_{0}\models\fk\setminus\{\foreignkey{R}{i}{S}\}$ and $\rep_{0}$ satisfies primary keys.
Let $\rep_{0}^{*}$ be a database instance such that $\rep_{0}^{*}\closer{\db_{0}}\rep_{0}$ and $\rep_{0}^{*}$ is consistent with respect to foreign keys in $\fk\setminus\{\foreignkey{R}{i}{S}\}$ and primary keys.
Thus,
\begin{align}
\rep_{0}\cap\db_{0} & \subseteq\rep_{0}^{*}\label{eq:lower}\\
\rep_{0}^{*} & \subseteq\rep_{0}\cup\db_{0}\label{eq:upper}
\end{align}
It suffices to show $\rep_{0}^{*}\not\models q\setminus\{S\}$.
Assume for the sake of contradiction that there is a valuation~$\theta$ over $\queryvars{q\setminus\{S\}}$ such that $\theta(q\setminus\{S\})\subseteq\rep_{0}^{*}$.

Assume towards a contradiction that $\theta(R)\notin\rep_{0}$.
From~\eqref{eq:upper}, it follows $\theta(R)\in\db_{0}$.
By construction of $\db_{0}$, we have that $\theblock{\theta(R)}{\db}$ is relevant for $\fkclosure{R}{q}{\fk}$ in $\db$.
By Corollary~\ref{cor:relevantblockbis},
$\rep$ must contain an $R$-fact (call it $B$) of $\theblock{\theta(R)}{\db}$.
Then $B\in\rep_{0}\cap\db_{0}$, and by~\eqref{eq:lower}, $B\in\rep_{0}^{*}$.
Since $\rep_{0}^{*}$ is consistent with respect to primary keys, we obtain $\theta(R)=B$, and hence $\theta(R)\in\rep_{0}$, a contradiction.
We conclude by contradiction that $\theta(R)\in\rep_{0}$.

Since $\rep_{0}\not\models q\setminus\{S\}$ and by~\eqref{eq:upper}, there is a fact $A\in\theta(q\setminus\{S\})$ (and therefore $A\in\rep_{0}^{*}$) such that $A\in\db_{0}\setminus\rep_{0}$.
Let $\alpha$ be the set of facts in $\rep_{0}$ that are key-equal to~$A$.
Clearly, $\card{\alpha}\leq 1$.
Let $\beta$ be the set of facts in $\theta(q\setminus\{S\})$ that are not key-equal to a fact in $\rep_{0}$.
We have $\beta\subseteq\db_{0}$ by~\eqref{eq:upper}.
Let $\sep=\formula{\rep\setminus\alpha}\cup\{A\}\cup\beta$.
It can be verified that $\sep\closerneq{\db}\rep$ and that $\sep$ is consistent with respect to foreign keys in $\fk$ and primary keys. In particular, from $\theta(R)\in\rep_{0}\subseteq\rep$, it follows that $\theta(R)$ is not dangling in $\sep$ with respect to $\foreignkey{R}{i}{S}$. 
But then $\rep$ is not a repair, a contradiction.

\framebox{Proof of \ref{it:opoB}}
Assume $\db_{0}$ is a ``no''-instance of  $$\certainty{q\setminus\{S\}}{\fk\setminus\{\foreignkey{R}{i}{S}\}}.$$
Let $\rep_{0}$ be a repair of $\db$ (with respect to foreign keys in $\fk\setminus\{\foreignkey{R}{i}{S}\}$ and primary keys) that falsifies~$q\setminus\{S\}$. 

Construct $\rep$ from $\rep_{0}$ as follows:
\begin{itemize}
\item
for every $S$-fact $A$ in $\db$,
insert a fact from $\theblock{A}{\db}$;
\item
{\emph{Chase step:}}
as long as some $R$-fact $R(a_{1},\dots,a_{n})$ is still dangling with respect to $\foreignkey{R}{i}{S}$, insert an $S$-fact $R(\underline{a_{i}},\vec{b})$ where $\vec{b}$ is a sequence of fresh constants.
\end{itemize}
By construction $\rep$ is consistent with respect to foreign keys in $\fk$ and primary keys.

\begin{claim}\label{cla:opo}
$\rep$ is a repair of $\db$ with respect to foreign keys in $\fk$ and primary keys.
\end{claim}
\begin{proof}[Proof sketch of Claim~\ref{cla:opo}.]
By Corollary~\ref{cor:relevantblockbis}, $\rep_{0}$ contains an $R$-fact from every $R$-block of $\db_{0}$. 
The invented $S$-facts inserted in the \emph{Chase Step} are needed to satisfy $\foreignkey{R}{i}{S}$.

Assume there is a repair $\sep$ of $\db$ with respect to foreign keys in $\fk$ and primary keys that contains more $R$-facts of $\db$. 
Since these additional $R$-facts are not relevant for $\fkclosure{R}{q}{\fk}$ in $\db$, the repair $\sep$ must also contain invented fresh $S$-facts not in $\rep$, and therefore $\sep$ and $\rep$ would not be comparable by $\closer{\db}$.

Finally, note that, since $\incoming{S}{\fk}=\{\foreignkey{R}{i}{S}\}$, the insertion of $S$-facts in the \emph{Chase step} does not entail further insertions. 
\end{proof}
From $\rep_{0}\not\models q\setminus\{S\}$, it follows $\rep\not\models q$. 
This concludes the proof of Lemma~\ref{lem:opo}.
\end{proof}

\subsubsection{Removal of $\dpd$ Foreign Keys}

\begin{lemma}[Type $\dpd$]\label{lem:dpd}
Let $q$ be query in $\sjfbcq$, and $\fk$ a set of foreign keys about~$q$.
Let $\foreignkey{R}{i}{S}$ be a strong foreign key of type $\dpd$ in $\fk$.
Following Lemma~\ref{lem:nonProper}, assume that every foreign key in $\fk$ is strong.
Let $\fk_{0}=\fk\setminus\{\foreignkey{R}{i}{S}\}$.
Then, $\fk_{0}$ is about $q$, and the following hold:
\begin{itemize}
\item
$\certainty{q}{\fk}\reducesTo{\FO}\certainty{q}{\fk_{0}}$; and
\item
if $(q,\fk)$ has no block-interference,
then $(q,\fk_{0})$ has no block-interference. 
\end{itemize}
\end{lemma} 
\begin{proof}[Proof of Lemma~\ref{lem:dpd}]\mbox{}\\
\textbf{Proof of the second item.}
Let $P$ be a set of non-primary-key positions concerning some relation name that appears in $q$. We claim that $P$ is obedient over $(\fk,q)$ if and only if it is obedient over $(\fk_0,q)$. This is sufficient for the second item, since the query $q$ is not modified in the reduction.

Assume first that $P$ is obedient over $(\fk,q)$. Since $\foreignkey{R}{i}{S}$ is of type $\dpd$ (with respect to $\fk$), it follows that $(R,i)\notin \posclosure{P}{\fk}$. In particular, it holds that $\posclosure{P}{\fk}=\posclosure{P}{\fk_0}$. By Theorem \ref{thm:syntactic-obedience} we obtain that $P$ is obedient over $(\fk_0,q)$.

Assume then that $P$ is obedient over $(\fk_0,q)$. If $\posclosure{P}{\fk}=\posclosure{P}{\fk_0}$, then, analogously to the previous case, $P$ is obedient over $(\fk,q)$. 
Finally, we show that the case $\posclosure{P}{\fk}\neq \posclosure{P}{\fk_0}$ cannot occur.
To this end, assume, for the sake of contradiction, that $\posclosure{P}{\fk}\neq \posclosure{P}{\fk_0}$ holds true.
It must be the case that $(R,i)\in \posclosure{P}{\fk_0}$ and $(S,1)\notin \posclosure{P}{\fk_0}$. But then, since $q$ is about $\fk$, we note that $(S,1)$ and $(R,i)$ are occupied by the same variable in $q$. 
This violates Theorem~\ref{thm:syntactic-obedience}~\ref{it:disjoint}, contradicting the assumption that $P$ is obedient over $(\fk_0,q)$. 
This concludes the proof of the claim and thus that of the second item.

\textbf{Proof of the first item.}
We show the following:
\begin{enumerate}[label=(\Alph*)]
\item\label{it:dpdA}
if $\db$ is a ``no''-instance of $\certainty{q}{\fk}$,
then $\db$ is a ``no''-instance of  $\certainty{q}{\fk\setminus\{\foreignkey{R}{i}{S}\}}$; and
\item\label{it:dpdB}
if $\db$ is a ``no''-instance of $\certainty{q}{\fk\setminus\{\foreignkey{R}{i}{S}\}}$,
then $\db$ is a ``no''-instance of $\certainty{q}{\fk}$.
\end{enumerate}
Note that the reduction is the identity.

\framebox{Proof of \ref{it:dpdA}}
Assume $\db$ is a ``no''-instance of $\certainty{q}{\fk}$.
There is a repair $\rep$ of $\db$ (with respect to foreign keys in $\fk$ and primary keys) that falsifies $q$. 
Clearly, $\rep\models\fk\setminus\{\foreignkey{R}{i}{S}\}$.
Let $\rep^{*}$ be a database instance such that $\rep^{*}\closer{\db}\rep$ and $\rep^{*}$ is consistent with respect to foreign keys in $\fk\setminus\{\foreignkey{R}{i}{S}\}$ and primary keys.
Thus,
\begin{align}
\rep\cap\db & \subseteq\rep^{*}\label{eq:dpdlowerbis}\\
\rep^{*} & \subseteq\rep\cup\db\label{eq:dpdupperbis}
\end{align}
It suffices to show $\rep^{*}\not\models q$.
Assume for the sake of contradiction that there is a valuation~$\theta$ over $\queryvars{q}$ such that $\theta(q)\subseteq\rep^{*}$.
By~\eqref{eq:dpdupperbis}, $\theta(q)\subseteq\rep\cup\db$.

Let $\theta(q)\setminus\rep=\{A_{1}, A_{2},\dots, A_{\ell}\}$.
Since $\rep\not\models q$, we have $\ell\geq 1$.
By~\eqref{eq:dpdlowerbis}, no $A_{i}$ is key-equal to a fact in $\rep\cap\db$.  
Then, by Lemma~\ref{lem:norepair}, $\rep$ is not a repair, a contradiction.

\framebox{Proof of \ref{it:dpdB}}
Suppose the $S$-atom is of the form $S(\underline{{s}},\vec{t})$.

Assume $\db$ is a ``no''-instance of $$\certainty{q}{\fk\setminus\{\foreignkey{R}{i}{S}\}}.$$
There is a repair $\rep_{0}$ of $\db$ (with respect to foreign keys in $\fk\setminus\{\foreignkey{R}{i}{S}\}$ and primary keys) that falsifies $q$.
Let $\rep_{1}$ be the database instance obtained from $\rep_{0}$ as follows:
\begin{itemize}
\item[]
\emph{Chase step:} if some $R$-fact $R(\underline{a_{1},\dots,a_{k}}, a_{k+1},\dots,a_{n})$ in $\rep_{0}$ is dangling with respect to $\foreignkey{R}{i}{S}$, then insert an invented $S$-fact $S(\underline{a_{i}},\vec{b})$, for some sequence $\vec{b}$ of fresh constants.
Note that $R=S$ is possible given $k=1$. 
\end{itemize}
We first show that $\rep_1$ is irrelevantly dangling (with respect to $(\db,\fk,q)$). As per Definition \ref{def:irr}, we need to show that whenever a fact $A=T(\underline{b_{1},\dots,b_{k}}, b_{k+1},\dots,b_{n})\in \rep_1$ is dangling in $\rep_1$ for a foreign key $\foreignkey{T}{i}{U}\in {\fk}$, then 
$P$ contains $(T,i)$ and is disobedient over $\fk$ and~$q$, where $P_A$ is defined as the set of non-primary-key positions $(T,j)$ such that $b_j$ is orphan in $\db\cup\rep_1$ and does not belong to $\queryconst{q}$.

Note that the facts added by the chase step are the only facts of $\rep_1$ that can be dangling. 
Suppose $B=S(\underline{a},\vec{b})$ is dangling for $\foreignkey{S}{i}{T}\in {\fk}$. 
By construction, the set $P_B$ consists of all non-primary-key positions of $S$. This readily implies that $P_B$ is disobedient, because we know that $S$ is disobedient. Moreover, $\foreignkey{S}{i}{T}$ is strong, since $\fk$ contains only strong foreign keys. Consequently, $(S,i)\in P_B$. Furthermore, we note that $\rep_1$ clearly satisfies all the primary keys. We conclude that $\rep_1$ is irrelenvantly dangling,

Define $\rep$ as $\rep_1$ if $\rep_1$ is a pre-repair. Otherwise, let $\rep$ be any pre-repair of $\db$ over $\fk$ and $q$ such that $\rep \capcloserneq{\db} \rep_1$. We show that $\rep\not\models q$. Assume toward contradiction that this is not the case. Let $\theta$ be a valuation such that $\theta(q)\subseteq \rep$.

First we note that 
\begin{equation}\label{eq:emptyset}
\theta(q)\cap (\rep_1\setminus \rep_0)=\emptyset,
\end{equation}
 i.e., $\theta(q)$ does not contain any fact that has been added by the chase step. Indeed, if $S(\underline{a},\vec{b})$ is such a fact, then $\vec{b}$ lists orphan constants of $\rep_1$, each of which not appearing in $q$. The terms occupying the non-primary-key positions of $S(\underline{\vec{s}},t_{k+1}, \dots ,t_n)$ however cannot be orphan variables of $q$, since this atom is disobedient. It follows that $S(\underline{a},\vec{b})$ is irrelevant for $q$ in $\rep_1$. 

Note that $\rep\subseteq \rep_1 \cup \db$. Hence we obtain by \eqref{eq:emptyset} and $\theta(q)\subseteq \rep$ that 
\begin{equation}\label{eq:indluded}
\theta(q)\subseteq \rep_0\cup \db.
\end{equation} 

Since $\rep_0\not\models q$,  we find a fact $A\in \theta(q)\setminus \rep_0$. By \eqref{eq:emptyset} we obtain that $A\notin \rep_1$, whence $A\in \rep \cap \db$. Moreover, $\rep$ cannot have been defined as $\rep_1$, whence $\rep \closer{\db} \rep_1$. 
 In particular, we have $\rep_0 \cap\db \subseteq \rep_1 \cap\db \subseteq \rep \cap \db$. Since $A\in \rep$ and $\rep\models \pk$, 
this means that $\rep_0 \cap \db$ contains no fact that is key-equal with $A$. Hence, and by \eqref{eq:indluded}, we conclude by Lemma \ref{lem:norepair} that $\rep_0$ is not a repair. This leads to a contradiction, by which we conclude that $\rep\not\models q$. Since $\rep$ is a pre-repair of $\db$ over $\fk$ and $q$, we conclude by Theorem \ref{thm:pre-repair} that there exists a repair of $\db$ over $\fk$ falsifying $q$. This concludes the proof of~\ref{it:dpdB}.
The proof of Lemma~\ref{lem:dpd} is now concluded.
\end{proof}

\subsubsection{Removal of $\dpo$ Foreign Keys}

Finally, we show two lemmas for removing strong foreign keys of type $\dpo$.
Lemma~\ref{lem:dpo} deals with queries $q$ such that $\queryvars{F}\neq\emptyset$ for every $F\in q$.
Lemma~\ref{lem:constantkey} deals with queries containing an atom $F$ with $\queryvars{F}=\emptyset$.

\begin{lemma}[Type $\dpo$]\label{lem:dpo}
Let $q$ be query in $\sjfbcq$, and $\fk$ a set of foreign keys about~$q$.
Following Lemmas~\ref{lem:nonProper}, \ref{lem:opo}, and~\ref{lem:dpd}, assume that all foreign keys in $\fk$ are strong and of type $\dpo$.
Assume the following:
\begin{enumerate}
\item\label{it:dpovarfree}
for every $F\in q$, $\keyvars{F}\neq\emptyset$;
\item\label{it:dponobi}
$(q,\fk)$ has no block-interference; and
\item\label{it:dpoacyclic}
the attack graph of $q$ is acyclic.
\end{enumerate}
Let $\foreignkey{N}{i}{O}$ belong to $\fk$ (and therefore, by our previous assumption, $\fkclosure{O}{q}{\fk}=\{O\}$).
Let $q_{0}=q\setminus\{O\}$ and $\fk_{0}=\fk\setminus\{\foreignkey{N}{i}{O}\}$.
Then, $\fk_{0}$ is about $q_{0}$, and the following hold:
\begin{itemize}
\item
$\certainty{q}{\fk}\reducesTo{\FO}\certainty{q_{0}}{\fk_{0}}$; 
\item
$(q_{0},\fk_{0})$ has no block-interference; and
\item
the attack graph of $q_{0}$ is acyclic.
\end{itemize}
\end{lemma}
\begin{proof}[Proof of Lemma~\ref{lem:dpo}]
That $\fk_0$ is about $q_0$ follows by Lemma \ref{lem:obedience-opo}.
This implies
\begin{equation}\label{eq:dpoio}
\incoming{O}{\fk}=\{\foreignkey{N}{i}{O}\}.
\end{equation}
Since $O$ is obedient and all foreign keys are strong and of type $\dpo$,
it follows $\outgoing{O}{\fk}=\emptyset$.

Concerning foreign keys outgoing $N$, we can write
\begin{equation}\label{eq:unique-out}
\outgoing{N}{\fk}=\{\foreignkey{N}{i}{O_1}, \dots ,\foreignkey{N}{i}{O_m}\}.
\end{equation}
That is, only the $i$th position of $N$ has outgoing foreign keys.
For this,  assume toward contradiction that $\foreignkey{N}{j}{O'}\in \fk$, where $i\neq j$. Denote by $P_N$ the non-primary-key positions of $N$. By the hypothesis of the lemma statement, $\foreignkey{N}{i}{O}\in {\fk}$ is a non-block-interfering strong foreign key of type $\dpo$. 
By the assumptions that $\keyvars{F}\neq\emptyset$ for every $F\in q$, and that $\fk$ is about $q$, the $i$th term of $N$ is a variable. In particular, this variable belongs to
 \begin{equation}\label{eq:emptyV}
 V=\{v\in \queryvars{q'}\mid \FD{q}\not\models\fd{\emptyset}{\{v\}}\}.
 \end{equation}
  We conclude that conditions \eqref{it:obedientO} and \eqref{it:no-constant} of block-interference hold true for $\foreignkey{N}{i}{O}$. But then, conditions \eqref{it:marking} and \eqref{it:bifour} of block-interference must both be false. Recall that \eqref{it:marking} being false means  that $P_N\setminus \{(N,i)\}$ is obedient. Now, the exact same reasoning as above can be repeated for $\foreignkey{N}{j}{O'}$. We conclude that $P_N\setminus \{(N,j)\}$ is likewise obedient. But then Corollary \ref{cor:obedient-position} implies that $P_N$ is obedient, since we assumed that $i\neq j$.
   In particular, $N$ is an obedient atom, contradicting the assumption that $\foreignkey{N}{i}{O}$ is of type $\dpo$. We conclude by contradiction that only the $i$th position of $N$ can have outgoing foreign keys. 

Before proceeding with the proof, let us make some observations about the $N$-atom of $q$. 
 We established that the non-primary-key position $(N,i)$ is occupied by a variable. Since all the remaining non-primary-key positions of $N$ are obedient and do not have outgoing foreign keys, they must be occupied by orphan variables of $q$. We conclude that the $N$-atom is of the form $N(\underline{\vec{t}},y_{k+1}, \dots ,y_n)$ over a sequence of terms $\vec{t}$ and variables $y_{k+1}, \dots ,y_n$, of which those in $\{y_{k+1},\dots ,y_n\}\setminus \{y_i\}$ are orphan. 

Concerning variable $y_i$, let us turn back to condition \eqref{it:bifour} of block-interference. Since this condition is false for the foreign key $\foreignkey{N}{i}{O}$, the variable $y_i$ is not connected  to any variable listed in $\vec{t}$ in $\constgraph{V}{q'}$, for $V$ defined in \eqref{eq:emptyV} and 
\begin{equation}\label{eq:qprime}
q'\defeq q\setminus\{N(\underline{\vec{t}},y_{k+1},\dots,y_{n})\}.
\end{equation}
Observe that $V=\queryvars{q'}$ by the assumption that $\keyvars{F}\neq\emptyset$ for every $F\in q$. 

\textbf{Proof of the third item.}
Easy.

\textbf{Proof of the second item.}
Follows by Lemma \ref{lem:obedience-opo}.

\textbf{Proof of the first item.}
Let $\db$ be a database instance that is input to $\certainty{q}{\fk}$.
We define $\db_{0}$ as the smallest database instance satisfying the following two conditions:
\begin{itemize}
\item
for every relation name $R$ that occurs in $q$ such that $R\notin\{N, O\}$,
$\db_{0}$ contains all $R$-facts of $\db$; and
\item\emph{Relevance restriction:}
for the relation name $N$, $\db_{0}$ includes all (and only) those $N$-blocks of $\db$ that contain at least one fact that is not dangling with respect to $\outgoing{N}{\fk}$.
\end{itemize}
Clearly, $\db_{0}\subseteq\db$. The following claim has an easy proof.
\begin{claim}\label{cla:dporel}
Every repair of $\db$ with respect to foreign keys in $\fk$ and primary keys contains an $N$-fact from every $N$-block of $\db_{0}$.
\end{claim}

It suffices to show the following:
\begin{enumerate}[label=(\Alph*)]
\item\label{it:redyy}
if $\db$ is a ``no''-instance of $\certainty{q}{\fk}$,
then $\db_{0}$ is a ``no''-instance of $\certainty{q\setminus\{O\}}{\fk\setminus\{\foreignkey{N}{i}{O}\}}$; and
\item\label{it:rednn}
if $\db_{0}$ is a ``no''-instance of 
$$\certainty{q\setminus\{O\}}{\fk\setminus\{\foreignkey{N}{i}{O}\}},$$
then $\db$ is a ``no''-instance of $\certainty{q}{\fk}$.
\end{enumerate}

\framebox{Proof of \ref{it:redyy}}
Assume that  $\db$ is a ``no''-instance of the problem  $\certainty{q}{\fk}$.
We can assume a repair~$\rep$ with respect to foreign keys in $\fk$ and primary keys such that $\rep\not\models q$.

We construct $\rep_{0}$ from $\rep$ by applying the following steps:
\begin{description}
\item[Deletion step~1:]
First, delete from $\rep$ all $N$-facts that are not in $\db_{0}$, and delete all $O$-facts.
\item[Deletion step~2:]
Then, for each $\foreignkey{N}{i}{O'}$ in $\outgoing{N}{\fk}$, delete all $O'$-facts of $\rep\setminus\db$ that are no longer referenced by an $N$-fact.
\end{description}
For example, assume a query with atoms atom $N(\underline{x},y)$, $O(\underline{y})$, and $O'(\underline{y})$, together with foreign keys $\foreignkey{N}{2}{O}$ and $\foreignkey{N}{2}{O'}$.
A repair $\rep$ may \emph{(i)}~share, with the input database instance $\db$, the facts $N(\underline{a},2)$ and $O(\underline{2})$, and \emph{(ii)}~invent a fresh fact $O'(\underline{2})$ not in $\db$.
If the former two facts are deleted in the first step, then the latter fact is subject to deletion in the second step. 

Since $\incoming{N}{\fk}=\emptyset$,
it follows that $\rep_{0}$ satisfies foreign keys in $\fk\setminus\{\foreignkey{N}{i}{O}\}$ and primary keys.
By Claim~\ref{cla:dporel}, $\rep_{0}$ contains an $N$-fact from every $N$-block in $\db_{0}$.
By construction, $\rep\cap\db_{0}\subseteq\rep_{0}\subseteq\rep$.

We show
\begin{equation}\label{eq:sepq}
\rep_{0}\not\models q\setminus\{O\}.
\end{equation}
Assume for the sake of contradiction that $\rep_{0}\models q\setminus\{O\}$.
Then there is a valuation $\mu$ over $\queryvars{q\setminus\{O\}}$ such that $\mu(q\setminus\{O\})\subseteq\rep_{0}\subseteq\rep$.
Let $N(a_{1},\dots,a_{n})$ be the $N$-fact of $\mu(q)$.
Since $\rep$ satisfies $\foreignkey{N}{i}{O}$,
it contains a fact of the form $O(\underline{a_{i}},\filler)$.
Since the $O$-atom is obedient, $\mu$ can be extended to a valuation $\mu^{+}$ over $\queryvars{q}$ such that $\mu^{+}(q)\subseteq\rep$, contradicting $\rep\not\models q$.

Let $\rep_{0}^{*}$ be a database instance, consistent with respect to foreign keys in $\fk\setminus\{\foreignkey{N}{i}{O}\}$ and primary keys, such that and $\rep_{0}^{*}\closer{\db_{0}}\rep_{0}$.
That is,
\begin{align}
\rep_{0}\cap\db_{0} & \subseteq\rep_{0}^{*}\label{eq:sandwichA}\\
\rep_{0}^{*} & \subseteq\db_{0}\cup\rep_{0}\subseteq\db\cup\rep\label{eq:sandwichB}
\end{align}

It suffices to show $\rep_{0}^{*}\not\models q\setminus\{O\}$. 
Assume for the sake of contradiction that there is a valuation $\theta$ over $\queryvars{q\setminus\{O\}}$ such that $\theta(q\setminus\{O\})\subseteq\rep_{0}^{*}$.
By~\eqref{eq:sandwichB}, $\theta(q\setminus\{O\})\subseteq\db\cup\rep$.
Since $\rep\models\foreignkey{N}{i}{O}$ and since the $O$-atom is obedient, $\theta$ can be extended to a valuation $\theta^{+}$ over $\queryvars{q}$ such that $\theta^{+}(q)\subseteq\db\cup\rep$.

Since $\rep_{0}\not\models q\setminus\{O\}$, there must be a fact $A\in\theta(q\setminus\{O\})$ such that $A\not\in\rep_{0}$.
By~\eqref{eq:sandwichB}, $A\in\db_{0}$.
 Since, as argued before, $\rep_{0}$ contains an $N$-fact of every $N$-block of $\db_{0}$, applying~\eqref{eq:sandwichA} it can be seen that $A$ cannot be an $N$-fact. We now obtain that $A\in\theta^{+}(q)\setminus\rep$. Moreover, 
if $\rep_{0}$ contains an atom $A'$ such that $A'\keyeq A$, then $A'\notin\db_{0}$ by~\eqref{eq:sandwichA}.
Hence, we also obtain that $\rep\cap\db$ contains no fact that is key-equal to~$A$.

By Lemma~\ref{lem:norepair}, it is now correct to conclude that $\rep$ is not a repair with respect to foreign keys in $\fk$ and primary keys, a contradiction.

\framebox{Proof of \ref{it:rednn}}
Assume that $\db_{0}$ is a ``no''-instance of the problem $\certainty{q\setminus\{O\}}{\fk\setminus\{\foreignkey{N}{i}{O}\}}$.
Among all repairs (with respect to foreign keys in $\fk\setminus\{\foreignkey{N}{i}{O}\}$ and primary keys) of $\db_{0}$ that falsify~$q\setminus\{O\}$ (there is at least one such repair),
let $\rep_{0}$ be one that $\subseteq$-maximizes the set of $N$-facts that are not dangling {in $\db$} with respect to $\outgoing{N}{\fk}$.
Recall that in moving from $\db$ to $\db_{0}$, an $N$-block is removed only if \emph{all} its facts are dangling  {in $\db$} with respect to $\outgoing{N}{\fk}$.
Thus, $\db_{0}$ can contain $N$-facts that are dangling  {in $\db$} with respect to $\outgoing{N}{\fk}$.
It can be easily verified that $\rep_{0}$ will contain a fact from every $N$-block in $\db_{0}$.
The proof now constructs a repair of $\db$, called $\rep$, that falsifies $q$.

We construct $\rep$ from $\rep_{0}$ by applying the following steps:
\begin{description}
\item[Insertion step~1:]
First, insert into $\rep_{0}$ all $O$-facts of $\db$.
Then, chase $\rep_{0}$ with the foreign key $\foreignkey{N}{i}{O}$.
That is, if there is a fact $N(\underline{\vec{a}},b_{k+1},\dots ,b_{n})$ that is dangling with respect to $\foreignkey{N}{i}{O}$,
then insert $O(\underline{b_{i}}, \vec{c})$ for some sequence $\vec{c}$ of fresh constants.
\item[Insertion step~2]
Consider every $N$-block of $\db$ that is not in $\db_{0}$.
If, due to the insertions in the previous step, one fact of such $N$-block is no longer dangling with respect to $\outgoing{N}{\fk}$, then insert a fact from that block.
\end{description}
By construction, $\rep$ is consistent with respect to foreign keys in $\fk$ primary keys. 

\begin{claim}\label{cla:rthree}
$\rep\not\models q$.
\end{claim}
\begin{proof}
Assume towards a contradiction that there is a valuation $\theta$ over $\queryvars{q}$ such that $\theta(q)\subseteq\rep$.
Let $N(\underline{\vec{a}},b_{k+1},\dots,b_{n})$ be the (unique) $N$-fact in $\theta(q)$.
Since $\rep_{0}\not\models q\setminus\{O\}$,
we observe that the fact $N(\underline{\vec{a}},b_{k+1},\dots,b_{n})$ does not belong to $\db_{0}$ and was inserted in \emph{Insertion step~2}.
Thus, every fact in the block $N(\underline{\vec{a}},*)$ is dangling in $\db$ with respect to $\outgoing{N}{\fk}$.
Then, there is a fact  $N(\underline{\vec{c}},p_{k+1},\dots,p_{n})\in\rep_{0}\cap\db_{0}$ such that
\begin{enumerate}
\item
$N(\underline{\vec{c}},p_{k+1},\dots,p_{n})$ is dangling in $\db$ with respect to $\outgoing{N}{\fk}$; and
\item
we have $b_{i}=p_{i}$ (since $\outgoing{N}{\fk}$ is of the form \eqref{eq:unique-out}).
Informally, due to $N(\underline{\vec{c}},p_{k+1},\dots,p_{n})\in\rep_{0}$, we insert, in \emph{Insertion step~1}, the invented fact $O(\underline{p_{i}})$ which in turn entails the insertion, in \emph{Insertion step~2}, of $N(\underline{\vec{a}},b_{k+1},\dots,b_{n})$.
\end{enumerate}
By our choice of $\rep_{0}$, there is a fact $N(\underline{\vec{c}},d_{k+1},\dots,d_{n})$ that is not dangling in $\db$ with respect to $\outgoing{N}{\fk}$, and a valuation $\mu$ over $\queryvars{q}$ such that 
\begin{equation}\label{eq:pivot}
\mu(q\setminus\{O\})\subseteq\formula{\rep_{0}\setminus\{N(\underline{\vec{c}},p_{k+1},\dots,p_{n})\}}\cup\{N(\underline{\vec{c}},d_{k+1},\dots,d_{n})\}.
\end{equation}
Recall that the $N$-atom of $q$ is of the form $N(\underline{\vec{t}},y_{k+1}, \dots ,y_n)$ over a sequence of terms $\vec{t}$ and variables $y_{k+1}, \dots ,y_n$, of which those in $\{y_{k+1},\dots ,y_n\}\setminus \{y_i\}$ are orphan.
We define a valuation $\gamma$ over $\queryvars{q\setminus\{O\}}$ as follows. 
Let $\gamma(y_j)=p_j$ for $j\in \{k+1, \dots ,n\}\setminus\{i\}$. For $u\in \queryvars{q\setminus\{O\}}\setminus (\{y_{k+1},\dots ,y_n\}\setminus \{y_i\})$,  let
$$
\gamma(u)=
\begin{cases}
\theta(u) & \mbox{if $u$ is connected to $y_i$ in $\constgraph{V}{q'}$}\\
\mu(u) & \mbox{otherwise}
\end{cases}
$$
Recall that $V$ and $q'$ have been defined in \eqref{eq:emptyV} and \eqref{eq:qprime}. In particular, $\constgraph{V}{q'}$ is the standard Gaifman graph of $q'$ since $V=\queryvars{q'}$. 

It can now be seen, by $\theta(q)\subseteq \rep$ and \eqref{eq:pivot}, and by the fact that $\rep$ and $\rep_0$ are identical over all relation names except $N$ and $O$, that $\gamma(q\setminus\{O,N\})\subseteq \rep_0$. 

Next we show that $\gamma(N(\underline{\vec{t}},y_{k+1}, \dots ,y_n))=N(\underline{\vec{c}},p_{k+1},\dots,p_{n})$. 
Note first that if any primary-key position of atom $N(\underline{\vec{t}},y_{k+1}, \dots ,y_n)$ is occupied by a constant, the same constant must occupy the same position also in $N(\underline{\vec{c}},p_{k+1},\dots,p_{n})$;  the reason is that  
 $N(\underline{\vec{c}},d_{k+1},\dots,d_{n})$ is a relevant fact by \eqref{eq:pivot} and since $\rep_0\not\models q\setminus\{O\}$.
Analogously we see that if a primary-key position in $N(\underline{\vec{t}},y_{k+1}, \dots ,y_n)$ is occupied by a variable $x$, then the constant $\mu(x)$ must occupy this position in $N(\underline{\vec{c}},p_{k+1},\dots,p_{n})$. Having established earlier that $y_i$ is not connected in $\constgraph{V}{q'}$ to any variable listed in $\vec{t}$, we furthermore obtain  that $\gamma(x)=\mu(x)$. Concerning the non-primary-key positions, we note that $\gamma(y_i)=\theta(y_i)$ since $y_i$ is vacuously connected to itself. Since $N(\underline{\vec{a}},b_{k+1},\dots,b_{n})$ is the (unique) $N$-fact in $\theta(q)$, and since $b_i=p_i$, we obtain that $\gamma(y_i)=p_i$. Moreover, we have specifically imposed that $\gamma(y_j)=p_j$ for $j\in \{k+1, \dots ,n\}\setminus\{i\}$. We conclude from all these remarks that $\gamma(N(\underline{\vec{t}},y_{k+1}, \dots ,y_n))=N(\underline{\vec{c}},p_{k+1},\dots,p_{n})$. See Example~\ref{ex:gamma} for an illustration.

Since $N(\underline{\vec{c}},p_{k+1},\dots,p_{n})\in \rep_0$, we have established that $\gamma(q\setminus\{O\})\subseteq \rep_0$, contradicting our assumption that $\rep_0$ is a repair that falsifies $q\setminus\{O\}$. 
This proves the claim that $\rep\not\models q$.
\end{proof}

\begin{example}\label{ex:gamma}
Assume that $q=\{Y(\underline{y})$, $N(\underline{x},y,u)$, $O(\underline{y})\}$ with foreign key $\foreignkey{N}{2}{O}$.
Note that $u$ is an orphan variable.
Let
$$
\db
=
\begin{array}{ccc}
\begin{array}[t]{c|c}
Y & \underline{y}\bigstrut\\\cline{2-2}
  & d_{1} \\\cdashline{2-2}
\end{array}
&
\begin{array}[t]{c|ccc}
N & \underline{x} & y & u\\\cline{2-4}
  & c & d_{1} & d_{2}\\
  & c & p_{1} & p_{2}\\\cdashline{2-4}
  & a & p_{1} & b_{2}\\\cdashline{2-4}
\end{array}
&
\begin{array}[t]{c|c}
O & \underline{y}\bigstrut\\\cline{2-2}
  & d_{1} \\\cdashline{2-2}
\end{array}
\end{array}
$$
We obtain:
$$
\db_{0}
=
\begin{array}{ccc}
\begin{array}[t]{c|c}
Y & \underline{y}\bigstrut\\\cline{2-2}
  & d_{1} \\\cdashline{2-2}
  \end{array}
&
\begin{array}[t]{c|ccc}
N & \underline{x} & y & u\\\cline{2-4}
  & c & d_{1} & d_{2}\\
  & c & p_{1} & p_{2}\\\cdashline{2-4}
\end{array}
\end{array}
$$
The following repair $\rep_{0}$ is the only repair of $\db_{0}$ falsifying $\{Y(\underline{y})$, $N(\underline{x},y)\}$:
$$
\rep_{0}
=
\begin{array}{ccc}
\begin{array}[t]{c|c}
Y & \underline{y}\bigstrut\\\cline{2-2}
  & d_{1} \\\cdashline{2-2}
\end{array}
&
\begin{array}[t]{c|ccc}
N & \underline{x} & y & u\\\cline{2-4}
  & c & p_{1} & p_{2}\\\cdashline{2-4}
\end{array}
\end{array}
$$
In particular, $\mu=\{x\mapsto c$, $y\mapsto d_{1}$, $u\mapsto d_{2}\}$ would be a satisfying valuation for a repair containing $N(\underline{c},d_{1},d_{2})$.
We obtain $\rep$ from $\rep_{0}$ as follows:
$$
\rep
=
\begin{array}{ccc}
\begin{array}[t]{c|c}
Y & \underline{y}\bigstrut\\\cline{2-2}
  & d_{1} \\\cdashline{2-2}
\end{array}
&
\begin{array}[t]{c|ccc}
N & \underline{x} & y & u\\\cline{2-4}
  & c & p_{1} & p_{2}\\
  & a & p_{1} & b_{2}\\\cdashline{2-4}
\end{array}
&
\begin{array}[t]{c|c}
O & \underline{y}\bigstrut\\\cline{2-2}
  & d_{1} \\\cdashline{2-2}
  & p_{1}\\\cdashline{2-2}
\end{array}
\end{array}
$$
We have that $\rep$ is a repair of $\db$ that falsifies~$q$.
In particular, the fact $N(\underline{c},p_{1},b_{2})$ cannot be used for making the query true.
\qed
\end{example}

To finish the proof, we make the following claim.

\begin{claim}\label{cla:dpoB}
$\rep$ is a repair of $\db$ with respect to foreign keys in $\fk$ and primary keys.
\end{claim}
\begin{proof}[Proof sketch of Claim~\ref{cla:dpoB}.]
By Claim~\ref{cla:dporel}, there exists a repair of $\db$ with respect to foreign keys in $\fk$ and primary keys that contains all $N$-facts of $\rep_{0}$.
The invented $O$-facts inserted in \emph{Insertion step~1} are needed to satisfy $\{\foreignkey{N}{i}{O}\}$.
The $N$-facts inserted in \emph{Insertion step~2} all belong to $\db$.

Assume there is a repair $\sep$ of $\db$ with respect to foreign keys in $\fk$ and primary keys that contains more $N$-facts of $\db$. Then $\sep$ would have to insert also invented fresh $O$-facts, and hence $\sep$ and $\rep$ would not be comparable by $\closer{\db}$.

Finally, note that, since $\incoming{N}{\fk}=\emptyset$, the insertion of $N$-facts does not entail further insertions. 
\end{proof}

The proof of Lemma~\ref{lem:dpo} is now concluded.
\end{proof}

\begin{lemma}\label{lem:constantkey}
Let $q$ be query in $\sjfbcq$, and $\fk$ a set of foreign keys about~$q$.
Following Lemmas~\ref{lem:nonProper}, \ref{lem:opo}, and~\ref{lem:dpd}, assume that all foreign keys in $\fk$ are strong and of type $\dpo$.
Let $N$ be an atom of $q$ such that $\keyvars{N}=\emptyset$, and let $\vec{x}$ be the variables of $\atomvars{N}$.
Let $b$ be an arbitrary constant, and $\vec{b}=\tuple{b,b,\ldots,b}$, a sequence of the same length as~$\vec{x}$.
Let $q_{0}=q\setminus\fkclosure{N}{q}{\fk}$ and $\fk_{0}=\restrict{\fk}{q_{0}}$.
Then $\fk_{0}$ is about $q_{0}$, and the following hold:
\begin{itemize}
\item
$\certainty{q}{\fk}\reducesTo{\FO}\certainty{\substitute{q_{0}}{\vec{x}}{\vec{b}}}{\fk_{0}}$;
\item
if $(q,\fk)$ has no block-interference,
then $(\substitute{q_{0}}{\vec{x}}{\vec{b}},\fk_{0})$ has no block-interference; and
\item
if the attack graph of $q$ is acyclic, then the attack graph of $\substitute{q_{0}}{\vec{x}}{\vec{b}}$ is acyclic.
\end{itemize}
\end{lemma}
\begin{proof}[Proof of Lemma~\ref{lem:constantkey}]
Note that if $\outgoing{N}{\fk}=\emptyset$, then $\fkclosure{N}{q}{\fk}=\{N\}$.

It is clear that $\fk_{0}=\restrict{\fk}{q_{0}}$ is about $q_0$ and thus about $\substitute{q_{0}}{\vec{x}}{\vec{b}}$.

\textbf{Proof of the third item.}
Easy.

\textbf{Proof of the second item.}
Since all the foreign keys are of the form $\dpo$, it follows that $\fkclosure{N}{q}{\fk}=\{N,O_1,\dots ,O_m\}$ such that for all $j\in [m]$, $O_j$ is obedient and there is a foreign key of the form $\foreignkey{N}{i_j}{O_j}\in \fk$. Consider the following claim.

\begin{claim}\label{claim:subP}
 Let $P$ be a set of non-primary-key positions concerning some relation name that appears in $\substitute{q_{0}}{\vec{x}}{\vec{b}}$. Then, $P$ is obedient over $\fk_0$ and $\substitute{q_{0}}{\vec{x}}{\vec{b}}$  if and only if it is obedient over $\fk$ and $q$. 
 \end{claim}
 \begin{proof}[Proof of Claim \ref{claim:subP}]
 Suppose $P$ is obedient over  $\fk_0$ and $\substitute{q_{0}}{\vec{x}}{\vec{b}}$. We claim that $\posclosure{P}{\fk}=\posclosure{P}{\fk_0}$. For the sake contradiction, suppose this is not the case. Then we find a foreign key $\foreignkey{T}{i}{U}\in \fk$ where $(T,i)\in \posclosure{P}{\fk_0}$ and $U\in \fkclosure{N}{q}{\fk}$. We show that this leads to a contradiction. 
 
 Suppose first that $U=N$. In this case, since $\fk$ is about $q$ and $(N,1)$ is occupied by a constant in $q$, $(T,i)$ must be occupied by a constant in $q$. In particular, $(T,i)$ is occupied by a constant in $\substitute{q_{0}}{\vec{x}}{\vec{b}}$, since $\fk_0$ is about $\substitute{q_{0}}{\vec{x}}{\vec{b}}$. This contradicts the assumption that $P$ is obedient over $\fk_0$ and $q_0$.
 
Suppose then that $U=O_j$ for some $j\in [m]$. For the same reason as above, $(O_j,1)$ cannot be occupied by a constant in $q$. But then, it is occupied in $q$ by some variable $x$ listed in $\vec{x}$. Since $\fk$ is about $q$, also $(T,i)$ is occupied by $x$ in $q$. In particular, $(T,i)$ is occupied by $x$ in $q_0$. 
 However, due to obedience of $P$ over $q_0$ and $\fk_0$, no position of $\posclosure{P}{\fk_0}$ can be occupied by the constant $b$ in $\substitute{q_{0}}{\vec{x}}{\vec{b}}$. This implies that $(T,i)$ in particular cannot be occupied by  $x$ in $q_0$. Thus we obtain a contradiction. 
 
 We conclude by contradiction that $\posclosure{P}{\fk}=\posclosure{P}{\fk_0}$. Since $P$ is obedient over $\fk_0$ and $\substitute{q_{0}}{\vec{x}}{\vec{b}}$, using Theorem \ref{thm:syntactic-obedience} it suffices to show that there is no variable occuring at a position of $\posclosure{P}{\fk}$ and at a position of $\poscomplement{P}{\fk}$ in $q$. In particular, we may restrict attention to those positions which appear in $q$ but not in $q_0$. That is, we only consider those variables that occur at an
  $N$-position or an $O_j$-position, for $j\in [m]$. Note that these positions belong to $\poscomplement{P}{\fk}$, which means that these variables must not occur at positions of $\posclosure{P}{\fk}$ in $q$. First, observe that $\vec{x}$ lists $\atomvars{N}\cup\bigcup_{j=1}^m\keyvars{O_j}$ since $\fk$ is about $q$.  Clearly, no variable from $\vec{x}$ may appear at a position of $\posclosure{P}{\fk}$ in $q$, because then the constant $b$ would appear at a position of $\posclosure{P}{\fk_0}$ in $\substitute{q_{0}}{\vec{x}}{\vec{b}}$, contradicting obedience of $P$ over $\substitute{q_{0}}{\vec{x}}{\vec{b}}$ and $\fk_0$. Any remaining variable $y$ occurs in $q$ at a non-primary-key position of some $O_j$. Since there are no foreign keys outgoing obedient $O_j$, such a variable $y$ must be orphan in $q$. In particular, $y$ cannot occur at  a position of $\posclosure{P}{\fk}$ in $q$. This establishes that $P$ is obedient over $\fk$ and $q$.
 
 For the other direction, suppose $P$ is obedient over $\fk$ and $q$. The primary key positions of $N$ must belong to $\poscomplement{P}{\fk}$ since they are occupied by constants in $q$. Since the relation name  associated with $P$ is distinct from $N$, the non-primary-key position of $N$ must likewise belong to $\poscomplement{P}{\fk}$. 
 Using Theorem~\ref{thm:syntactic-obedience}~\ref{it:disjoint}, we observe that the variables listed in $\vec{x}$ cannot occur in $q$ at any position of $\posclosure{P}{\fk}$. Since $\posclosure{P}{\fk_0}\subseteq \posclosure{P}{\fk}$, the same holds with respect to $q_0$ and $\posclosure{P}{\fk_0}$. In particular, the substitution of $\vec{b}$ for $\vec{x}$ in $q_0$ does not introduce constants at positions of $\posclosure{P}{\fk_0}$. It is now straightforward to verify that all the conditions of Theorem \ref{thm:syntactic-obedience} are preserved under $\substitute{q_{0}}{\vec{x}}{\vec{b}}$ and $\fk_0$. We conclude that $P$ is obedient over $\fk_0$ and $\substitute{q_{0}}{\vec{x}}{\vec{b}}$. This concludes the proof of the claim.
 \end{proof}
 
 We are now ready to prove that $(\substitute{q_{0}}{\vec{x}}{\vec{b}},\fk_0)$ has no block-interference, assuming $(q,\fk)$ has no block-interference.
Toward contradiction, assume that this is not the case. 
Then, $\lclosure{\fk_0}$ contains a block-interfering (with respect to $(\substitute{q_{0}}{\vec{x}}{\vec{b}},\fk_0)$)  foreign key $\sigma\defeq \foreignkey{N}{j}{O}$ for some atoms 
 $N(\underline{t_{1}, \dots ,t_k},t_{k+1},\dots,t_{n})$ and  $O(\underline{t_j},\vec{y})$ of $\substitute{q_{0}}{\vec{x}}{\vec{b}}$.

Claim \ref{claim:subP} now entails that, with respect to $(q,\fk)$, $\sigma$ satisfies 
Definition~\ref{def:block-int}~\eqref{it:obedientO}, 
and also Definition~\ref{def:block-int}~\eqref{it:marking} if this holds with respect to $(\substitute{q_{0}}{\vec{x}}{\vec{b}},\fk_0)$. 
It remains to show that $\sigma$ satisfies  Definition~\ref{def:block-int}~\eqref{it:no-constant}, and also  Definition~\ref{def:block-int}~\eqref{it:bifour} if this holds with respect to $(\substitute{q_{0}}{\vec{x}}{\vec{b}},\fk_0)$. 
Let us denote the terms $t_i$ and $t_j$ of the remaining conditions by $u$ and $z$, respectively.

Concerning Definition~\ref{def:block-int}~\eqref{it:no-constant}, we claim that $z$ belongs to
\begin{equation*}
V=\{v\in \queryvars{q}\mid \FD{q}\not\models\fd{\emptyset}{\{v\}}\}.
\end{equation*}
 By the assumption that $(q,\fk)$ has no block-interference, we obtain that $z$ belongs to 
 \begin{equation*}
 V_0=\{v\in \queryvars{\substitute{q_{0}}{\vec{x}}{\vec{b}}}\mid \FD{\substitute{q_{0}}{\vec{x}}{\vec{b}}}\not\models\fd{\emptyset}{\{v\}}\}. 
\end{equation*}

To establish the claim, we can use similar reasoning as in the proof of Lemma \ref{lem:obedience-opo}. Indeed, assume toward contradiction that $z\in V_0\setminus V$. Then  $\FD{q}\models\fd{\emptyset}{\{z\}}$ and $\FD{\substitute{q_{0}}{\vec{x}}{\vec{b}}}\not\models\fd{\emptyset}{\{z\}}$. It follows that any proof of $\FD{q}\models\fd{\emptyset}{\{z\}}$ must use either $N$ or $O_j$, for some $j\in [m]$. The latter option can be discarded since the non-primary-key positions of $O_j$ are occupied by orphan variables which must be distinct from $z$. For the former option, since $\keyvars{N}=\emptyset$ any (shortest) proof which uses $N$ must be of the form $(F_1,\dots,F_m)$, where 
 $z\in \atomvars{F_n}$, $F_1=N$, and $\keyvars{F_{i+1}}\subseteq \bigcup_{j=1}^i \atomvars{F_j}$ for $i\in [m-1]$. Because $\vec{x}$ lists $\atomvars{N}$, such a proof can be simulated in $\substitute{q_{0}}{\vec{x}}{\vec{b}}$ by $(\substitute{F_1}{\vec{x}}{\vec{b}}, \dots, \substitute{F_m}{\vec{x}}{\vec{b}})$, in which case $z\notin V_0$. This leads to a contradiction, showing that $\sigma $ satisfies Definition~\ref{def:block-int}~\eqref{it:no-constant} with respect to $(q,\fk)$.

Furthermore, as in Lemma \ref{lem:obedience-opo}, 
this implies that $\sigma$ satisfies Definition~\ref{def:block-int}~\eqref{it:bifour} with respect to $(q,\fk)$ if the same holds with respect to $(\substitute{q_{0}}{\vec{x}}{\vec{b}},\fk_0)$.
We omit further details for now.

Again, we obtain that $\sigma$ is block-interfering in $(q,\fk)$, which contradicts our assumption. We conclude by contradiction that $(\substitute{q_{0}}{\vec{x}}{\vec{b}},\fk_0)$ has no block-interference, if $(q,\fk)$ has no block-interference.

\textbf{Proof of the first item.}
Assume that the $N$-atom of $q$ is $N(\underline{\vec{c}},\vec{t})$.
Let $\db$ be a database instance that is input to $\certainty{q}{\fk}$.
If $\db$ contains no fact of the form $N(\underline{\vec{c}},\filler)$,
then $\db$ is obviously a ``no''-instance.
Otherwise, if every fact in the block $N(\underline{\vec{c}},*)$ is dangling with respect to $\outgoing{N}{\fk}$ (which can be tested in $\FO$), then $\db$ is also a ``no''-instance.
Indeed, in this case, there exists a repair that contains no fact of the form $N(\underline{\vec{c}},\filler)$.

Assume from here on that $\db$ contains an $N$-fact $N(\underline{\vec{c}},\vec{b})$ that is not dangling with respect to $\outgoing{N}{\fk}$.
It can be seen that every repair of $\db$ with respect to foreign keys in $\fk$ and primary keys will contain a fact from $\theblock{N(\underline{\vec{c}},\vec{b})}{\db}$.
Moreover, every fact of this block belongs to some repair.
It can now be seen that the following are equivalent:
\begin{itemize}
\item
$\db$ is a ``yes''-instance of $\certainty{q}{\fk}$; and
\item
for every fact  $N(\underline{\vec{c}},\vec{d})$ in $\db$, there is a valuation $\theta$ over $\sequencevars{\vec{x}}$ such that $\theta(\vec{t})=\vec{d}$ and $\db$ is a ``yes''-instance of \\
$\certainty{\substitute{q_{0}}{\vec{x}}{\theta(\vec{x})}}{\fk_{0}}$.
\end{itemize}
The latter test is in $\FO$.
Finally, since $q$ is self-join-free, we have that 
$$\certainty{\substitute{q_{0}}{\vec{x}}{\theta(\vec{x})}}{\fk_{0}}
\reducesTo{\FO}\certainty{\substitute{q_{0}}{\vec{x}}{\vec{b}}}{\fk_{0}}.$$
We argue that we can always reduce to a problem in which $b$ is the only constant used in a query.
In the following, a \emph{term} is a variable or a constant.
For every atom $R(s_{1},\dots,s_{n})$ in $\substitute{q_{0}}{\vec{x}}{\theta(\vec{x})}$, 
we replace every atom $R(a_{1},\dots,a_{n})$ in $\db$,
with $R(f(a_{1},s_{1}),\dots,f(a_{n},s_{n}))$,
where $f$ maps each pair of a constant and a term to a constant,
such that $f(c_{1},t_{1})=f(c_{2},t_{2})$ if and only if $c_{1}=c_{2}$ and $t_{1}=t_{2}$ (i.e., $f$ is injective) and for every constant $c$, we define $f(c,c)\defeq b$. 
The latter reduction is correct for conjunctive queries that are self-join-free.
For example, if we replace $R(\underline{x},c)$ with $R(\underline{x},b)$ in a query, then we replace every fact $R(\underline{s},a)$ with $R(\underline{s},f(a,c))$, where $f(a,c)$ can be seen as a fresh constant, and we replace $R(\underline{s},c)$ with $R(\underline{s},b)$.
This concludes the proof of Lemma~\ref{lem:constantkey}.
\end{proof}

\subsection{Proof of Lemma~\ref{lem:inFO}}

We can now give a proof of Lemma~\ref{lem:inFO}.

\begin{proof}[Proof of Lemma~\ref{lem:inFO}]
Assume that the attack graph of $q$ is acyclic and $(q,\fk)$ has no block-interference.
We first repeatedly apply the reduction of Lemma~\ref{lem:nonProper} to remove all weak foreign keys.
Then we apply the reductions of Lemmas~\ref{lem:opo} and~\ref{lem:dpd} to remove strong foreign keys of a type in $\{\opo, \dpd\}$.
Whenever the resulting query contains an atom $F$ such that $\keyvars{F}=\emptyset$, we apply Lemma~\ref{lem:constantkey}.
Whenever every atom in the resulting query has a variable at some primary-key position, we apply Lemma~\ref{lem:dpo}.
Eventually, we have reduced to some problem $\certainty{q''}{\fk''}$ with $\fk''=\emptyset$, such that the attack graph of $q''$ is acyclic.
The latter problem is known to be in $\FO$.
The desired result holds by induction on the number of reductions, since for every intermediate problem $\certainty{q'}{\fk'}$, it holds that the attack graph of $q'$ is acyclic and $(q',\fk')$ has no block-interference.
\end{proof} 

\end{document}